\documentclass[10pt]{amsart}

\usepackage{amsthm,amsmath,amssymb}
\usepackage{mathtools,color}
\usepackage{graphicx}
\usepackage[comma, sort&compress]{natbib}
\usepackage{tabularx}
\usepackage{multirow}
\usepackage{amssymb}
\usepackage[utf8]{inputenc}
\usepackage{setspace}
\usepackage{subfig}
\usepackage{leftidx}
\usepackage{mathabx}
\usepackage[linesnumbered, ruled, noline]{algorithm2e}

\usepackage{paralist}
\RequirePackage{hyperref}

\newtheorem{theorem}{Theorem}
\newtheorem{assumption}{Assumption}

\newtheorem{lemma}[theorem]{Lemma}

\theoremstyle{definition}
\newtheorem{definition}{Definition}
\theoremstyle{remark}

\newtheorem{claim*}{Claim}

\numberwithin{equation}{section}
\numberwithin{theorem}{section}
\numberwithin{example}{section}
\numberwithin{definition}{section}

\numberwithin{figure}{section}

\newcommand{\boldB}{\boldsymbol B}

\newcommand{\assumpref}[1]{Assumption~\ref{assump:#1}}
\newcommand{\assumpsref}[1]{Assumptions~\ref{assump:#1}}
\newcommand{\assumpssref}[1]{\ref{assump:#1}}
\newcommand{\figref}[1]{Figure~\ref{fig:#1}}
\newcommand{\figsref}[1]{Figures~\ref{fig:#1}}
\newcommand{\figssref}[1]{\ref{fig:#1}}

\newcommand{\secref}[1]{Section~\ref{sec:#1}}

\newcommand{\algref}[1]{Algorithm~\ref{alg:#1}}

\newcommand{\appref}[1]{Appendix~\ref{app:#1}}
\newcommand{\appsref}[1]{Appendices~\ref{app:#1}}
\newcommand{\appssref}[1]{\ref{app:#1}}

\newcommand{\defref}[1]{Definition~\ref{def:#1}}
\newcommand{\defsref}[1]{Definitions~\ref{def:#1}}
\newcommand{\defssref}[1]{\ref{def:#1}}
\newcommand{\exref}[1]{Example~\ref{ex:#1}}

\newcommand{\lemref}[1]{Lemma~\ref{lem:#1}}

\newcommand{\thmref}[1]{Theorem~\ref{thm:#1}}

\newcommand{\tabref}[1]{Table~\ref{tab:#1}}
\newcommand{\tabsref}[1]{Tables~\ref{tab:#1}}
\newcommand{\tabssref}[1]{\ref{tab:#1}}

\title[]{ZAP: $Z$-value Adaptive Procedures for False Discovery Rate Control with Side Information}

\author[D.~Leung]{Dennis Leung}
\address{School of Mathematics and Statistics, University of Melbourne}
\email{dennis.leung@unimelb.edu.au}

\author[W.~Sun]{Wenguang Sun}
\address{Center for Data Science, Zhejiang University}
\email{wgsun@zju.edu.cn}

\begin{document}

\begin{abstract}
Adaptive multiple testing with covariates is an important research direction that has gained major attention in recent years. It has been widely recognized that leveraging side information provided by auxiliary covariates  can improve  the power of false discovery rate (FDR) procedures. Currently, most such  procedures are devised with $p$-values as their main statistics. However, for two-sided hypotheses,  the usual data processing step that transforms the primary statistics, known as $z$-values, into $p$-values not only leads to a loss of information carried by the main statistics, but can also undermine the ability of the covariates to assist with the FDR inference. We develop a $z$-value based covariate-adaptive (ZAP) methodology  that operates on the intact structural information  encoded jointly by the $z$-values and covariates. It seeks to emulate the oracle $z$-value procedure via a working model, and its rejection regions significantly depart from those of the $p$-value adaptive testing approaches. The key strength of ZAP is that the FDR control is guaranteed  with minimal assumptions, even when the working model is misspecified.  We demonstrate the state-of-the-art performance of ZAP using both simulated and real data, which shows that the efficiency gain  can be substantial in comparison with p-value based methods. Our methodology is implemented in the R package \texttt{zap}.

\end{abstract}



\maketitle

\section{Introduction} \label{sec:intro}

In modern scientific studies, a ubiquitous  task is to test a multitude of two-sided hypotheses regarding the presence of nonzero effects. The problem of multiple testing with covariates has received much recent attention, as  leveraging contextual information beyond what is offered by the main statistics can enhance both the power and interpretability of existing false discovery rate (FDR; \citealp{benjamini1995controlling}) methods. 
This has marked a gradual paradigm shift from the Benjamini-Hochberg (BH) procedure and its immediate variants (e.g. \citealp{benjamini2000adaptive, storey2002direct}) that are based solely on the $p$-values. 
For instance, in the differential analysis of RNA-sequencing data, 
 the average read depths across samples can provide useful side information alongside  individual $p$-values, and incorporating such  information promises to improve the efficiency of existing methods. The importance of this direction has been reflected by its intense research activities; see \cite{lei2016adapt, zhang2020covariate, yurko2020selective, chen2017functional, ignatiadis2016data, li2019multiple, boca2018direct} for an incomplete list of related works. 
In contrast with BH and its variants that apply a universal threshold to all $p$-values, these methods boil down to setting varied $p$-value thresholds that are adaptive to the covariate information.  

This seemingly natural modus operandi, which involves using $p$-values as the basic building blocks, however, is suboptimal for two-sided testing; the $p$-values involved are typically formed via a data reduction step, which applies a non-bijective transformation to ``primary" test statistics such as the $z$-values, $t$-statistics \citep{ritchie2015limma} or Wald statistics \citep{love2014moderated}. \cite{sun2007oracle} and \cite{storey2007optimal} argued that reducing $z$-values to two-sided $p$-values may lead to substantial loss of information, particularly when the $z$-values exhibit distributional asymmetry. The main thrust of this article is to reveal a new source of information loss in the context of covariate-adaptive multiple testing, and to develop a $\bf z$-value cov{\bf a}riate-ada{\bf p}tive  methodology, which we call ``ZAP" for short, that bypasses the data reduction step. As illustrated in  \secref{infoLoss}, the interactive relationship between the $z$-values and the covariates can capture structural information that can be exploited for more testing power. However, this interactive information may be undercut, and in some scenarios, completely forgone when converting the $z$-values to $p$-values. Hence, the data reduction step not only leads to a loss of information carried by the main statistics, but also undermines the ability of the covariates to assist with the FDR inference. 


Few works on covariate-adaptive testing have pursued the $z$-value direction for two-sided testing since combining the $z$-values and covariates poses an additional layer of challenges. Existing $z$-value based procedures either make strong assumptions on the data generating model \citep{scott2015false}, or are not robust for handling multi-dimensional covariate data \citep{cai2019covariate}. By contrast, ZAP retains the merits of $z$-value based methods and avoids the information loss from ``collapsing  into $p$-values", without relying on strong assumptions nor forgoing robustness. It faithfully preserves the interactive structure between the primary statistics and covariates as a starting point for inference, and is deployed with a working model, whose potential misspecification will not invalidate the FDR control. 

Our contribution is twofold. First, ZAP represents a $z$-value based, covariate-adaptive testing framework that attains state-of-the-art power performance under minimal assumptions, filling an important gap in the literature. Second, in light of a plethora of $p$-value based covariate-adaptive methods that have emerged in recent years, our study explicates new sources of information loss in data processing, which provides new insights and gives caveats for conducting covariate-adaptive inference in practical settings.  

The rest of the paper is structured as follows. \secref{ZAPframework} states the problem formulation and describes the high-level ideas of ZAP. \secref{zapMethods} formally introduces our two data-driven methods of ZAP and their implementation details. Numerical results based on both simulated and real data are presented in  \secref{numeric}. \secref{discuss} concludes the article with a discussion of open issues.

\section{Problem Formulation and Basic Framework}\label{sec:ZAPframework}


\subsection{The problem statement} \label{sec:problemState}

Suppose we are interested in making inference of $m$ real-valued effects $\mu_i$, $i =1, \dots, m$, and for each $i$, we observe
a primary statistic $Z_i \in \mathbb{R}$ (``$z$-value") and  an auxiliary covariate  $X_i \in \mathbb{R}^p$ that can be multivariate. We consider a multiple testing problem where the goal is to identify nonzero effects or, equivalently, determine the values of the indicators
\begin{equation} \label{hypotheses_indic}
H_i  \equiv I(\mu_i \neq  0) = \begin{dcases*}
        1  & if $\mu_i \neq  0$\\
        0 & otherwise
        \end{dcases*}.
\end{equation}
Assume that the triples $\{H_i, Z_i, X_i\}_{i =1}^m$ are independent and identically distributed, and the data are described by the following mixture model:
\begin{equation} \label{mixtureDens}
 Z_i | X_i=x \sim  f_{x}(z)   \equiv  f(z |x)= (1 - w_x) f_0 (z) + w_x f_{1,x} (z), 
\end{equation}
where $w_{x} \equiv P(H_i = 1|X_i=x)$ is the conditional probability of having a non-zero effect given $X_i=x$ and $f_{1,x}(z)\equiv f(z| H_i=1, X_i=x)$ is the conditional density under the alternative. $f_0$ denotes the null density, which is invariant to the covariate value. In this article we assume $f_0(z)\equiv \phi(z)$, the density of a $\mathcal N(0,1)$ variable\footnote{This can be easily achieved via the composite transformation $\Phi^{-1}\circ G_0(\cdot)$ if the primary test statistic has a known null distribution function $G_0(\cdot)$, e.g. a t-distribution, where $\Phi(\cdot)$ is the standard normal distribution function.}.  In contrast with \cite {scott2015false} which assumes a fixed alternative density, i.e. $f_{1,x}\equiv f_1$, the data generating model in \eqref{mixtureDens} provides a more general framework for multiple testing with covariates by allowing both $w_x$ and $f_{1,x} (z)$ to vary in $x$.   

Let $\mathcal R \subset \{1, \dots, m\}$ be the set of hypotheses rejected by a multiple testing procedure. In large-scale testing problems, the widely used FDR is defined as 
$$
\mbox{FDR} = \mathbb{E}\left[\frac{V} {R \vee 1}\right],
$$
where $V  = \sum_{i = 1}^m (1 - H_i) I(i \in \mathcal{R}) \text{ and } R = \sum_{i = 1}^m  I(i \in \mathcal{R})$
are respectively the number of false positives and the number of rejections. 
Throughout, $\mathbb{E}[\cdot]$ denotes an expectation operator with respect to the joint distribution of $\{H_i, Z_i, X_i\}_{i = 1}^m$, and $\mathbb{E}[\cdot|\cdot]$ denotes a conditional expectation that should be self-explanatory from its context. The ratio ${V}/(R \vee 1)$ is known as the false discovery proportion (FDP). The power of a testing procedure can be evaluated using the expected number of true discoveries $\mbox{ETD}=\mathbb{E}[R - V]$ or the true positive rate 
$$
\mbox{TPR}=\mathbb{E}\left[\frac{R - V}{(\sum_{i = 1}^m H_i)\vee 1} \right].
$$ 
Our goal is to devise a powerful procedure that can control
 the FDR under a pre-specified level $\alpha \in (0,1)$.

\subsection{Information loss in covariate-adaptive testing} \label{sec:infoLoss}

A two-sided $p$-value is formed by the non-bijective transformation $P_i=2\Phi(-|Z_i|)$, where $\Phi$ is the cumulative distribution function of a $\mathcal N(0,1)$ variable. We call a testing procedure \emph{z-value based} if it makes rejection decisions based on the full dataset $\{Z_i, X_i\}_{i =1}^m$, and \emph{p-value based} if it does so only based on the reduced dataset $\{P_i, X_i\}_{i=1}^m$. This section presents examples to illustrate that the interactive structure between $Z_i$ and $X_i$ is generally not preserved by transforming into $p$-values; the associated information loss can lead to decreased power in the FDR inference. 

Consider Model \eqref{mixtureDens}, and suppose $X_i\sim \mbox{Unif}(-1,1)$. Our study examines three situations: 
\begin{enumerate}
  \item [Example 2.1] Asymmetric alternatives: $f(z|x) = \frac{8 - x }{10}f_0(z)  + \frac{ x+ 2 }{10}\phi(z-1.5)$. \label{ex:S1}
  \item [Example 2.2] Unbalanced covariate effects on the non-null proportions: 
  \[f(z|x) = 0.8 f_0(z)  +
 \frac{1-x}{10}\phi(z+1.5)+\frac{1+x}{10}\phi(z-1.5).\] \label{ex:S2}
  \item [Example 2.3] Unbalanced covariate effects on the alternative means: \[f(z|x) = 0.9 f_0(z) + 0.1\phi(z-1.5 \;  \text{sgn}(x)), \; \text{where}\; \text{sgn}(x)=I(x\geq0)-I(x < 0). \]  \label{ex:S3}
   \end{enumerate}

We investigate two approaches to FDR analysis for these examples that respectively reject hypotheses with suitably small posterior probabilities $\{P(H_i=0|P_i, X_i)\}_{i=1}^m$ and $\{P(H_i=0|Z_i, X_i)\}_{i =1}^m$. The latter probabilities are assumed to be known by an oracle\footnote{In practice these posterior probabilities are unknown.}. In the literature the $z$-value based quantity $P(H_i=0|Z_i, X_i)$ is also called the conditional local false discovery rate (CLfdr, \citealp{Efr08, CaiSun09}). 
It is known that the \emph{optimal} $p$-value and $z$-value based procedures, which maximize true discoveries subject to false discovery constraints, have the respective forms
\begin{equation*}
\pmb\delta^\mathcal{P}  =  \bigl\{I[P(H_i=0|P_i, X_i)\leq t_\mathcal{P} ]\bigr\}_{i =1}^m \text{ and }\pmb\delta^\mathcal{Z}  =  \bigl\{I[P(H_i=0|Z_i, X_i)\leq t_\mathcal{Z}]\bigr\}_{i =1}^m,
\end{equation*}
where the rejection decisions are  expressed by indicators, and the thresholds $t_\mathcal{P}$ and $t_\mathcal{Z}$ are calibrated such that the nominal FDR level is exactly $\alpha$; see \appref{oracleOptim} for a review.
 In our comparisons we choose suitable thresholds such that the FDR of both methods is exactly $0.1$, and their powers are reported as the TPR empirically computed by 150 repeated experiments for $m = 1000$:
\begin{itemize}
\item[] Example 2.1: $\mbox{TPR}_{\pmb\delta^\mathcal{P} } = 4.4\%$; $\mbox{TPR}_{\pmb\delta^\mathcal{Z} } = 11.7\%$.
\item[] Example 2.2: $\mbox{TPR}_{\pmb\delta^\mathcal{P} } = 3.4\%$; $\mbox{TPR}_{\pmb\delta^\mathcal{Z} } = 5.5\%$.
\item[] Example 2.3: $\mbox{TPR}_{\pmb\delta^\mathcal{P} } = 0.6 \%$; $\mbox{TPR}_{\pmb\delta^\mathcal{Z} } = 2.6 \%$.
\end{itemize}
Apparently, ${\pmb\delta^\mathcal{Z}}$ is more powerful  than ${\pmb\delta^\mathcal{P}}$. 

To understand the differences in power, we first remark that either  oracle procedure essentially amounts to one by which $i$ is rejected if and only if
\begin{equation} \label{oracle_rej_region}
Z_i \in \mathcal{S}(X_i) \subset \mathbb{R}
\end{equation}
for some rejection region $\mathcal{S}(\cdot)$ on the $z$-value scale that is a function of the covariate value; the theoretical derivation  is sketched in  \appref{rejRegion}. 
Let $\mathcal{S}^\mathcal{P}(x)$ and $\mathcal{S}^\mathcal{Z}(x)$ denote the respective rejection regions of $\pmb\delta^\mathcal{P}$ and $\pmb\delta^\mathcal{Z}$ on the $z$-value scale for a given covariate value $x$, which are plotted for the three examples in \figref{oracle_rej_reg}.
On the left panel, both $\mathcal{S}^\mathcal{P}(x)$ and $\mathcal{S}^\mathcal{Z}(x)$ enlarge as the covariate value increases, suggesting that the covariates are informative for both methods. The information loss leading to the lesser power of $\pmb\delta^\mathcal{P}$ in Example 2.1 is intrinsically within the main statistics when converting $z$-values to $p$-values (\citealp{sun2007oracle,storey2007optimal}).
By contrast, the middle and right panels show that $\mathcal{S}^\mathcal{Z}(x)$ changes with $x$, while $\mathcal{S}^\mathcal{P}(x)$ is completely insensitive to the changes in $x$; see  \appref{rejRegion}  for the relevant calculations. Hence, the covariates are only informative for $\pmb\delta^\mathcal{Z}$. This fundamental phenomenon reveals that upon reduction to $p$-values, the information loss not only can occur internally within the main statistics, but also externally due to the failure of $\pmb\delta^\mathcal{P}$ in fully capturing the original interactive information between $Z_i$ and $X_i$. When the latter interactive structure represents the bulk of the information provided by the covariates for testing, reduction to $p$-values can substantially undermine the covariates' ability to assist with inference.

\begin{figure}[t]
\centering
\includegraphics[width=\textwidth]{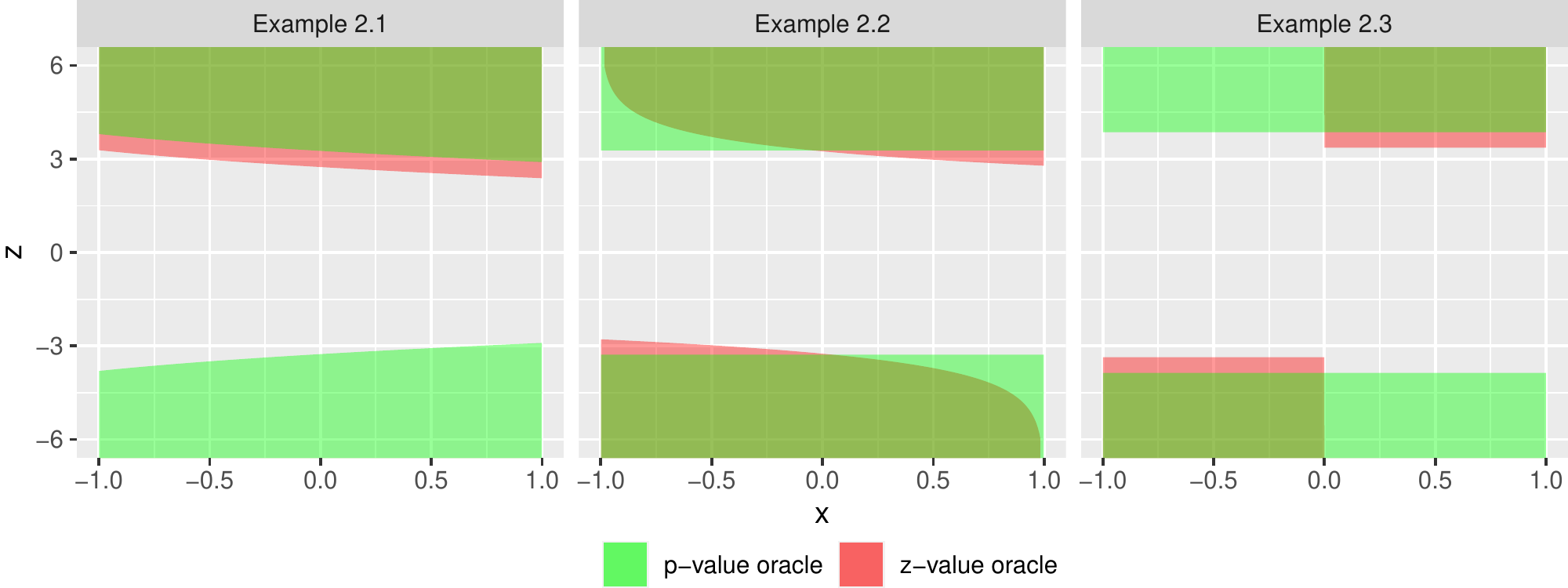}
\caption{ 
Comparisons of rejection regions. For each covariate value $x \in (-1, 1)$, the rejection regions  $\mathcal{S}^\mathcal{P}(x)$ and $\mathcal{S}^\mathcal{Z}(x)$  are respectively marked in green and red on the $z$-value scale. The overlapped region is depicted in yellow.}
\label{fig:oracle_rej_reg}
\end{figure}

\subsection{The ZAP framework and a preview of contributions} \label{sec:ProposalZAP}
  
The previous examples  motivate us to focus on $z$-value adaptive procedures to avoid information loss. This naturally boils down to pursuing the oracle procedure $\pmb\delta^\mathcal{Z}$ in some shape or form, which presents unique challenges. Existing $z$-value based works such as \cite{scott2015false} and \cite{cai2019covariate} are built directly upon $\{P(H_i=0|Z_i, X_i)\}_{i=1}^m$, the CLfdr statistics, which unfortunately involves unknown quantities that can be difficult to estimate in the presence of covariates. Commonly used algorithms may not produce desired estimates, and even lead to invalid FDR procedures if the modelling assumptions are violated. That the theory on FDR control critically depends on the quality of these estimates has limited the scope and applicability of these works. 

We aim to develop a new class of $z$-value adaptive (ZAP) procedures  that are assumption-lean, robust and capable of effectively exploiting the interactive information between $Z_i$ and $X_i$. 
The key idea is to emulate the oracle procedure $\pmb\delta^\mathcal{Z}$ while circumventing the direct estimation of $P(H_i=0|Z_i, X_i)$. Next we first outline the key steps (ranking and thresholding) of our framework and then provide a preview of its contributions. 

In the first \emph{ranking} step, we introduce the new concept of \emph{assessor functions}, which can be estimated from the data based on a {working} model, to construct a new sequence of significance indices $\{T_i\}_{i=1}^m$ as proxies for $\{P(H_i=0|Z_i, X_i)\}_{i=1}^m$. While many potential working models can be used, in this work we focus on a class of beta-mixture models that are carefully defined on a bijective transformation of the $z$-values and particularly suitable for two-sided testing (\secref{betaMod}). In the second \emph{thresholding} step, ZAP calibrates a threshold along the ranking produced by $\{T_i\}_{i=1}^m$. The essential idea is to count the
number of false rejections by any candidate threshold value with  a ``mirroring" sequence of  the rejected significance indices, which can be created via either simulation (\algref{zapAsymp}) or partial  data masking (\algref{zapFinite}). 
 The key strength of ZAP over the methods in \cite{scott2015false} and \cite{cai2019covariate} is that it seeks to emulate the oracle $z$-value procedure while avoiding a direct substitution of $P(H_i|Z_i, X_i)$ with its estimate. ZAP is assumption-lean and robust in the sense that it is provably valid for FDR control under misspecifications of the working model. We stress that the resulting rejection regions of ZAP significantly depart from those of  $p$-value adaptive methods, including the closely related CAMT \citep{zhang2020covariate} and AdaPT \citep{lei2016adapt}. Our simulation and real data studies show that the resulting gain in testing power can be substantial. 

\section{Data-Driven ZAP Procedures}\label{sec:zapMethods}

This section develops the framework of ZAP and its data-driven algorithms for covariate-adaptive FDR inference.  \secref{assessor} introduces the concept of an assessor function and a prototype procedure inspired by the oracle $z$-value procedure. The assessor function can be constructed based on a working beta-mixture model, which is proposed in Section \ref{sec:betaMod}. Sections \ref{sec:ZAPasymp} and \ref{sec:ZAPfinite} lay out two variants of data-driven ZAP procedures and establish their theoretical properties. Further implementation details are discussed in  \secref{beta_expr}.


\subsection{Preliminaries: oracle $z$-value procedure, assessor function and a prototype ZAP algorithm}\label{sec:assessor}

To facilitate the development of a working model, we consider the following lossless transformation: $U_i = \Phi(Z_i)$. The transformed statistic $U_i$ is referred to as a \emph{$u$-value}\footnote{Despite the similarity in their constructions, the $u$-values should \emph{not} be treated as $p$-values for one-sided tests, which are not the subject matter of this work.}, which, according to \eqref{mixtureDens}, obeys the induced mixture model
  \begin{equation} \label{U_cond_dens}
   U_i| X_i=x \sim h_{x} (u) \equiv h(u|x)  = (1 - w_x) h_0(u)+ w_x h_{1,x} (u),
  \end{equation}
with $h_0(u)$ and $h_{1,x} (u) \equiv h(u | H_i =1, X_i = x)$ respectively being the null $\mbox{Unif}(0,1)$ and conditional alternative densities. An optimal FDR procedure (\citealp{sun2007oracle, CaiSun09, heller2021optimal}) is a thresholding rule based on the conditional local false discovery rates (CLfdr)
  \begin{equation} \label{zlfdr_alternative_rep}
\mbox{CLfdr}_i \equiv P(H_i = 0 | Z_i, X_i)  =P(H_i = 0 |  U_i, X_i) =  \frac{(1 - w_{X_i}) h_0(U_i)}{ h_{X_i}(U_i)}, i =1, \dots, m.
\end{equation} 
Since  each $\mbox{CLfdr}_i$ is a function of $U_i$ conditional on $X_i$,
we let $\text{CLfdr}_x(u):(0,1) \rightarrow (0,1)$ be the corresponding function defined on the $u$-value scale for a given  realized covariate value $x$. Related data-driven CLfdr procedures involve first estimating the CLfdr statistics, and second determining a threshold for them using, for example, step-wise algorithms (\citealp{sun2007oracle}), randomized rules (\citealp{Basetal18}) or linear programming (\citealp{heller2021optimal}).   
However, the first estimation step poses significant challenges as it boils down to a hard \emph{density regression} problem \citep{dunson2007bayesian}. For example, to estimate $f_{x}(\cdot)$ (or equivalently $h_x(\cdot)$), a line of works  \citep{scott2015false, tansey2018false, deb2021two} proceeds by assuming a fixed alternative density, i.e.
\begin{equation} \label{FDRregAssumption}
f_{1, x}(z) \equiv f_1(z),
\end{equation}
to make way for the application of an EM algorithm. If the assumption fails to hold, the CLfdr statistics can be poorly estimated and lead to both invalid FDR control and adversely affected power. The non-parametric CARS procedure developed in \citet{cai2019covariate} does not require the assumption in \eqref{FDRregAssumption}. However, it still employs the CLfdr statistics as its basic building blocks, which are estimated with kernel density methods. Due to the curse of dimensionality, the methodology becomes unstable in the presence of multivariate covariates, which has limited its applicability. 

By contrast, ZAP strives to sensibly emulate the oracle procedure without heavy reliance on the quality of the CLfdr estimates, which is its key strength. To motivate our data-driven procedures in the next sections, we shall first discuss a prototype ZAP procedure to illustrate two key steps of our testing framework: (a) how to combine $Z_i$ (or equivalently $U_i$) and $X_i$ for assessing the significance of hypotheses; and (b) how to threshold the new significance indices. 

Step (a) involves the construction of an \emph{assessor function}\footnote{Or simply known as an assessor.} $a_x(u) : (0, 1) \rightarrow (0, 1)$, which seeks to approximate the $\text{Clfdr}_x(u)$ function  to integrate the information in both the $u$-value and covariate. For the present assume that $a_x(\cdot)$ is pre-determined. Let $T_i\equiv a_{X_i}(U_i)$ be the new significance index for $i$  and $c_i(t)\equiv P(T_i \leq t |H_i = 0, X_i)$ be its null distribution conditional on $X_i$. Assume that $c_i(\cdot)$ is continuous and strictly increasing\footnote{Both  are true  as the consequences of the way we will construct $a_x(\cdot)$; see the discussion after  \lemref{recipProp} in \appref{workModProp}.}, and denote its inverse by $c_i^{-1} (\cdot)$. All hypotheses will then be ordered according to the $T_i$'s, with a smaller $T_i$ indicating a more significant hypothesis.

In Step (b), we aim to determine a threshold for the $T_i$'s to control the FDR. This involves the construction of a conservative FDP estimator for any candidate threshold $t$  by the Barber-Candes (BC) method  \citep{arias2017distribution, barber2015controlling}:
\begin{equation} \label{FDP_est_prototype}
\widehat{\mbox{FDP}}(t) \equiv \frac{ 1 + \# \{i: S_i  \geq 1- c_i(t)\}}{\# \{i:   T_i \leq  t \} \vee 1} =\frac{ 1 + \# \{i: T^ \mathfrak{m}_i \leq  t \}}{\# \{i:   T_i \leq t \} \vee 1},
\end{equation} 
where, given that $i$ is a true null,  $S_i \equiv c_i(T_i)$ is the probability of realizing a smaller significance index and $T^ \mathfrak{m}_i\equiv c_i^{-1} (1-S_i)$ is the \emph{mirror statistic} that ``reflects" $T_i$'s position in the distribution $c_i$. Define 
\begin{equation} \label{t_alpha_ideal}
\hat{t}(\alpha) \equiv \max\{ t \in (0, t_{\max}]:\widehat{\mbox{FDP}}(t)   \leq \alpha \}, 
\end{equation}
where $t_{\max} \equiv  \max\{t: c_i(t) \leq 0.5 \text{ for all } i \}$. It follows from  \citet[Lemma 1]{barber2019knockoff} that a procedure which rejects $i$ whenever $T_i \leq \hat{t}(\alpha)$ controls the FDR at level $\alpha$; see \appref{Pf_ideal}.  Importantly,  the FDR  is controlled under the desired level $\alpha$ whether $a_x(\cdot)$ is a  good approximation of $\text{CLfdr}_x(\cdot)$ or not.

However, the assessor $a_x(\cdot)$, which is taken as pre-determined thus far, is to be estimated from the observed data in practice. This leads to additional difficulties in both methodological and theoretical developments; for one thing, the theory in \cite{barber2019knockoff} cannot be directly applied to prove the FDR controlling property.  Section  \ref{sec:betaMod} discusses  a working beta-mixture model, whose parameters can be estimated from the observed data and subsequently used to construct a data-driven assessor  $\hat a_x(\cdot)$. From there we can test the hypotheses in a data-driven manner by either implementing the prototype procedure directly using $\hat a_x(\cdot)$ as if it is pre-determined (\secref{ZAPasymp}), or mimic the prototype procedure in a more nuanced manner by leveraging the partial data masking technique   in \cite{lei2016adapt} (\secref{ZAPfinite}).
 These two variants of ZAP have their own relative strengths and weaknesses: the direct approach offers asymptotic FDR control under suitable regularity conditions, and is both computationally and power efficient, while the data masking approach offers finite-sample FDR control but is computationally intensive and moderately less powerful in practice; see \appref{speed} for a discussion on aspects of their computational costs.

\subsection{A beta-mixture model}  \label{sec:betaMod}

We now develop a working model to approximate \eqref{U_cond_dens}, which will be subsequently used to construct the assessor. We propose to capture the overall shape of $h_x(u)$ using a three-component mixture: 
\begin{equation}\label{hu-beta}
h_{x}(u) = (1-\pi_{l, x} - \pi_{r, x}) h_0(u)+\pi_{l, x} h_{l, x}(u)+\pi_{r, x}h_{r,x}(u), 
\end{equation}
where, given $X_i=x$, $\pi_{l,x}$ and $\pi_{r,x}$ respectively denote the mixing probabilities that $\mu_i<0$ and $\mu_i>0$ \footnote{The different symbols $\pi_{l,x}$ and $\pi_{r,x}$ are used in the working model. In the true data generating model \eqref{U_cond_dens}, the mixing probability is  denoted $w_x$.}, and $h_{l,x}$ and $h_{r,x}$ respectively represent the densities of the negative and positive effects (on the left and right sides of the null). Our working model assumes that $\pi_{l,x}$ and $\pi_{r,x}$ are multinomial probabilities with regression parameter vectors $\theta_l$ and $\theta_r$:
$$
\pi_{l,x} = \frac{\exp( \tilde{x}^T \theta_l)}{1 + \exp( \tilde{x}^T \theta_r) + \exp( \tilde{x}^T\theta_l)}, \;
  \pi_{r, x} = \frac{ \exp( \tilde{x}^T \theta_r)}{1 + \exp( \tilde{x}^T\theta_r) + \exp( \tilde{x}^T \theta_l)},
$$
where $\tilde{x}= (1, x^T)^T$ is the intercept-augmented covariate vector. Further, $h_{l,x}$ and $h_{r,x}$ are chosen to be beta densities with regression parameters $\beta_l$ and $\beta_r$: 
$$
h_{l,x}(u) = \frac{1}{B(k_{l,x}  ,\gamma_l ) } u^{k_{l,x}  - 1} (1 - u )^{\gamma_l - 1}, 
h_{r,x}(u) =\frac{1}{B( \gamma_r , k_{r,x} ) } u^{\gamma_r - 1} (1 - u )^{k_{r,x} - 1},
$$
where $k_{l,x} = \{1 + \exp(- \tilde{x}^T \beta_l)\}^{-1}$ and $k_{r,x} = \{1 + \exp(- \tilde{x}^T \beta_r)\}^{-1}$, for  two \emph{fixed} shape parameters $\gamma_l$ and $\gamma_r$. $h_{l,x}$ and $h_{r,x}$ are respectively left-leaning (right-skewed) and right-leaning (left-skewed) functions. We require that $\gamma_l>2$ and $\gamma_r> 2$ to ensure that both  are strictly monotone and \emph{convex}, and thus provide a reasonable approximation to the underlying true density in practice; see \lemref{convex} in \appref{workModProp} for a precise result. The exact choices for $\{\gamma_l, \gamma_r\}$ will be further discussed in \secref{beta_expr}. 
The working model may be generalized to capture non-linearity in $x$ using, say, spline functions. 

Beta mixtures have long been identified as a flexible modeling tool for variables taking values in the unit interval; see \citet{pounds2003estimating, ji2005applications, parker1988identifying, markitsis2010censored, migliorati2018new, ferrari2004beta} for related works. In the context of covariate-adaptive multiple testing, \cite {lei2016adapt} and \cite{zhang2020covariate} employ a two-component beta-mixture model for the $p$-values that consists of a uniform  and another left-leaning beta component. Our working model defined on the $u$-value scale can be viewed as a natural extension of these works to capture important patterns in the $u$-value distribution associated with two-sided covariate-adaptive testing.

The assessor can be constructed as the $\text{CLfdr}_x(\cdot)$ function with respect to our working model \eqref{hu-beta}. Since $h_0 \equiv 1$, it follows that 
\begin{equation} \label{assessor_form}
a_x(u) \equiv \frac{1 -\pi_{l,x} - \pi_{r,x}}{1 - {\pi}_{l,x} - {\pi}_{r,x} + {\pi}_{l,x} {h}_{l,x}(u) + {\pi}_{r,x}{h}_{r,x}(u)}, \quad 0<u<1.
\end{equation} 
The corresponding data-driven assessor is denoted by $\hat{ a}_x(u)$ if the parameters $\{\theta_l, \theta_r, \beta_l, \beta_r\}$ are estimated from the data for its construction. 
\subsection{Asymptotic ZAP} \label{sec:ZAPasymp}

We now develop a direct data-driven version of the prototype algorithm in \secref{assessor}. To construct $\hat{ a}_x(u)$, we first obtain the maximum likelihood estimates (MLE) of   the  unknown regression parameters $\{\theta_l, \theta_r, \beta_l, \beta_r\}$  with the data $\{U_i, X_i\}_{i =1}^m$; the EM algorithm for their computations are provided in \appref{EM_zap_asymp}. Denote $\hat{T}_i \equiv \hat{a}_{X_i}(U_i)$, and let $\hat{c}_i (\cdot)$ be its null distribution by treating $\hat{a}_{X_i}(\cdot)$ as if it is pre-determined. With $\hat{S}_i \equiv  \hat{c}_i (\hat{T}_i)$, the estimated mirror statistics are correspondingly defined as $\hat{T}^ \mathfrak{m}_i  \equiv \hat{c}_i^{-1}(1 -\hat{S}_i)$, which can be computed numerically by performing quantile estimation. The FDP for a candidate threshold $t$ can be estimated as 
\begin{equation} \label{FDP_est_zap_asymp}
\widehat{\mbox{FDP}}_{asymp}(t) \equiv
 \frac{1 + \#\{  i: \hat{T}^ \mathfrak{m}_i \leq t\}}{\# \{i: \hat{T}_i \leq t\} \vee 1}.  \end{equation}
Define
$\hat{t}_{asymp} (\alpha)\equiv \sup\{ 0 \leq t \leq 1 : \widehat{\mbox{FDP}}_{asymp } (t) \leq \alpha \}$, and reject $i$ whenever $\hat{T}_i \leq \hat{t}_{asymp} (\alpha)$. In practice, it suffices to consider only the values of $\hat{T}_1, \dots, \hat{T}_m$ as candidate thresholds. This procedure is summarized in \algref{zapAsymp}.

\begin{algorithm}[h]
\caption{Asymptotic ZAP}\label{alg:zapAsymp}

Construct $\hat{a}_{X_i}(\cdot)$'s using the MLEs obtained via the EM algorithm  in \appref{EM_zap_asymp} and compute $\hat{T}_i=\hat{a}_{X_i}(U_i)$ for each $i$.
  
Compute the mirror statistics $\{\hat{T}^ \mathfrak{m}_i\}_{i=1}^m$: \newline
 \begin{inparaenum} 
 \item Generate i.i.d. realizations $u_1, \dots, u_N$ from $\mbox{Unif}(0,1)$ for a large $N$. 
 \item For each $i$, evaluate $\hat{a}_{X_i}(u_1), \dots, \hat{a}_{X_i}(u_N)$ to simulate the null distribution $\hat{c}_i(\cdot)$. Compute $\hat{T}^\mathfrak{m}_i$ via e.g. \texttt{quantile()} in \texttt{R}. 
 \end{inparaenum}
 
Order $\{\hat{T}_i\}_{i=1}^m$ as $\hat{T}_{(1)} \leq \dots \leq \hat{T}_{(m)}$. Reject $i$ if $T_i \leq  \hat{T}_{(k)}$, 
where $ 
k=\max \left\{ l \in \{1, \dots, m\} : \frac{1  + \# \{ i:  \hat{T}^ \mathfrak{m}_i \leq  \hat{T}_{(l)}\} }{l\vee 1}\leq \alpha \right\}.$

\end{algorithm}

The main theory requires the following classical assumption from the literature on misspecified models \citep{white1981consequences, white1982maximum}:

\begin{assumption}[Existence of a unique maximizer] \label{assump:cpct}
The expected log-likelihood
\[
 \mathbb{E} \log [ (1 - \pi_{l, X_i} - \pi_{r, X_i}) + \pi_{l, X_i}h_{l, X_i} (U_i) + \pi_{r, X_i} h_{r, X_i} (U_i) ]
\]
of the beta-mixture model \eqref{hu-beta} 
has a unique maximum at $\{\theta^*, \beta^*\}$ over $\theta \in \Theta$ and $\beta \in \boldB$ for  compact spaces $\Theta$ and $\boldB$, where $\theta \equiv (\theta_l^T, \theta_r^T)^T$ and $\beta \equiv (\beta_l^T, \beta_r^T)^T$. The expectation is taken with respect to the true joint distribution of $\{H_i, Z_i, X_i\}$.
\end{assumption}
Together with \assumpsref{regularity} - \assumpssref{convgAssump} in \appref{add_assume}, which are standard regularity and strong-law conditions, we can prove the following asymptotic FDR controlling property.  

\begin{theorem} \label{thm:ZAPasympControl}
Let $\hat{a}_x(\cdot)$ be constructed with the MLE 
$
\{\hat{\theta}, \hat{\beta} \} 
= \underset{\theta, \beta}{\mathrm{argmax}} \sum_{i =1}^m\log h_{X_i} (U_i)
$
of the beta-mixture model \eqref{hu-beta}. 
Under \assumpsref{cpct}-\assumpssref{convgAssump}, 
the procedure that rejects $i$ whenever $\hat{T}_i \leq \hat{t}_{asymp} (\alpha)$ controls the FDR asymptotically in the sense that 
$\limsup_{m \rightarrow \infty} \mbox{FDR} \leq \alpha.$
\end{theorem}

We highlight two aspects of this result. First, it doesn't require  the estimated assessor function to be a good proxy for $\text{CLfdr}_x(\cdot)$. Hence, its theory is  more attractive than that of \cite{cai2019covariate}, which requires consistent CLfdr estimates to ensure asymptotic FDR control. 
Second,  to establish Glivenko-Cantelli results (\lemref{GCthm}) for the following three empirical processes 
\begin{equation*}
m^{-1} \sum_{i = 1}^m I\left(\hat{T}_i \leq t\right), \quad
m^{-1} \sum_{i = 1}^m (1 - H_i) I\left( \hat{T}_i  \leq t\right), \quad
m^{-1}\sum_{i = 1}^m I\left\{\hat{S}_i \geq  1 - \hat{c}_i(t)\right\},
\end{equation*}
 typical of similar asymptotic analyses \citep{zhang2020covariate, storey2004strong}, we heavily utilize the 
convexity properties  (\lemref{recipProp}) of the functional form in \eqref{assessor_form} to uniformly control the deviations of the estimated assessors $\hat{a}_{X_i}(\cdot)$ from the assessors $a_{X_i}^*(\cdot)$ constructed with the population parameters $\{\theta^*, \beta^*\}$;
the delicate techniques involved may be of independent interest.

\subsection{Finite-sample ZAP} \label{sec:ZAPfinite}




This section introduces an alternative ZAP procedure that offers finite-sample control of the FDR. The operation again involves approximating the CLfdr statistics via an assessor function. However, the thresholding step is based on a more nuanced approach to  FDP estimation inspired by the p-value method AdaPT \citep{lei2016adapt}. In this approach,  multiple testing is conducted in an iterative manner, where data are initially partially masked and then gradually revealed at steps $t=0, 1, \ldots$, with the thresholds sequentially updated based on the revealed data at each step. In what follows, if $g_1(\cdot)$ and $g_2(\cdot)$ are two functions defined on the same space, $g_1 \preceq g_2$ means $g_1(x) \leq g_2(x)$ for all $x$ in that space. If $C$ is a constant, $g_1 \preceq C$ means $g_1(x) \leq C$ for all $x$. Similarly we can define $g_1 \succeq g_2$ and $g_1 \succeq C$.

Since the inherent convexity\footnote{Refer to \lemref{recipProp}} of the assessor functional form in \eqref{assessor_form} suggests that rejecting hypotheses with small values of $T_i = a_{X_i}(U_i)$ amounts to rejecting extreme $u$-values near 0 or 1, our iterative algorithm emulates this essential operational characteristic of the prototype procedure. We first divide the covariate values into a left and right group based on the observed $u$-values:
\[
\mathcal{X}_l = \{X_i: U_i \leq 0.5\} \text{ and } \mathcal{X}_r = \{X_i: U_i > 0.5\}.
\]
At each step $t = 0, 1, \dots$, let $s_{l,t}: \mathcal{X}_l \rightarrow [0, 0.25]$ and $s_{r,t}: \mathcal{X}_r \rightarrow [0.75, 1]$ denote two corresponding thresholding functions, and define the candidate rejection set $\mathcal{R}_t \equiv \mathcal{R}_{l,t} \cup   \mathcal{R}_{r,t}$, where
\begin{equation}\label{R_t}
\mathcal{R}_{l,t} \equiv \{i:  U_i  \leq  s_{l,t}(X_i) \wedge  0.5\} \text{ and }  \mathcal{R}_{r,t} \equiv \{i:   U_i \geq  s_{r,t}(X_i) \vee 0.5\}. 
\end{equation}
Let $\mathcal{A}_t \equiv \mathcal{A}_{l,t} \cup \mathcal{A}_{r,t}$ be the corresponding set of ``accepted" hypotheses, where
$$
\mathcal{A}_{l,t} \equiv \{i: 0.5 - s_{l,t}(X_i) \leq U_i \leq 0.5 \} \text{ and }
 \mathcal{A}_{r,t} \equiv \{i: 0.5 < U_i \leq1.5 - s_{r,t}(X_i) \}.
$$
Intuitively, $|\mathcal{A}_{l,t}|$ estimates the number of false rejections in the left candidate rejection set $\mathcal{R}_{l,t}$: Given $H_i = 0$ and $U_i \leq 0.5$,  the events $\{U_i < s_{l,t}(X_i)\}$ and $\{U_i > 0.5- s_{l,t}(X_i)\}$ are equally likely. The FDP of possibly rejecting $\mathcal{R}_{t}$ at step $t$ can then be estimated as
\begin{equation}\label{FDP_ZAP_finite}
\widehat{\mbox{FDP}}_{finite}(t)= \frac{ 1+ |\mathcal{A}_t|}{|\mathcal{R}_t |\vee 1}.
\end{equation}
If $\widehat{\text{FDP}}_{finite}(t)\leq\alpha$,  the algorithm terminates and the hypotheses in $\mathcal{R}_t$ are rejected. Otherwise, the algorithm proceeds to the next step $t+1$ and updates the two thresholding functions under two restrictions. First, it must be that $s_{l,t+1}\preceq s_{l,t}$ and $s_{r,t+1} \succeq s_{r,t}$; this ensures that $\mathcal{R}_t$ shrinks in size as $t$ increases. Second, $s_{l,t+1}$ and $s_{r,t+1}$ must be updated based on the knowledge of $|\mathcal{R}_t|$, $|\mathcal{A}_t|$ and the  partially masked data $\{\tilde{U}_{t, i}, X_i\}_{i = 1}^m$ only, where 
\begin{equation} \label{U_tilde}
\widetilde{U}_{t, i}   \equiv 
\begin{dcases*}
        U_i  & if $U_i \not\in  \mathcal{A}_t \cup  \mathcal{R}_t$\\
        \{ U_i , \widecheck{U}_i\} & if $U_i \in  \mathcal{A}_t \cup  \mathcal{R}_t$
        \end{dcases*}
\end{equation}
is a singleton or a two-element set depending on whether $i$ is in the ``masked" set $\mathcal{A}_t \cup \mathcal{R}_t$, and $\widecheck{U}_i$ is  the ``reflection" of $U_i$ about the  ``middle" axis at $u = 0.25$ or $u = 0.75$, depending on which group (left or right) $U_i$ belongs to:
$$
\widecheck{U}_i \equiv (1.5 - U_i)  I(U_i > 0.5) +(0.5 - U_i ) I(U_i \leq 0.5).
$$
For example, if the underlying $U_i$ is $0.1$ and $i$ is  masked  at step $t$, the algorithm can only update for $s_{l, t+1}$and $s_{r, t+1}$  with the partial knowledge that $U_i$ is either $0.1$ or its reflection value $0.4$.   \algref{updateThreshold} in \appref{EM_finite_update} describes one such updating scheme which first applies an  EM algorithm (\appref{EM_finite_ZAP}) acting only on the partially masked data  to estimate the beta-mixture model \eqref{hu-beta}. With the latter subsequently used to produce estimated $\text{CLfdr}_i$'s for comparing the significance of the  hypotheses in the masked set  $ \mathcal{A}_t \cup  \mathcal{R}_t$, the thresholding functions are updated in such a way that the masked $i$ deemed to be the least significant will have its $u$-value  revealed at step $t+1$. 
\figref{uvPlot} illustrates how the data $\{U_i, X_i\}_{i=1}^m$ are partitioned into $\mathcal{A}_t$, $\mathcal{R}_t$ and the unmasked set $\{1, \dots, m\} \backslash \{\mathcal{A}_t \cup \mathcal{R}_t\}$ at a given step $t$, based on Example 2.2 in \secref{infoLoss}.  In particular,  we remark that the algorithm cannot tell the true data point from a given red-pink (blue-cyan)  pair in the plot $(b)$ where the reflection points $\{\widecheck{U}_i\}_{i \in \mathcal{A}_t \cup \mathcal{R}_t}$ are also shown.

\begin{figure}[h]
\centering
\includegraphics[width= \textwidth]{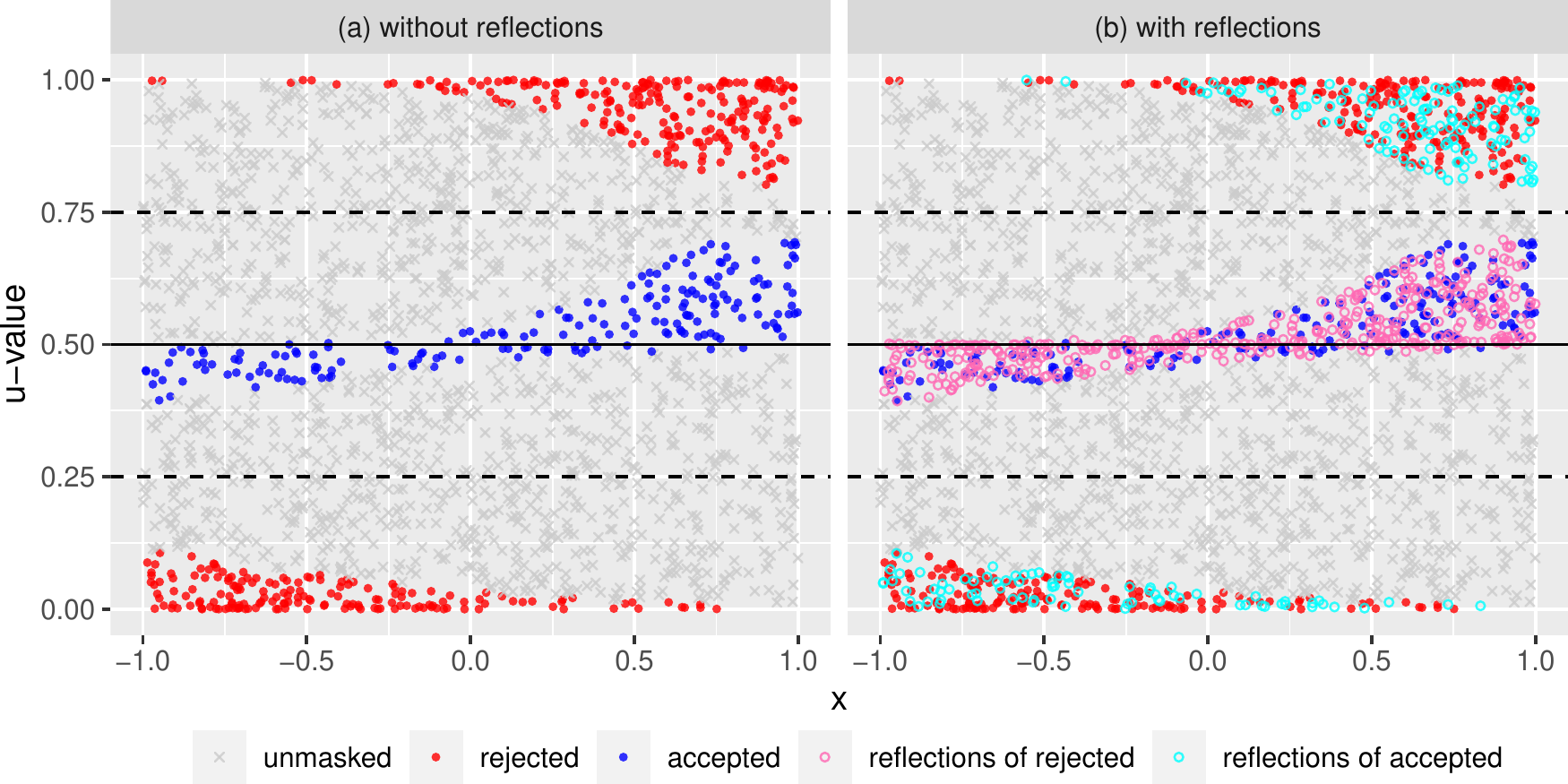}
\caption{ 
Illustration of  \algref{zapFinite} at a step $t$ based on Example 2.2, $m = 2000$. {\bf (a)}: The red, blue and grey are points in the respective sets $\mathcal{R}_t$, $\mathcal{A}_t$ and $\{1, \dots, m\}\backslash \{\mathcal{A}_t \cup \mathcal{R}_t\}$, where $\widehat{\text{FDP}}_{finite} (t) = 193/389 \approx 0.5$.  {\bf (b)}: The reflections of the points in $\mathcal{R}_t $ and $\mathcal{A}_t$ are respectively shown in pink and cyan. A pink $\widecheck{U}_i$ below (above) $0.5$ is the reflection of a red $U_i$ with the same covariate value about the middle axis at $u = 0.25$ ($u = 0.75$) of the left (right) group; the cyan are the reflections of the blue.  }
\label{fig:uvPlot}
\end{figure}

The steps described above are summarized in \algref{zapFinite}, whose  finite-sample FDR controlling property is stated in  \thmref{ZAPfiniteControl}.

\begin{algorithm}[h]
\caption{Finite-sample ZAP}
 \label{alg:zapFinite} 
\KwData{ $\{U_i, X_i \}_{i = 1}^m$}
\KwIn{FDR level $\alpha$, initial thresholding functions $s_{l,0} \preceq 0.5$ and $s_{r,0} \succeq 0.5$; }

\For{t= 0,1 \dots, }{
Compute $\widehat{\mbox{FDP}}_{finite}(t)$  in \eqref{FDP_ZAP_finite}\;
  \eIf{$\widehat{\mbox{FDP}}_{finite}(t) > \alpha$}{
   Update $s_{l,t+1}$ and $s_{r,t+1}$ while respecting the two conditions in \thmref{ZAPfiniteControl}. E.g. Apply \algref{updateThreshold} in \appref{EM_finite_update}\;
   }{
      Record $\mathcal{R}_t$; break\;
  }
 }
\KwOut{Reject all hypotheses in $\mathcal{R}_t$.}
\end{algorithm}

\begin{theorem}[Finite-sample FDR control]\label{thm:ZAPfiniteControl}
Under the conditions that
 \begin{inparaenum} 
 \item  $s_{r,t +1 } 	\succeq s_{r,t }$ and  $s_{l,t +1 } 	\preceq s_{l,t }$ and 
 \item $s_{r,t +1 }$, and $s_{l,t +1 }$ are updated based on $|\mathcal{R}_t|$, $|\mathcal{A}_t|$ and $\{\tilde{U}_{t, i}, X_i\}_{i = 1}^m$ only,
 \end{inparaenum} \algref{zapFinite} controls the FDR under $\alpha$ for finite samples. Specifically, we have
\[
\mathbb{E}\left[ FDP\Big|\{H_i, X_i\}_{i = 1}^m\right] \leq \alpha. 
\]
\end{theorem}

Lastly, we highlight a crucial difference between \algref{zapFinite} and AdaPT in the present context.  Operating  on the two-sided $p$-values, AdaPT  proceeds iteratively with a single thresholding function $s_t$ defined on $\{X_i\}_{i=1}^m$ such that $s_t \preceq 0.5$, and the ratio $\frac{1 + |\{i: P_i \geq 1 - s_t(X_i)\}|}{ 1 \vee |\{i: P_i \leq s_t(X_i)\}|}$ is used as an FDP estimator for the candidate rejection set $\{i: P_i \leq s_t(X_i)\}$. 
It is easy to see that 
\begin{equation} \label{oper_char_adapt}
P_i \leq s_t(X_i)  \iff U_i \leq s_t(X_i)/2 \text{ or } U_i\geq 1 - s_t(X_i)/2. 
\end{equation}
Hence, on the $u$-value scale, AdaPT always adopts symmetric rejection regions about $u = 0.5$. By contrast, \algref{zapFinite} employs two different thresholding functions $s_{l,t}$ and $s_{r,t}$, which allow for asymmetric rejection regions, and therefore provides additional flexibility to fully capitialize on covariate information for two-sided tests.  As seen in \figref{uvPlot}, the pattern of the candidate rejection points in red agrees with the middle panel of \figref{oracle_rej_reg}; as the covariate increases from $-1$ to $1$, the algorithm's rejection priorities change from the $u$-values near $0$ to those near $1$.

\subsection{Implementation details}  \label{sec:beta_expr}
An \texttt{R} package \texttt{zap} for our two data-driven methods is available on  \url{https://github.com/dmhleung/zap}, and we shall  discuss further details of their implementation.

For both data-driven procedures, the shape parameters $\{\gamma_l, \gamma_r\}$ of the working model need to be pre-specified before running the EM algorithms.  While requiring $\gamma_l, \gamma_r > 2$ ensures a convex shape for  the three-component beta-mixture density (\lemref{convex}), we recommend choosing $(\gamma_l, \gamma_r) = (4, 4)$ as a default, which has yielded consistently good performance in our numerical studies. 

To illustrate the effectiveness of our recommendation, we simulate 8000 i.i.d. $z$-values $Z_1, \dots, Z_{8000}$ from the normal mixture model
\begin{equation} \label{toy}
0.78 f_0(z)+0.15 \phi(z+1.5)+0.07 \phi(z-2)
\end{equation}
without any covariates. The histogram of the corresponding $u$-values is plotted in \figref{shape}(a), overlaid with the true underlying density function, as well as estimated densities of  the beta mixture  \eqref{hu-beta} fitted with regression intercepts only, where $(\gamma_l,  \gamma_r)$ is respectively fixed at $(1, 1)$ and $(4, 4)$. When modeling p-values with a two-component beta mixture,  \citet{lei2016adapt} and \citet{zhang2020covariate} set an analogous shape parameter to be $1$, so $(\gamma_l,  \gamma_r) = (1, 1)$ would be a seemingly natural choice to extend  their model for two-sided tests. Both fitted densities visually coincide with the true density, attesting to the flexibility of beta mixtures for modeling data on the unit interval. However, the estimated component probabilities differ significantly for $(\gamma_l,  \gamma_r) = (1, 1)$ vs $(\gamma_l,  \gamma_r) = (4, 4)$. In \figref{shape}(b), we present the estimated quantities pertaining to the non-null components. We can see that setting $(\gamma_l,  \gamma_r) = (1, 1)$ has drastically overestimated the left and right non-null probabilities, whereas setting $(\gamma_l,  \gamma_r) = (4, 4)$ provides good approximations to the truths. 
\begin{figure}[t]
\centering
\includegraphics[width=\textwidth]{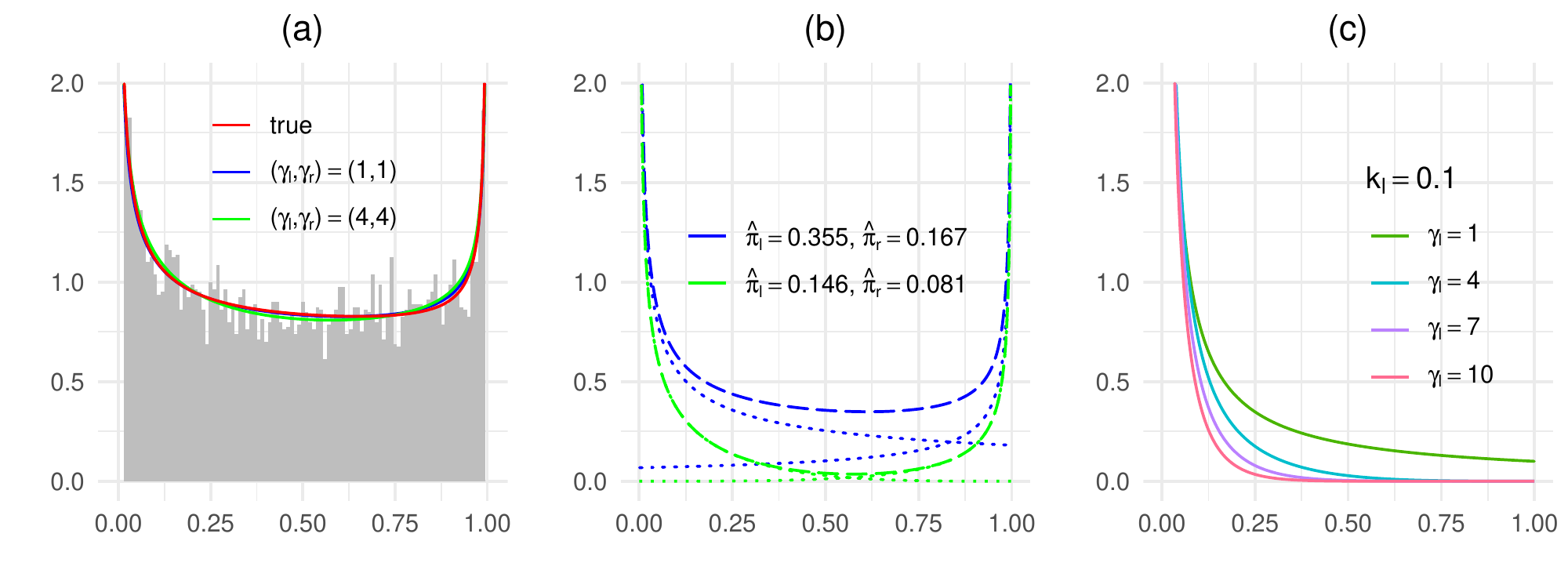}
\caption{
{\bf (a)}: Histogram of $U_i = \Phi(Z_i)$ generated by \eqref{toy}. The red curve is the true density; while the blue and green curves respectively correspond to the estimated densities of our beta-mixture model with $(\gamma_l, \gamma_r)$ set at $(1, 1)$ and $(4, 4)$. 
{\bf (b)}: The long dashed blue curve is the entire non-null component of the solid blue estimated density in $(a)$. The two dotted blue curves are the left-leaning and right-leaning non-null components that add up to the long dashed blue. The green curves are constructed analogously with respect to the solid green density in $(a)$. The legend shows the estimated probabilities  for the left and right-leaning components, where the subscript ``$x$" is omitted from $\hat{\pi}_l$ and $\hat{\pi}_r$ as the fit uses intercepts only. 
{\bf (c)}: Plot of the left-leaning beta density $B(k_l , \gamma_l)^{-1} u^{k_l - 1} (1 - u )^{\gamma_l- 1}$ for $k_l = 0.1$ and different values of $\gamma_l$. }
\label{fig:shape}
\end{figure}
To gain insight into why larger shape parameters are preferred, in \figref{shape}(c) we plot  the density of a left-leaning beta density 
\[
B(k_l , \gamma_l)^{-1} u^{k_l - 1} (1 - u )^{\gamma_l- 1}
\]
for different values of $\gamma_l$ and a fixed $k_l = 0.1$, which supposedly captures the negative effects in two-sided tests. We can see that small values of $\gamma_l$ tend to yield a density component that slants in the middle of the unit interval. As a result, when added to another right-leaning beta density for the positive effects with a similar but mirroring shape, it gives rise to an overall non-null density component with a large plateau in the middle of the interval $(0, 1)$ akin to the U-shaped blue curve in  \figref{shape}$(b)$. This inflates the non-null component probability estimates. In contrast, larger values of $\gamma_l$ and $\gamma_r$ effectively mitigate the issue by rendering sharply convex non-null component densities like the purple and pink curves in \figref{shape}(c), avoiding overestimation of the non-null probabilities and leading to a better estimated overall non-null density component like the U-shaped green curve in \figref{shape}$(b)$. 
More setups are experimented in \appref{betaModSim}; the careful choice of $(\gamma_l, \gamma_r) = (4,4)$  produces reasonable probability estimates throughout.

Other aspects of implementation are as follows. For the asymptotic method (\algref{zapAsymp}), since a large $N$ allows us to compute the mirror statistics up to arbitrary precision, we evaluate at $N = 50000$ uniform realizations by default. For the iterative finite-sample method (\algref{zapFinite}), we set the initial thresholding functions as $s_{l0} \equiv 0.2$ and $s_{r0} \equiv 0.8$, but other values close to $0.25$ and $0.75$ tend to be equally effective.  We also update the thresholding functions   every $\lceil m/100 \rceil$ steps. Ideally one would want to update at every step along the way to reveal the masked $u$-values sooner. However, it is more practical to carry out intermittent updates since the EM component involved in \algref{updateThreshold} is computationally costly. Lastly, one can also perform feature selection at any step if $X_i$ is multivariate, as long as it is done properly based on the masked data, akin to what was suggested by \citet[Section 4.2]{lei2016adapt}. We have not performed this step for simplicity.

\section{Numerical studies}\label{sec:numeric}

We conduct numerical studies to gauge the performance of ZAP alongside other methods on  both simulated and real data.  For expositional considerations, here we only limit the comparisons to  a  selection of representative FDR methods. This makes the ensuing graphs (\figsref{all_setup}-\figssref{real_data_plot}) less cluttered with lines and easier to read. Comparisons with more methods in the literature are included  in \appsref{more_simulation_results} and \appssref{real_data_append}, but the basic conclusions do not change. Here, we consider:
 \begin{enumerate}[(a)]
  \item ZAP (asymp): \algref{zapAsymp} with specifications described in \secref{beta_expr}.
 \item ZAP (finite):  \algref{zapFinite} with specifications described in \secref{beta_expr}.
  \item CAMT:  the covariate-adaptive multiple testing method  \citep{zhang2020covariate}.
   \item AdaPT:  the adaptive $p$-value thresholding method \citep{lei2016adapt}. Their working model is updated based on the EM algorithm for every $\lceil m/100 \rceil$ steps; other default specifications are chosen based on the R package \texttt{adaptMT}.

   \item IHW: Independent hypothesis weighting  \citep{ignatiadis2016data, ignatiadis2021covariate}, based on the default provided by the R package  \texttt{IHW}, version 1.16.0, which provides asymptotic FDR control  \citep[Proposition 1]{ignatiadis2021covariate}. This method only handles univariate covariates.
    \item FDRreg: false discovery rate regression method \citep{scott2015false}. The theoretical null $\mathcal N(0, 1)$ has been used.
    \item AdaPT-$\text{GMM}_g$: A  $z$-value based variant of AdaPT by \citet{chao2021adapt} which is also a  data masking procedure similar to ZAP (finite); unlike ZAP (finite), it instead employs a covariate-dependent Gaussian mixture working model for the $z$-values, which is updated with an EM algorithm for every $\lceil m/100 \rceil$ steps. Other default specifications are chosen as in the R package \texttt{AdaPTGMM}.  \end{enumerate}
Among them, ZAP, AdaPT-$\text{GMM}_g$, and FDRreg are $z$-valued based, while all other methods are $p$-value based.

\subsection{Simulated data}\label{sec:sim_data}

%
%


We simulate data  to test $m = 5000$ hypotheses. Two-dimensional covariates $X_i = (X_{1i}, X_{2i})^T$, $i =1, \ldots, m$, are independently generated from the bivariate normal distribution
$
 \mathcal N\left\{(\begin{smallmatrix}
0 \\
0
\end{smallmatrix}),(\begin{smallmatrix}
1/2 & 0\\
0 & 1/2
\end{smallmatrix})\right\}.
$
Conditional on $X_i = x \equiv (x_1, x_2)^T$, $Z_i$ is generated with  a normal mixture  density
\begin{equation} \label{normMix}
(1 - w_{l, x} - w_{r, x}) f_0(z)+w_{l, x} \phi\left(\frac{z-\mu_{l, x}}{\sigma}\right)+  w_{r, x} \phi\left(\frac{z-\mu_{r, x}}{\sigma}\right), 
\end{equation}
where $w_{l, x}$ and $w_{r, x}$ \footnote{These data generating probabilities $\{w_{l,x}, w_{r, x}\}$ should again be distinguished from $\{\pi_{l, x}, \pi_{r, x}\}$ in the working model \eqref{hu-beta}.} are probabilities that control the sparsity levels of negative and positive effects, $\mu_{l, x} <0$ and $\mu_{r, x} > 0$ are negative and positive non-null normal means, and $\sigma^2$ is the variance of the alternative components.  The covariate-adjusted overall non-null density is then given by
\begin{equation} \label{setup_cond_alt_density}
 f_{1, x}(z)  =  \frac{w_{l, x} \phi( \frac{z-\mu_{l, x}}{\sigma})+  w_{r, x} \phi(\frac{z-\mu_{r, x}}{\sigma})}{w_{l, x} + w_{r, x}} . 
\end{equation}
We shall allow $\{w_{l, x},  w_{r, x}, \mu_{l, x}, \mu_{r, x}\}$ to depend on $x$ in different ways to induce the simulation setups below, which can be considered as more realistic versions of the stylized examples in Section \ref{sec:infoLoss}. We fix $\sigma^2 = 1$ in this section, while other values of $\sigma^2 $ will be explored in \appref{more_simulation_results}, which also contains more extensive simulation studies as well as a comparison of the computational efficiency of the different methods. Note that the sum $X_{\bullet i}\equiv X_{1i} + X_{2i}$ of the covariate components is $\mathcal{N}(0, 1)$-distributed, and $x_{\bullet} \equiv x_1 + x_2$ will denote a realized value of it below. 
\begin{enumerate}[]
\item {\bf Setup 1} (Asymmetric alternatives). \label{set:S1}The quantities in \eqref{normMix} are 
 \begin{multline*}
 w_{r, x} =  \frac{1}{1 + \exp(-   \eta - 	\zeta x_{\bullet})},  \ \ 
\mu_{r, x} =   \frac{2\varepsilon}{1 + \exp ( - 	\zeta x_{\bullet} )}, \ \ w_{l, x} = 0, \ \ \mu_{l, x} =  0,
\end{multline*}
with the simulation parameters  ranging as 
\begin{equation*}
\zeta \in \{0, 0.5, 1\}, \quad \varepsilon \in \{1.3, 1.5, 1.7, 1.9, 2.1\} \text{ and  } \eta = -2.
\end{equation*}
Since $w_{l, x}  = 0$, all the non-null statistics come from the right centered alternative density $\phi(z-\mu_{r, x})$. We briefly explain the simulation parameters.  $\varepsilon$ is an effect size parameter. Generally, $\zeta$ controls the \emph{informativeness} of the covariates in relation to both the non-null probabilities  and alternative means: when $\zeta >0$, a greater value of $x_{\bullet}$ makes the signals denser and stronger (i.e. $w_{r, x}$ and $\mu_{r, x}$ become larger). The value of $\eta$ controls the sparsity levels. For example, when the covariates are non-informative at $\zeta = 0$, setting $\eta = -2$ yields a baseline signal proportion of roughly $12\%$, i.e. $w_{r, x}  = w_{l, x} + w_{r, x}  =  11.9\%$.  Note that $f_{1, x}$ in \eqref{setup_cond_alt_density} varies in $x$ given the dependence of $\mu_{r, x}$ on $x$, so  \eqref{FDRregAssumption} is an invalid assumption.


\medskip

\item {\bf Setup 2} (Unbalanced covariate effects on the non-null proportions). Let
\small{\begin{multline*}
 w_{r,x} =  \frac{\exp(  \zeta x_{\bullet})}{\exp( -\eta ) +  \exp( - \zeta x_{\bullet}) + \exp(\zeta x_{\bullet} )}, \; w_{l,x} =  \frac{\exp(  - \zeta x_{\bullet})}{\exp( -\eta ) +  \exp( - \zeta x_{\bullet} ) + \exp(\zeta x_{\bullet})},  
\end{multline*}}
\normalsize
$\mu_{rx} = \varepsilon$ and $\mu_{lx} = - \varepsilon$. We fix $\eta = -2.5$ and vary other parameters in the range
\[
\zeta \in \{0, 0.7, 1\}, \quad \varepsilon \in \{1.3, 1.5, 1.7, 1.9, 2.1\}.
\]
Only $w_{l,x}$  and $w_{r,x}$ depend on the covariate value: for $\zeta > 0$, $w_{r,x}$ increases and $w_{l,x}$ decreases as $x_{\bullet}$ increases, and vice versa as $x_{\bullet}$ decreases. In consideration of \eqref{setup_cond_alt_density}, the conditional non-null density $f_{1, x}(z)$ will change sharply in shape from concentrating on negative $z$-values to concentrating on positive $z$-values as $x_{\bullet}$ increases from being negative to positive. This relationship provides important structural information which can be leveraged for enhancing the power. However, if one collapses the $z$-values into two-sided $p$-values, then the analogous conditional $p$-value density is less likely to capture drastic changes in  $x_{\bullet}$, since both very negative and positive $x_{\bullet}$ can correspond to very small $p$-values, making the interactive relationship between the $p$-values and the covariates less pronounced. Intuitively, this would lead to power loss of $p$-value based methods. The choice of $\eta$ corresponds to a baseline signal proportion of roughly $14\%$ when $\zeta = 0$.
 
\medskip

 \item {\bf Setup 3} (Unbalanced covariate effects on the alternative means). Let
{\small
\begin{multline*}
 w_{r, x} =  \frac{1/2}{1 + \exp(-   \eta )},  w_{l, x} =  \frac{1/2}{1 + \exp(-   \eta)},
\mu_{r, x} =   \frac{2\varepsilon}{1 + \exp ( - 	\zeta x_{\bullet} )}, \mu_{l, x} =  \frac{- 2\varepsilon}{1 + \exp ( 	\zeta x_{\bullet} )}.
\end{multline*}
}
The simulation parameters range as
\[
\zeta \in \{0, 1.5, 3\}, \quad \varepsilon \in \{1.3, 1.5, 1.7, 1.9, 2.1\} \text{ and } \eta = -2.
\]
Our choice of $\eta$ corresponds to the a baseline signal proportion of roughly  $12\%$ when $\zeta = 0$. 
When the covariates are informative ($\zeta > 0$),   $\mu_{r, x}$ and $\mu_{l, x}$ respectively become more positive and less negative as $x_{\bullet}$ increases. Such a directional relationship can be potentially exploited by ZAP for power improvement. However, if one collapses the $z$-values into $p$-values, then under $H_i = 1$ both very positive and negative values of  $X_{\bullet i}$ can imply a small $P_i$, and the interactive relationship between the main statistic $P_i$ and auxiliary statistic $X_i$ will be much weakened.  
%

\end{enumerate}

\begin{figure}[t]
\centering
\includegraphics[width=\textwidth]{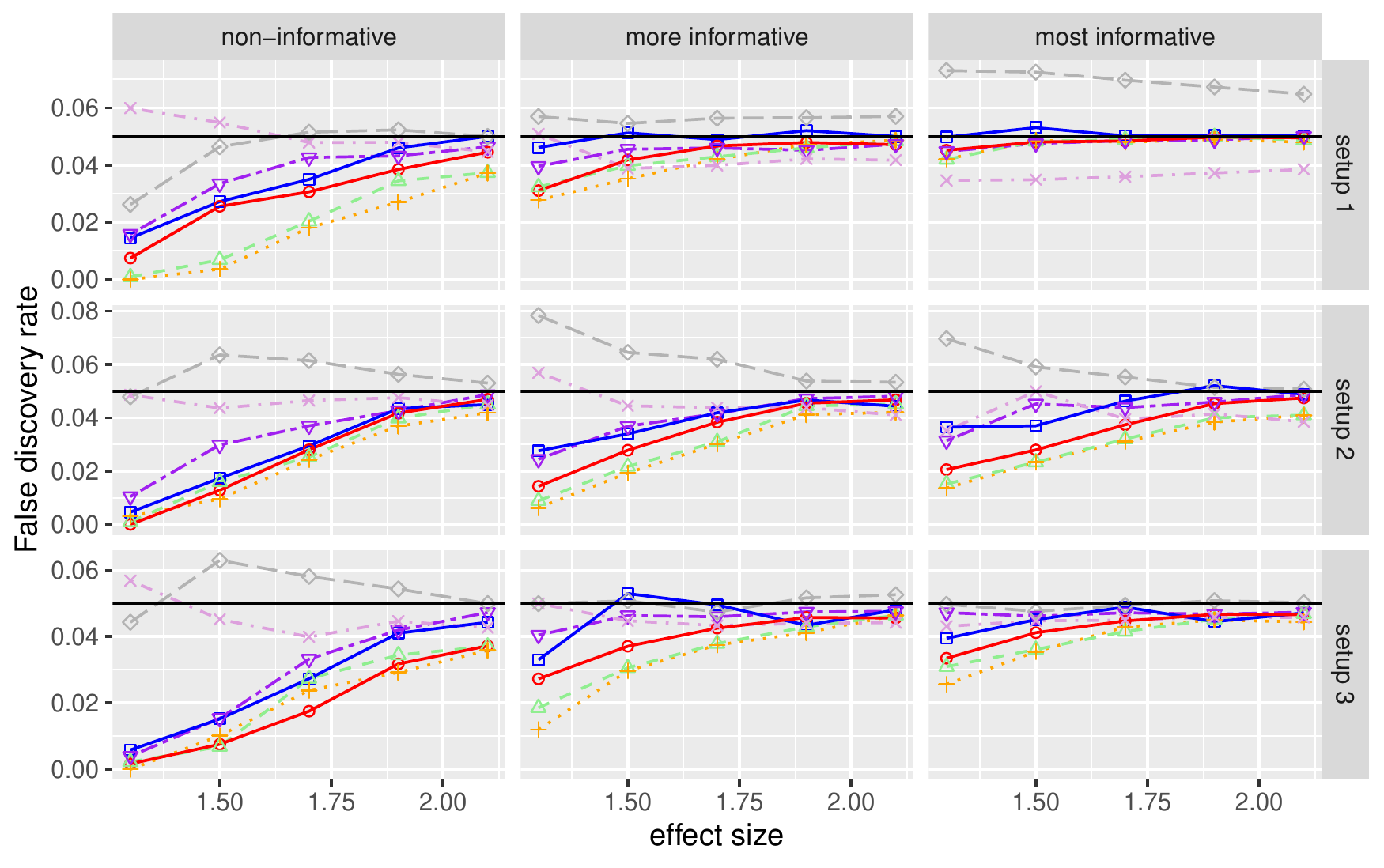}
\includegraphics[width=\textwidth]{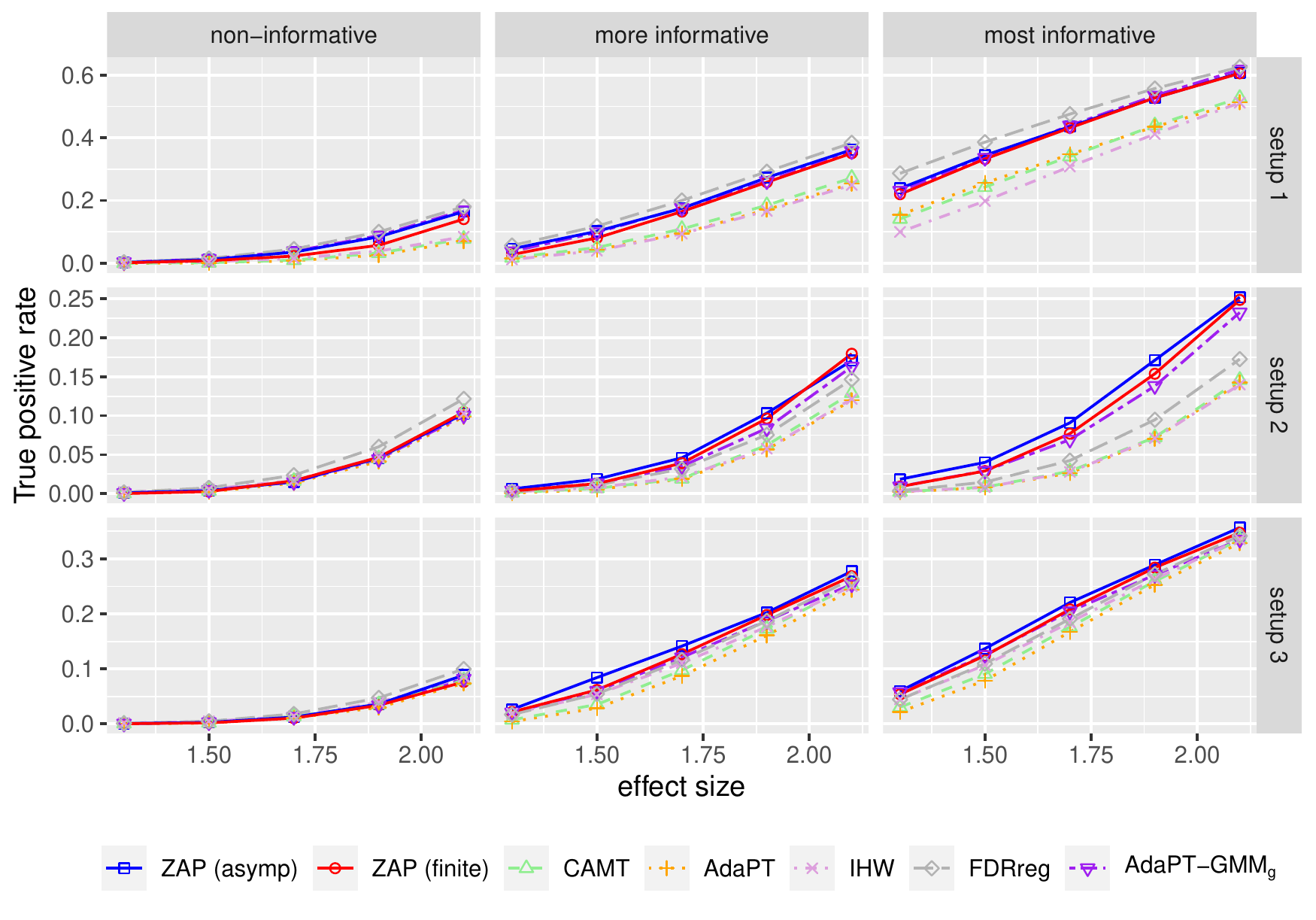}
\caption{
FDR and TPR performances of different methods under {\bf Setup 1 - 3} . All methods are applied at a targeted FDR level of $0.05$ (the horizontal black lines). The x-axes show the values of $\varepsilon$. \texttt{non-informative}, \texttt{more informative} and \texttt{most informative} correspond to different values of $\zeta$ from the smallest to the largest.
 }
\label{fig:all_setup}
\end{figure}

We apply the seven methods   at the nominal FDR level 0.05.  Since IHW can only handle univariate  covariates, it is applied with $X_{\bullet i}$, which is an effective summary covariate in all three setups. 
The simulation results are reported in \figref{all_setup}, where the empirical FDR and TPR levels of different methods are computed based on 150 repetitions. The following observations can be made:

\begin{enumerate}[(a)]

\item Asymptotic ZAP, depicted in blue, is in general more powerful than finite-sample ZAP, depicted in red. This is likely attributable to the latter's information loss from the ``$u$-value masking'' step, which can be perceived as a necessary trade-off for the strong finite-sample FDR controlling property. Even so, the asymptotic ZAP  has its FDR controlled under or around the nominal level $0.05$.

\item Both asymptotic and finite-sample ZAP methods  demonstrate superior performances over the $p$-value based methods (CAMT, AdaPT and IHW). The gains in power become more substantial when the covariates become more informative. The $z$-value based method AdaPT-$\text{GMM}_g$ has quite comparable power to the finite-sample ZAP in these setups too, and we will provide more discussion on the comparison between the them in \secref{discuss}.

\item The covariate-adjusted non-null density \eqref{setup_cond_alt_density} depends on $x$ for all three setups, so FDRreg, which makes the conflicting assumption in \eqref{FDRregAssumption}, is possibly invalid for FDR control. This is indeed observed for a number of settings. Moreover, it can't match the power of ZAP in Setup 2, likely because the assumption $f_{1,x}(z)\equiv f_1(z)$ itself obstructs the interactive information between the $z$-values and the covariates to be utilized.

\item In Setup 3, ZAP only has slightly discernible power advantage over the other methods when the covariates are informative. In fact, no current FDR methods can demonstrate near-optimal power under this setup, as shown in \appref{more_simulation_results}.  In \secref{discuss}, we discuss possible future work that might address this.
\end{enumerate}

\subsection{Real data} \label{sec:real_data}

This section investigates the performance of ZAP  using several publicly available real datasets  summarized in \tabref{DatTab}.  Three datasets (\texttt{bottomly}, \texttt{airway}, \texttt{hippo}) are generated by RNA sequencing (RNA-Seq) experiments for detecting differential expressions in transcriptomes, where the primary statistic $Z_i$ measures the observed difference in the expression level of a gene under two experimental conditions. Meanwhile, an auxiliary covariate, the average normalized read count for each gene, is collected alongside the primary data. The datasets \texttt{bottomly} and \texttt{airway} have been analyzed by the works of \citet{ignatiadis2016data, lei2016adapt, zhang2020covariate} with the methods IHW, AdaPT and CAMT respectively.  The more recent data set \texttt{hippo} \citep{harris2019hippocampal} is generated by the cutting-edge single-cell RNA (sc-RNA) sequencing technology to study differential expressions in mouse  hippocampus. For all datasets above, we have adopted the standard data pre-processing step, which filters out genes with excessively low read counts across samples before further downstream analyses \citep{chen2016reads} such as model fitting and multiple testing. This is a common practice among bioinformaticians for a number of reasons; see \appref{rnaseq_processing} for more discussion. The fourth dataset is based on the experiments in \citet{smith2008spatial} and \citet{kelly2010local}, where each $Z_i$ is a  normalized test statistic  that, for a given pair of neurons  in the primary visual cortex, measures how synchronous  their spike trains are, and \citet{scott2015false} has applied FDRreg to it for detecting neural interactions. Correspondingly, each such hypothesis has two covariates: the distance and the correlation of the ``tuning curves" between the two activated neurons. We have named this dataset \texttt{scott} for short.

 \begin{table}[t]
    \begin{tabular}{ | c |  c | p{10cm} |}
    \hline
     Name & \# tests&\hfil  Brief description \\ \hline
   \texttt{bottomly}& 11484 & DE in striatum for  the two mouse strains C57BL/6J(B6) and DBA/2J(D2); bulk  RNA-seq \citep{bottomly2011evaluating}. \\ \hline
 \texttt{airway} & 20941 & DE in human airway smooth muscle cell lines in response to dexamethasone; bulk RNA-seq \citep{himes2014rna}.\\ \hline
   \texttt{hippo} & 15000&  DE in mouse hippocampus in response to enzymatic dissociation in comparison to standard tissue homogenization; scRNA-seq \citep{harris2019hippocampal}. \\ \hline
\texttt{scott}&   7004 &  Synchronous firing of pairs of neurons, based on neuron recordings in the primary visual cortex of an anesthetized monkey in response to visual stimuli \citep{scott2015false}. \\ 
    \hline   
    \end{tabular}
    \caption{Description of four real datasets. ``\# tests" shows the number of tests for each dataset after any necessary data pre-processing.  DE = Differential Expression.   } 
  \label{tab:DatTab}
 \end{table} 
\begin{figure}[t]
\centering
\includegraphics[width=\textwidth]{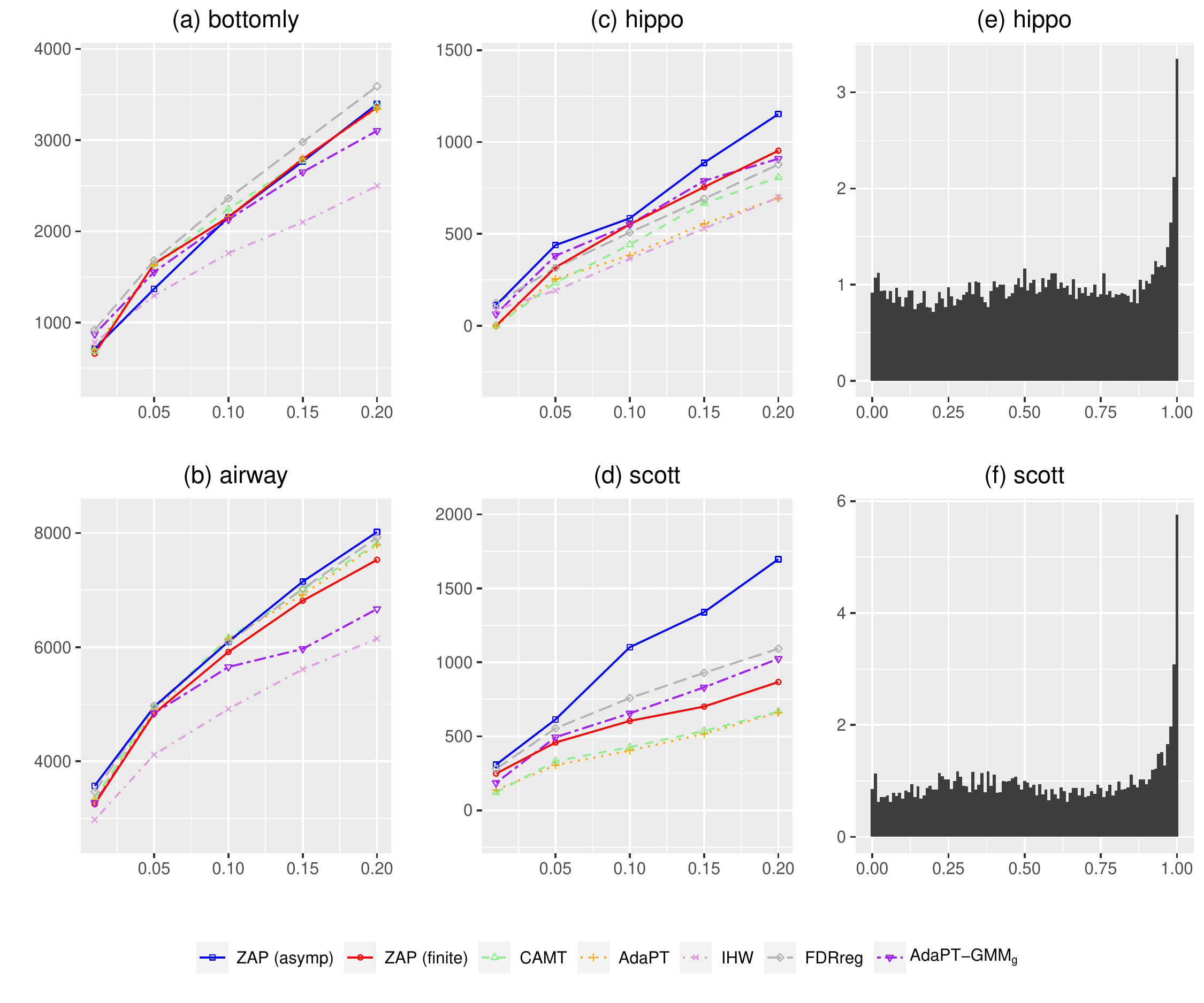}
\caption{{$\bf (a)-(d)$} plot the numbers of rejections for different methods across datasets, against targeted FDR level at $0.01, 0.05, 0.1, 0.15, 0.2$. $\bf (e)$ and $\bf (f)$ are respectively the histograms of the ``$u$-values" for the \texttt{hippo} and \texttt{scott} datasets.
 }
\label{fig:real_data_plot}
\end{figure}

%

For the RNA-seq datasets, all the methods that accommodate multivariate covariates (CAMT, AdaPT, FDRreg, AdaPT-$\text{GMM}_g$ and the two methods of ZAP) are applied with the log mean normalized read count expanded by a natural cubic spline basis with 4 interior knots, using the \texttt{ns} function in the \texttt{R} package \texttt{splines} (with its \texttt{df} argument set to $6$), and  IHW is applied with the original log mean normalized read count as it can only handle a univariate covariate. As there are two covariates for the neural dataset \texttt{scott},  IHW is not applied, and following what was done in \citet{scott2015false}, the multivariate methods  are applied with each of the two covariates expanded by  a B-spline basis using the \texttt{bs} function in \texttt{R} with its argument \texttt{df} set to 3,  which results in six expanded covariates in total.

The number of rejections for the various methods are shown in \figref{real_data_plot}$(a)$-$(d)$, and ZAP has attained top power performances in general. For the dataset \texttt{bottomly}, FDRreg shows moderately more rejections than other methods, but the power gain may be due to overflow in FDR resulting from making the assumption in \eqref{FDRregAssumption}. For the dataset \texttt{airway}, it is surprising to see the power of AdaPT-$\text{GMM}_g$ drops quite a bit beyond the target FDR level $0.1$.  The histogram plots in \figref{real_data_histograms_unfiltered} of \appref{real_data_append} show that the $u$-values are almost symmetrically distributed for these two datasets, which suggests that $p$-value and $z$-value based methods tend to have comparable power, unless the reduction to $p$-values fails to capture the interactive information between the $z$-values and covariates. For the datasets \texttt{hippo} and \texttt{scott}, the histograms, which are shown in \figref{real_data_plot}$(e)$ and $(f)$, show that the $u$-value distribution is asymmetric. This can explain why the $z$-value based methods
exhibits conspicuous power improvement over the $p$-value methods in \figref{real_data_plot} $(c)$ and $(d)$, which are in agreement with the intuition that the former are more capable of exploiting the distributional asymmetry. Similar to what we observed in the simulation studies, the asymptotic ZAP tends to reject more hypotheses than the finite-sample ZAP, whose FDR controlling validity is based on fewer assumptions (\thmref{ZAPfiniteControl}). For the \texttt{scott} dataset, the increase in power of the finite-sample ZAP trails off a bit beyond the target FDR level of $0.15$. Overall, the real data analyses affirm that $z$-value based approaches to covariate-adaptive testing promise to better exploit the full data $\{Z_i, X_i\}_{i = 1}^m$ to boost testing power. In \secref{discuss}, we will offer a general comment on the power of data masking procedures.

\section{Discussion} \label{sec:discuss}

We have introduced ZAP, which is a $z$-value based covariate-adaptive testing framework  that offers control of the FDR under minimal assumptions. The  main thrust  of our proposal is to avoid the common data reduction step of forming two-sided $p$-values used by most other covariate-adaptive methods in the recent literature, so as to preserve the intact structural information in the data as a starting point to devise more powerful procedures. While there is no ``one-size-fits-all" solution to all FDR analysis problems, as is seen in the extensive simulation studies in the recent paper of {\citet{korthauer2019practical}, we believe  the current form of ZAP is a competitive choice  in many covariate-adaptive testing situations.

Although  the finite-sample ZAP offers exact FDR control (\thmref{ZAPfiniteControl}}), the requisite data masking machinery incurs some information loss as a trade-off. This often injects extra instability into its power performance; for one thing,  a working model estimated based on the partially masked data does not necessarily fit the original full data as well, particularly when it is  flexible enough to overfit the partially masked data. On the other hand, a reasonably flexible working model is needed to better capture the true underlying data generating mechanism, which will be translated into power. Prior works on data masking have also mentioned similar power-instability issues; see \citet[Appendix A.3]{lei2016adapt} and \citet[Section 3.2]{lei2017star} for instance. 
Hence, when the number of hypotheses $m$  is large enough (e.g. $m \geq 5000$) to justify the large-sample theory, 
we generally advocate the use of the asymptotic version of ZAP.

We have chosen to operate on the $u$-value scale, leveraging a covariate-dependent beta mixture as our working model. A natural alternative choice is to employ some form of Gaussian mixture models defined on the original $z$-value scale, which is also popular among researchers for performing FDR testing in different applications \citep{mclachlan2006simple, nguyen2018false}. Regardless of the model choice, 
since the asymptotic ZAP aims to directly mimic the optimal oracle procedure, for it to be powerful  the assessor function $a_x(\cdot)$ should be a reasonably  accurate proxy for the true conditional local false discovery rate function $\text{CLfdr}_x(\cdot)$. We have had success in this regard with  our carefully calibrated beta-mixture model; as explained in  \secref{beta_expr}, setting $(\gamma_l, \gamma_r) = (4, 4)$ generally renders reasonably close estimates for the (non-)null probabilities, which are critical for constructing a more accurate assessor function with the formula in \eqref{assessor_form} with respect to the working model. In unreported simulations,  we have also experimented with the normal mixture model, but with discouraging results; namely, it is very hard to consistently obtain good estimates of the mixing probabilities, even with the full  data. This is not a surprise since Gaussian deconvolution is known to be a very difficult problem; see \citet{fan1991optimal} for a theoretical discussion, or  \citet[Section 2.1]{scott2015false} for an illustration with a simple simulation setup.

\citet{chao2021adapt} appeared at around the same time as our present paper. (We thank an anonymous referee for pointing out this work.) Being most comparable to our finite-sample ZAP method, their AdaPT-$\text{GMM}_g$ is  a $z$-value based data masking algorithm that does leverage a working Gaussian mixture density with $K$ components of the form
\begin{equation} \label{Chao_Gauss_mixture}
Z_i | X_i = x\sim  \sum_{k = 1}^K\pi_{k, x} \phi\left(\frac{z - \mu_k}{\sigma_k}\right)
\end{equation}
  to iteratively reveal the masked $z$-values, where each $\pi_{k, x}$ is a component probability depending on $x$ based on a model of choice (e.g. a neural network), and $\mu_k$ and $\sigma_k$ are component mean and standard deviation to be estimated. Unlike our current version of the finite-sample ZAP which uses estimated $\text{CLfdr}_i$'s as the measuring scores  to decide which one of the masked $i$'s is the least promising and  to be revealed in the next step (\algref{updateThreshold} in the \appref{EM_finite_update}),  AdaPT-$\text{GMM}_g$, in our style of notation, uses conditional probabilities of the type
\begin{equation} \label{chao_q}
P\Bigl(|\widecheck{Z}_i| > |Z_i| \;\ \Big|   \;\ X_i, \{ Z_i, \widecheck{Z}_i\} \Bigr)
\end{equation}
as scores to compare the masked $i$'s, where $\{ Z_i, \widecheck{Z}_i\}$ is a masked two-element set containing the value of $Z_i$ and its reflection that is defined similarly to $\{ U_i , \widecheck{U}_i\}$ in \eqref{U_tilde} but on the $z$-value scale. Intuitively, a masked $i$ with the largest such conditional probability is considered the least promising to contain a signal and should be revealed in the next step. Here, a Gaussian mixture model is still a sensible working model to use because  the quantity in \eqref{chao_q} can be reasonably estimated without relying on accurate deconvolution; see \citet{chao2021adapt} for more details. In this vein, it will be interesting to explore  a mixture-of-experts working model \citep{chamroukhi2019regularized, nguyen2016universal}, which is essentially \eqref{Chao_Gauss_mixture} but with the means $\mu_k$ also depending on covariates,  particularly for a situation like  Setup 3 in \secref{sim_data} where the covariate is informative via influencing the alternative means. We leave this to future research that may be  opportune in other instances.

\bibliographystyle{ba}
\bibliography{beta3}

\clearpage
\appendix 
%

\section{Oracle procedures}

\subsection{Optimality of $\pmb\delta^\mathcal{P}$ and $\pmb\delta^\mathcal{Z}$} \label{app:oracleOptim}

In this section we briefly review  the optimality properties of the oracle procedures $\pmb\delta^\mathcal{P}$ and $\pmb\delta^\mathcal{Z}$ in \secref{infoLoss} and how their thresholds $t_\mathcal{P}$ and $t_\mathcal{Z}$ are determined. To streamline the discussion, we will let $\{M_i\}_{i =1}^m$ denote the set of \emph{main statistics}, where it can either be that $M_i = P_i$ or  $M_i = Z_i$ for all $i$, depending on whether $p$-value or $z$-value based methods are considered. Given the data $\{M_i, X_i\}_{i=1}^m$, it is well-known that \emph{optimal} procedures, which aim to maximize true discoveries subject to false discovery constraints, should operate by rejecting $i$ if its corresponding posterior probability $P(H_i =0| M_i, X_i)$ falls below a data-dependent threshold $t_\mathcal{M}$.  We will use $\pmb\delta^\mathcal{M}$ (in a similar way as  $\pmb\delta^\mathcal{P}$ or $\pmb\delta^\mathcal{Z}$) to denote the procedure that thresholds the quantities  $\{P(H_i =0| M_i, X_i)\}_{i=1}^m$ with $t_\mathcal{M}$.

There are subtly different ways to define ``optimality", depending on the particular false discovery (e.g. FDR, mFDR, pFDR) and power (e.g. TPR, ETD, mFNR) measures used, which may  lead to different ways of setting $t_\mathcal{M}$. A most recent result of  \citet[Theorem 3.1]{heller2021optimal} suggests that among all the testing procedures  that are functions of $\{M_i, X_i\}_{i=1}^m$, the $t_\mathcal{M}$ that renders an ETD-maximing $\pmb\delta^\mathcal{M}$ with $FDR \leq \alpha$ can be  found by solving an integer optimization problem \citep[Theorem 3.1]{heller2021optimal}. For our purpose, by letting $L_{(1)} < \dots <L_{(m)}$ be the order statistics of the posterior probabilities $\{P(H_i= 0| M_i, X_i)\}_{i =1}^m$,
we have considered the computationally simpler optimal procedure first proposed in \citet{sun2007oracle} which takes $t_\mathcal{M} = L_{(j)}$, 
where 
  \begin{equation}\label{oraclej}
j \equiv \max \left\{ i' \in \{1, \dots, m\} :  \frac{\sum_{i = 1}^{i'}  L_{(i)} }{i'}  \leq \alpha \right\}.
\end{equation}
This procedure has $FDR \leq \alpha$ because for any procedure that produces a rejection set $\mathcal{R}$ based on $\{M_i, X_i\}_{i=1}^m$, its FDR can be written as  
  \begin{equation*} \label{genericFDRexpr}
FDR
= \mathbb{E}\left[ 
\underbrace{\mathbb{E}\left[  \frac{\sum_{i \in \mathcal{R}} (1 - H_i) }{|\mathcal{R} | \vee 1} \Big| \{M_i, X_i\}_{i=1}^m\right]}_{\mathbb{E} [ FDP | \{M_i, X_i\}_{i=1}^m]}
 \right] 
 = \mathbb{E} \left[ \frac{ \sum_{ i \in \mathcal{R}} P(H_i =0| M_i, X_i) }{  |\mathcal{R}|\vee 1} \right].
\end{equation*}
Conditional on any  instance of the  data $\{M_i, X_i\}_{i=1}^m$,   $\pmb\delta^\mathcal{M}$  prioritizes rejections of the hypotheses that are least likely to be true nulls, all the while controlling the conditional FDR
 $
 \mathbb{E}[ FDP |\{M_i, X_i\}_{i=1}^m ] 
 $
   below $\alpha$ by setting $t_\mathcal{M} = L_{(j)}$, as the ratio $ \sum_{i = 1}^{i'}  L_{(i)}/i'$ in \eqref{oraclej} is precisely  the conditional FDR of rejecting the most promising $i'$  hypotheses. As a result, its controls the FDR under $\alpha$ since the conditional version $\mathbb{E}[ FDP | \{M_i, X_i\}_{i=1}^m]$ is always not larger than $\alpha$.

We now give a more precise account  of the optimal property of the prior procedure. Another popular measure of type 1 errors  is the marginal FDR (mFDR), which for any rejection set  $\mathcal{R}$ is the ratio
\[
mFDR \equiv \frac{\mathbb{E}[V]}{\mathbb{E}[R]},
\]
where  $V$ and $R$ are defined  as in the main text. For $\alpha \in (0,1)$, it is known that among all the procedures based on $\{M_i, X_i\}_{i=1}^m$ with $mFDR\leq \alpha$, the ETD-maximizing procedure is the one given by $\pmb\delta^\mathcal{M}$ that sets $t_\mathcal{M}=\lambda^*_\mathcal{M}$, where
\begin{equation} \label{mFDR_fix_threshold}
\lambda^*_\mathcal{M} \equiv \sup\left\{ \lambda \in (0, 1] : 
\frac{ \sum_{i} \mathbb{E}[(1 -  H_i) I( P(H_i= 0| M_i, X_i) \leq \lambda ) ]}
{\sum_i  \mathbb{E}[ I(P(H_i= 0| M_i, X_i) \leq \lambda )]} \leq \alpha
 \right\}; 
\end{equation}
see \citet{sun2007oracle} and \citet{heller2021optimal}. In practice, since $\lambda^*_\mathcal{M}$ could be tricky to obtain even with oracle knowledge, this optimal procedure  for mFDR control is often approximated by our computationally handy version with $t_\mathcal{M}=L_{(j)}$ above when $m$ is large, as is the case with most FDR analyses. Their asymptotic equivalence can be shown by standard arguments, such as those in \citet[Section 4]{sun2007oracle}. The aforementioned references provide  a more detailed exposition. 

\subsection{Rejection regions for Examples 2.1-2.3}\label{app:rejRegion}

Consider the p-value conditional mixture density 
  \begin{equation} \label{mixtureDensP}
 P_i |X_i  = x\sim g_x (p) \equiv g(p|x) = (1  - w_x) g_0 (p) +w_x g_{1, x} (p) , 
 \end{equation}
  induced by \eqref{mixtureDens}, where $g_0 \equiv 1$ is the uniform null density of $P_i$, and $g_{1, x} (p)\equiv g(p|H_i = 1, X_i = x)$ is the conditional alternative density of $P_i$. 
 The rejection regions $\mathcal{S}^\mathcal{P}(x)$ and $\mathcal{S}^\mathcal{Z}(x)$ in \figref{oracle_rej_reg}  are derived based on the threshold $\lambda^*_\mathcal{M}$ in  \eqref{mFDR_fix_threshold}, where 
 \[
\mathcal{S}^\mathcal{Z} (x) \equiv \left\{z : P(H_i = 0|Z_i = z, X_i = x) \leq \lambda^*_\mathcal{Z} \right \} = \left\{z: \frac{w_x f_{1, x}(z)}{(1 - w_x) f_0(z)} \leq \frac{1 - \lambda^*_\mathcal{Z}}{  \lambda^*_\mathcal{Z}}\right\}
\]
and 
\[
\mathcal{S}^\mathcal{P}(x) \equiv\left\{z : P(H_i = 0|P_i = p, X_i = x) \leq  \lambda^*_\mathcal{P} \right \} = \left\{z: \frac{w_x g_{1, x}(2\Phi(- |z|))}{(1 - w_x) g_0( 2\Phi(- |z|))} \leq \frac{1 -  \lambda^*_\mathcal{P} }{  \lambda^*_\mathcal{P} }\right\}.
\]
 $\lambda^*_\mathcal{P} $ and $\lambda^*_\mathcal{Z} $ are  $\lambda^*_\mathcal{M} $ defined with $M_i = P_i$ and $M_i = Z_i$ for all $i$ respectively. 
Since the examples are relatively simple, these regions can be found by numerical means, and 
we will derive $\mathcal{S}^\mathcal{Z} (x)$ in \exref{S2} as an illustration: It is the set
\[
\left\{z:  \frac{0.2 \left[ \frac{1 - x}{2}  \exp( - \frac{(z + \mu)^2}{2}) + \frac{1 + x}{2}  \exp(- \frac{(z - \mu)^2}{2}) \right]}{ 0.8 \exp( - \frac{z^2}{2})}\leq \frac{1 - \lambda^*_{\mathcal{Z}}}{\lambda^*_{\mathcal{Z}}} \right\}, 
\]
where $\mu \equiv 1.5$.   By setting $\frac{w_x f_{1, x}(z)}{(1 - w_x) f_0(z)} =  \frac{1 - \lambda^*_\mathcal{Z}}{  \lambda^*_\mathcal{Z}}$, one arrives at the equation
\[
\frac{1  + x}{2} \exp( 2 \mu z) - 4\frac{1 - \lambda^*_\mathcal{Z}}{ \lambda^*_\mathcal{Z}} \exp\left(\frac{\mu^2}{2}\right) \exp(\mu z) + \frac{1 - x}{2} = 0.
\]
in $z$. 
To solve for a solution $z^*$, we can apply the formula for the solutions of a quadratic equation to get
\[
\exp( \mu z^* ) = \frac{4 \frac{1 - \lambda_\mathcal{Z}^*}{\lambda_\mathcal{Z}^*} \exp(\frac{\mu^2}{2}) \pm
\sqrt{  16 \left(\frac{1 - \lambda_\mathcal{Z}^*}{\lambda_\mathcal{Z}^*}\right)^2 \exp(\mu^2 )- (1 - x) (1+ x) } }{1 + x},
\]
which in turn implies the two boundary points
\[
z^* = \frac{1}{\mu} \log \left\{   \frac{4 \frac{1 - \lambda_\mathcal{Z}^*}{\lambda_\mathcal{Z}^*} \exp(\frac{\mu^2}{2}) \pm
\sqrt{  16 \left(\frac{1 - \lambda_\mathcal{Z}^*}{\lambda_\mathcal{Z}^*}\right)^2 \exp(\mu^2 )- (1 - x) (1+ x) } }{1 + x} \right \}
\]
for the red regions in the middle panel of  \figref{oracle_rej_reg}  as a function of $x$.

We now explain why $\mathcal{S}^\mathcal{P}(x)$ doesn't change with $x$  in Examples 2.2 and 2.3. From the conditional $p$-value density \eqref{mixtureDensP}, one can see that 
\[
P(H_i = 0| P_i, X_i) = \frac{(1 - w_{X_i}) g_0(P_i)}{g_{X_i} (P_i) } = \frac{(1 - w_{X_i}) g_0(P_i)}{(1 - w_{X_i}) g_0(P_i) + w_{X_i} g_{1, X_i} (P_i) }.
\]
Hence, if $w_x$ and $g_{1,x}$ do not depend on $x$, it is apparent that $\mathcal{S}_x^\mathcal{P}$ will not vary in $x$. Simple calculations can show that $w_x = 0.2$ and $w_x = 0.1$ for  Examples 2.2 and 2.3 respectively, and 
\[
g_{1, x}(p)= \frac{\phi( \Phi^{-1}( p/2) - 1.5) + \phi( -\Phi^{-1}( p/2) -1.5) }{2 \phi( -\Phi^{-1}( p/2))}  = \frac{\phi(-|z| - 1.5) + \phi( |z| - 1.5) }{2 \phi(|z|)}
\]
for both examples; all of these quantities do not depend on $x$.

\section{Proof for the prototype method} \label{app:Pf_ideal}


We will prove the FDR validity of the prototype procedure in \secref{assessor}. The false discovery proportion of the  prototype testing procedure, which thresholds the test statistics $T_i$'s with the threshold  $\hat{t}_\alpha \equiv \hat{t}(\alpha)$, can be written as
\begin{align*}
FDP &= \frac{  \#\{ i \text{ null: } T_i  \leq \hat{t}_\alpha \}}{1  \vee \#\{  T_i  \leq \hat{t}_\alpha\}} \\
& = \frac{ \#\{ i \text{ null: } T_i  \leq \hat{t}_\alpha \} }{1 +  \#\{ i \text{ null: } S_i  \geq 1 - c_i(\hat{t}_\alpha) \} }
\underbrace{\frac{ 1 +  \#\{ i \text{ null: } S_i  \geq 1 -  c_i( \hat{t}_\alpha) \} }{1  \vee \#\{  T_i  \leq \hat{t}_\alpha\} } }_{ \leq \alpha \text{ by the defintion of our procedure}}\\
&\leq \alpha  \frac{ \#\{ i \text{ null: } T_i \leq \hat{t}_\alpha \} }{1 +  \#\{ i \text{ null: } S_i  \geq 1 -  c_i( \hat{t}_\alpha) \} }. 
\end{align*}
We only have to show that 
$
\mathbb{E}\left[
 \frac{ \#\{ i \text{ null: } T_i \leq \hat{t}_\alpha \} }{1 +  \#\{ i \text{ null: } S_i  \geq 1 -  c_i( \hat{t}_\alpha) \} }
\right]
$
is bounded by 1  using the stopping time argument from \citep{barber2019knockoff}.

Without loss of generality we will assume the true nulls are the first $m_0$ hypotheses.  For each $i \in \{1, \dots, m_0\}$, define
\[
\widecheck{S}_i = \begin{dcases*}
        S_i  & when $S_i \leq 0.5$ \\
        1 - S_i & when $S_i >  0.5$
        \end{dcases*}
\]
and  $\widecheck{T}_i \equiv c_i^{-1}(\widecheck{S}_i )$. Using the order statistics $\widecheck{T}_{(1)} \leq \dots \leq \widecheck{T}_{(m_0)}$ of $\widecheck{T}_1, \dots, \widecheck{T}_{m_0}$, 
we moreover let $B_i \equiv I( S_{(i)} >   0.5 )$ for $i =1, \dots, m_0$, where the order of $S_{(i)}$'s here is inherited from the order of the $\widecheck{T}_{(i)}$'s, rather than the magnitudes of the $S_i$'s themselves. Let $1\leq J \leq m_0$ be the index such that 
\[
\widecheck{T}_{(1)} \leq \dots \leq   \widecheck{T}_{(J)} \leq \hat{t}_\alpha <  \widecheck{T}_{(J+1)} \leq \dots \leq \widecheck{T}_{(m_0)},
\]
 we then have 
 \begin{multline*}
 \frac{ \#\{ i \text{ null: } T_i  \leq \hat{t}_\alpha \} }{1 +  \#\{ i \text{ null: } S_i  \geq 1 - c_i(\hat{t}_\alpha) \} } 
 =\frac{ \#\{ i \text{ null: } T_i  \leq \hat{t}_\alpha \} }{1 +  \#\{ i \text{ null: }   c_i^{-1}(1 - S_i) \leq    \hat{t}_\alpha  \} }  \\
 = \frac{(1 - B_1) + \dots + (1 - B_J)}{1 + B_1 + \dots + B_J} =  \frac{1 + J}{1 + B_1 + \dots + B_J}- 1, 
 \end{multline*}
 considering that $\hat{t}_\alpha$ must be less than  $t_{\max} := \max\{t:  c_i(t) \leq 0.5 \text{ for all }  i  \}$.
Hence it amounts to showing $E[\frac{1 + J}{1 + B_1 + \dots + B_J}] \leq  2$. This final step can be shown by applying \citet[Lemma 1]{barber2019knockoff}, since \emph{conditional on} \begin{inparaenum}\item $\widecheck{T}_{(1)} \dots \widecheck{T}_{(m_0)}$ and  \item $\{T_i : i \text{ is non null}\}$\end{inparaenum}, $B_1, \dots, B_{m_0}$ are independent Bernoulli$(0.5)$ random variables, and 
$J$ can be seen as a stopping time in reverse time with respect to the filtrations $\{\mathcal{F}_j\}_{j = 1}^{m_0}$, where 
$
\mathcal{F}_j \equiv \{B_1 + \dots + B_j, B_{j+1}, \dots, B_{m_0}\}
$.

\section{Properties of the working model} \label{app:workModProp}

In this section we will develop some properties of the beta-mixture model in \secref{betaMod} and the assessor functions it induces. To simplify notation, we will use $\pi_{li}, \pi_{ri}, k_{li}, k_{ri}, h_{li}, h_{ri}$ to respectively denote the quantities and functions $\pi_{l, X_i}$, $\pi_{r, X_i}$, $k_{l, X_i}$, $k_{r, X_i}$, $h_{l, X_i}$, $h_{r, X_i}$  from Model \eqref{hu-beta} when the observed  covariate $X_i$ is used, where the underlying parameters $\{\theta_l, \theta_r, \beta_l, \beta_r\}$ are unspecified but common for all $i = 1, \dots, m$. Likewise, we also use 
\begin{equation} \label{accessor_i}
a_i(u) \equiv  a_{X_i}(u)= \frac{1 - \pi_{li} - \pi_{ri}}{(1 -  \pi_{li} - \pi_{ri}) + \pi_{li}h_{li} (u) + \pi_{ri} h_{ri} (u)  }
\end{equation}
to denote the assessor function constructed with them, and  $T_i$, $S_i$ and $c_i(\cdot)$ will  denote the test statistics and null distribution function based on $a_i(\cdot)$ as in \secref{assessor}.

First, the following lemma states properties concerning the left and right alternative functions $h_{li}$ and $h_{ri}$.

\begin{lemma}[Properties of the non-null component densities]  \label{lem:convex}
For $\gamma_l > 2$, $h_{li}(\cdot)$ is a strictly convex and strictly decreasing function with the properties
\[
\lim_{u \rightarrow 0}h_{li}(u) = \infty\text{ and }\lim_{u \rightarrow 1}h_{li}(u) = 0.\]
 Similarly, for $\gamma_r > 2$, $h_{r_i}(\cdot)$ is a strictly convex and strictly increasing function with the properties
\[\lim_{u \rightarrow 0}h_{ri}(u) = 0 \text{ and } \lim_{u \rightarrow 1}h_{ri}(u) = \infty.\]

\end{lemma}
\begin{proof}[Proof of \lemref{convex}]

It suffices to show the facts for $h_{li}$ since  those for $h_{ri}$ can be proven exactly analogously. Recall that 
\[ 
h_{li}(u):= B(k_{li}, \gamma_l)^{-1}u^{k_{li} - 1} (1 - u )^{\gamma_l - 1},
 \]
where for brevity we have suppressed the dependence on $X_i$ in notations. Differentiating with respect to $u$ we get
\begin{equation} \label{hli_prime}
h_{li}'(u) = B(k_{li}, \gamma_l)^{-1} [(k_{li} - 1) u^{k_{li} - 2} (1 - u)^{\gamma_l - 1} - ({\gamma_l - 1}) (1 - u)^{\gamma_l - 2} u^{k_{li} - 1}],
\end{equation}
which, given $\gamma_l > 2$ (actually $\gamma_l > 1$ is sufficient),  can be seen to be always negative for any $u \in (0,1)$ and hence proves that $h_{li}$ is strictly decreasing. 
For convexity, we differentiate one more time to get
\begin{align*}
 h_{li}''(u) 
& = B(k_l, \gamma_l) (k_{li} - 1)  \left[ (k_{li} - 2)u^{k_{li} - 3} (1 - u)^{\gamma_l - 1}    -  u^{k_{li} - 2} (\gamma_l - 1) (1 - u)^{\gamma_l - 2}\right]\\
 &\quad  - B(k_l, \gamma_l)  (\gamma_l - 1)\left[  (2 - \gamma_l) (1 - u)^{\gamma_l-3} u^{k_{li} - 1} + (k_{li} - 1) u^{k_{li} - 2} (1 - u)^{\gamma_l-2}\right] \\
 &=B(k_l, \gamma_l)  \underbrace{ u^{k_{li}- 3} (1 - u)^{\gamma_{l}-3 } }_{ > 0 } \times \\
 &\qquad
 \left[\underbrace{(k_{li}-1)(k_{li}- 2)(1 - u)^2}_{ > 0 }   - \underbrace{ 2(k_{li} - 1)(\gamma_l - 1) u (1 - u)  }_{> 0 }+ (\gamma_l - 1) (\gamma_l - 2)u^2  
  \right],
\end{align*}
where in the last equality, the positive terms are positive since $0 < k_{li} < 1$. As such, $h_{li}''(u)$ is strictly positive for all $u \in (0, 1)$ as long as $\gamma_l > 2$, which proves the strict convexity of $h_{li}$. 
\end{proof}

A closer inspection of the  proof above will reveal that for $\gamma_l \in (1, 2)$, $h_{li}$ may not even be convex, and the same is  true for $h_{ri}$. Hence we have required that $\gamma_l, \gamma_r >2$ in our model.
To facilitate the proof in later sections we will also define the \emph{reciprocal} assessor function
\begin{equation} \label{assessor_recip}
b_i(u) \equiv 1/a_i(u).
\end{equation} 
By the properties of $h_{li}$ and $h_{ri}$ in \lemref{convex}, one can readily conclude the following lemma, which will help us develop some useful facts later:
\begin{lemma}[Properties of the reciprocal assessor] \label{lem:recipProp}
The reciprocal assessor function defined in \eqref{assessor_recip} (for $\gamma_l,\gamma_r > 2$) is strictly convex and smooth, with the property that
\begin{equation} \label{recipProp_display}
\lim_{u \rightarrow 0} b_i(u) = \lim_{u \rightarrow 1} b_i(u) = \infty. 
\end{equation}
Hence, there exists a unique minimal $\underline{u}_i$ such that
\[
 b_i(\underline{u}_i) <   b_i(u)  \text{ for all } u \in (0, 1). 
\] 
\end{lemma}

$T_i \equiv a_i(U_i)$ as a random variable has the range $(0, a_i(\underline{u}_i) ]$ in light of \lemref{recipProp}. By construction, $a_i(\cdot)$'s level sets can only be of Lebesgue measure 0, so $c_i(\cdot)$ is continuous under the uniform null distribution of $U_i$. With the strict convexity of $b_i(\cdot)$, one can also conclude that $c_i(\cdot)$, its null distribution function,  is invertible (or equivalently, strictly increasing), since  no interval in the range $(0, a_i(\underline{u}_i) ]$ will have zero measure under the law of  $a_i(U_i)$ induced by the uniform null distribution of $U_i$ by the intermediate value theorem.  The smooth ``bowl" shape of  $b_i(\cdot)$ also implies that, for any $t \in (0, a_i(\underline{u}_i) ]$, the event $\{T_i > t\}$ is equivalent to $U_i$ taking values in a certain sub-interval of $(0, 1)$. One can define two smooth functions to describe this fact:

\begin{definition}[Expression for  the event $\{T_i > t\}$] \label{def:event_T}
For each $i$, $\omega_{iL}: (0, a_i(\underline{u}_i)] \rightarrow (0, \underline{u}_i]$ and $\omega_{iR}: (0, a_i(\underline{u}_i)] \rightarrow [ \underline{u}_i, 1)$ are respectively two smooth functions such that 
for any $t \in (0, a_i(\underline{u}_i)] $, 
\[
\{T_i > t\} = \{\omega_{iL} (t)  < U_i <  \omega_{iR} (t)  \} = \{S_i > c_i(t)\}  ,
\]
with $\omega_{iL}(\cdot)$ and  $\omega_{iR}(\cdot)$ being strictly increasing and strictly decreasing, respectively. 
Since  $S_i$ and $U_i$ are uniformly distributed when $H_i = 0$, \[\omega_{iR} (t) - \omega_{iL} (t)  = 1 - c_i(t).\] 
Moreover, $b_i( \omega_{iL} (t)) =  b_i( \omega_{iR} (t)) = 1/t$. 
\end{definition}

Of course the variable $S_i = c_i(T_i)$ has the range $(0, 1]$, and we can define functions to describe events of the form $\{S_i > s\}$ similar to \defref{event_T}:
\begin{definition} [Expression for the event $\{S_i > s\}$]\label{def:event_S}
For each $i$, $\psi_{iL}: (0, 1] \rightarrow (0, \underline{u}_i]$ and $\psi_{iR}: (0, 1] \rightarrow [ \underline{u}_i, 1)$ are respectively two smooth functions such that 
for any $s \in (0, 1] $, 
\[
 \{S_i > s\}= \{  \psi_{iL} (s)< U_i < \psi_{iR} (s) \} =  \{T_i >  c_i^{-1}(s)\} , 
\]
with $\psi_{iL}(\cdot)$ and  $\psi_{iR}(\cdot)$ being strictly increasing and strictly decreasing, respectively. Since  $S_i$ and $U_i$ are uniformly distributed when $H_i = 0$, 
\begin{equation} \label{lengthOfSetS}
\psi_{iR} (s) - \psi_{iL} (s) = 1 - s.
\end{equation}
Moreover,  $b_i( \psi_{iL} (s)) =  b_i( \psi_{iR} (s)) = 1/c_i^{-1}(s)$. 
\end{definition}

%

\section{Proof for the asymptotic method} \label{app:zap_asymp_pf}

Before proving \thmref{ZAPasympControl}, we remark that the theorem is established by assuming that the mirror statistic $\hat{T}^\mathfrak{m}_i$ is the \emph{exact} reflection of $\hat{T}_i$ under the null distribution $\hat{c}_i$. In practice, $\hat{T}^\mathfrak{m}_i$ can be determined up to arbitrary precision in  \algref{zapAsymp} as long as the number of uniform realizations $N$ is set to be  very large, as  recommended in \secref{beta_expr}. 

We will make heavy use of the notation and results  in \appref{workModProp}.  We will also use $a^*_i(\cdot)$, $b^*_i(\cdot)$, $c^*_i(\cdot)$, $T_i^*$ and $S_i^*$ to denote the respective functions and statistics when $\{\theta, \beta\}$ is taken to be the pair $\{\theta^*, \beta^*\}$ to construct $a_i(\cdot)$, $b_i(\cdot)$, $c_i(\cdot)$, $T_i$ and $S_i$. Similarly, the  quantities and functions appearing in  \lemref{recipProp} and \defsref{event_T} and \defssref{event_S} all have their ``star" versions: $\underline{u}^*_i$, $\omega^*_{iL}(\cdot)$, $\omega^*_{iR}(\cdot)$, $\psi^*_{iL}(\cdot)$ and $\psi^*_{iR}(\cdot)$. 
Generally speaking,  $C, c >0$  will denote unspecified universal constants required for the asymptotic arguments in this section.

\subsection{Additional assumptions for \thmref{ZAPasympControl}} \label{app:add_assume}

\begin{assumption}[Regularity conditions]  \label{assump:regularity}
\ \

\begin{enumerate}
\item $\max_i \|X_i\|_\infty \leq C$ almost surely for some universal constant $C>0$, where $\|\cdot\|_\infty$ indicates the sup norm.

\item 
Let  $[\delta_1, \delta_2]$ be any fixed compact interval in  $(0, 1)$. 
For each $i$, let $V_{1i}$ and $V_{2i}$ be two measurable subsets in $(\delta_1, \delta_2)$. Then for large enough $m$,
\begin{equation*}
\left| \frac{1}{m} \sum_{i = 1}^m P(U_i \in V_{1i}|X_i) - P(U_i \in V_{2i} | X_i) \right| 
\leq C(\delta_1, \delta_2)\max_{1 \leq i \leq m} \lambda\left( V_{i1} \Delta V_{i2}\right) ,
\end{equation*}
where $\lambda(\cdot)$ is the Lebesgue measure, $C(\delta_1, \delta_2)$ is a constant that may depends on $\delta_1$ and $\delta_2$,  and $ \mathcal{S}_1\Delta \mathcal{S}_2$ is their symmetric difference for any two sets $\mathcal{S}_1, \mathcal{S}_2 \subset \mathbb{R}$.

 \item $\mathbb{E}[\log (U_i)]$ and $\mathbb{E}[\log(1 - U_i)]$ are finite. 

\end{enumerate}
\end{assumption}

\begin{assumption} [Strong laws of large numbers] \label{assump:convgAssump}
Let $\pi_0 \equiv P(H_i = 0) > 0$.  For any $t \in (0,  1]$, it holds that 
\begin{multline}
m^{-1} \sum_{i = 1}^m I(T_i^* \leq t) \longrightarrow G(t), \\
m^{-1} \sum_{i = 1}^m (1 - H_i) I( T_i^*  \leq t)  \longrightarrow \pi_0 G_0(t)  \text{ and } \\
m^{-1}\sum_{i = 1}^m I(S_i^*  \geq  1 - c^*_i( t )) \longrightarrow  \bar{G}_0(t)
\end{multline}
almost surely, where $G(t)$, $G_0(t)$ and $\bar{G}_0(t)$ are positive continuous functions in $t$. Moreover, 
$\bar{G}_0(t_0)/G(t_0)< \alpha$ for some $t_0 > 0$, and the  limiting threshold 
\[
t_\alpha^\infty \equiv \sup \left\{t \in (0, 1]: \frac{\bar{G}_0(t)}{G(t)} \leq \alpha \right\}
\]
is such that $\max_{i \in \mathbb{N}}c_i^*(t_\alpha^\infty)< \bar{s}$ for $\bar{s} < 1$, or equivalently, $t_\alpha^\infty <  \min_i {c_i^*}^{-1}( \bar{s})$. Note that the strong laws above, as well as the marginal probability $\pi_0$, are with respect to the joint law of $\{H_i, Z_i, X_i\}$.

\end{assumption}

\assumpref{regularity} regulates the tail behaviors of the random variables $X_i$ and $U_i$; in particular, $(ii)$ implies that conditional on $X_i$, the density of $U_i$ can be unbounded at the two tails, which is natural for multiple testing as it provides room for non-null tail behaviors.  \assumpref{convgAssump} states properties of the strong law limits involved. For technical reasons,  that $\sup_{i \in \mathbb{N}}c_i^*(t_\alpha^\infty)$ is bounded away from 1  ensures that our result won't rely on the strong law limits for very large values of $t$,  which is hardly restrictive in practice: any sensible multiple testing procedure should only consider rejecting $S_i^*$, which is uniformly distributed under the null,  if  it is much less than the typically small target FDR level $\alpha$.  Similar assumptions have also appeared in the works of \citet{zhang2020covariate, storey2004strong}.

We remark that our current assumptions for \thmref{ZAPasympControl} are no stronger than those in \citet{zhang2020covariate} in any essential way, and can conceivably be further relaxed; for example, if \assumpref{regularity}$(i)$ is phrased as a probabilistic bound, one can still likely establish a version of \thmref{ZAPasympControl} which says that the FDR is less than $\alpha$ with probability approaching 1. Moreover, we have assume, as stated in \secref{problemState},  that $\{H_i, Z_i, X_i\}$ are independent across $i$, which can be further relaxed to a generic weak dependence condition under which the strong laws in \assumpref{convgAssump} hold. In fact it is possible to prove an FDR bound in terms of the conditional expectation $\mathbb{E}[\cdot|H_1, \dots, H_m]$, treating the hypotheses as fixed. These embellishments have not been pursued here for a more streamlined presentation.

\subsection{Technical lemmas} \label{sec:technical}
Under \assumpref{cpct} and \assumpref{regularity}$(i)$, all $k_{ri} = k_r(X_i)$, $k_{li} = k_l(X_i)$, $\pi_{ri} =  \pi_r(X_i)$, $\pi_{li} = \pi_l(X_i)$ are bounded away from one and zero, i.e. 
$k_{ri}, k_{li} \in [ \underline{k} , \bar{k}]$, $\pi_{ri}, \pi_{li} \in [ \underline{\pi} , \bar{\pi}]$ 
for some compact intervals $[ \underline{k} , \bar{k}], [\underline{\pi} , \bar{\pi}] \subset (0, 1)$, by the compactness of $\Theta \times {\bf B}$.  
Moreover, by continuity of the beta functions we can define 
\begin{equation*}
B_{\max} :=\max_{k \in  [ \underline{k} , \bar{k}]} (B(k , \gamma_l )\vee (B(\gamma_r, k ) )\text{ and }
B_{\min} :=\max_{k \in  [ \underline{k} , \bar{k}]} (B(k , \gamma_l )\vee (B(\gamma_r, k ) ),
\end{equation*}
which are both positive numbers.  

The following ``uniformity" properties  will be heavily relied on later: 

\begin{lemma}[Uniformity properties] \label{lem:uniformity}
Under  \assumpref{cpct} and \assumpref{regularity}$(i)$,  the  following   are true  for any $(\theta, \beta) \in \Theta \times \boldB$: 
\begin{enumerate}

\item  For any $t_0 >  0$, there exists a $u_0 = u_0(t_0) > 0$ not depending on $(\theta, \beta)$ such that for all $t \geq t_0$, $(\omega_{iL} (t), \omega_{iR} (t)) \subset [u_0 , 1 - u_0]$ for all $i$. 

\item For any $s_0 > 0$, there exists a $u_0 = u_0(s_0) > 0$ not depending on $(\theta, \beta)$ such that for all $s \geq s_0$, $(\psi_{iL} (s), \psi_{iR} (s)) \subset [u_0 , 1 - u_0]$ for all $i$.

\item There exists a small positive constant $ u_0 \in (0, 0.5)$ not depending on $(\theta, \beta)$ such that, 
\begin{equation*} \label{claimclaim}
\underline{u}_i \in [u_0, 1 - u_0]
\end{equation*}
for all $i$, where  $\underline{u}_i$ is as in  \lemref{recipProp}. 
\item
\[
\lim_{\epsilon \rightarrow 0} \max_{i \in \mathbb{N}, \theta \in \Theta, \beta \in \boldB} \lambda  \Bigl(  {b_i}^{-1} \bigl( [b_i (\underline{u}_i), b_i (\underline{u}_i)+ \epsilon  \bigr) \Bigr) = 0,
\]
where $\lambda$ is the Lebesgue measure, $\underline{u}_i$  is as in \lemref{recipProp} and 
\[
{b_i}^{-1} (\mathcal{T}) \equiv \{u: {b_i}(u) \in \mathcal{T}\}.
\]
for any interval $\mathcal{T}$ in $\mathbb{R}$. 
\end{enumerate}

\end{lemma}

\begin{proof}

$(i)$: Note  that 
\begin{equation} \label{twoBdds}
b_i(u) \geq  
\max\left(
\frac{\underline{\pi} u^{\bar{k} - 1} (1 - u)^{\gamma_l - 1}}{(1 - 2 \underline{\pi})B_{\max} }  ,
  \frac{\underline{\pi} u^{ \gamma_r- 1} (1 - u)^{\bar{k}- 1} }{ (1 - 2 \underline{\pi}) B_{\max}}  
  \right) \text{ for all } i,
\end{equation}
which implies $\lim_{u \rightarrow 0} \min_i  b_i(u)= \lim_{u \rightarrow 1}  \min_ib_i(u) = \infty$, since the right hand side of \eqref{twoBdds} tends to $\infty$ as $u$ tends to 0 or 1. Hence one must be able to find a small enough $u_0 > 0$ such that $\min(b_i(u_0), b_i(1 - u_0) )> 1/t_0$ for all $i$, which implies that $(\omega_{iL}(t) ,\omega_{iR}(t)) \subset [u_0, 1 - u_0] $ for all $i$ by \defref{event_T}. This proves $(i)$. 

$(ii)$: Suppose towards a contradiction, such a $u_0$ doesn't exist. Without loss of generality, we assume there is a subsequence $\{i_1, i_2, \dots\}$ such that $\lim_{j \rightarrow \infty}  \psi_{i_j L} (s_0) = 0$.  As such,  $\lim_{j \rightarrow \infty}  \psi_{i_j R} (s_0) = 1 - s_0$ by the property stated in \defref{event_S}, which implies 
\begin{equation} \label{bdbdbd1}
  \limsup_jb_{i_j}(  \psi_{i_j R} (s_0)  )  
  \leq 1  + \frac{\bar{\pi} (1 - s_0)^{\underline{k} - 1} ( s_0)^{\gamma_l - 1} }{ (1 - 2 \bar{\pi})B_{\min}}
+ \frac{\bar{\pi} ( 1 - s_0)^{\gamma_r- 1} (s_0)^{\underline{k} - 1}  }{ (1 - 2 \bar{\pi})B_{\min}}. 
\end{equation}
On the other hand, 
\begin{equation}\label{bdbdbd2}
\lim_jb_{i_j}^* (\psi^*_{i_j L} (s_0))   = \infty
\end{equation}
in consideration of  \eqref{twoBdds} and $\lim_{j \rightarrow \infty}  \psi^*_{i_j L} (s_0) = 0$.  \eqref{bdbdbd1} and \eqref{bdbdbd2} together reach  a contradiction since it must be that $\lim_jb_{i_j}^* (\psi^*_{i_j L} (s_0))   =   \lim_jb_{i_j}^*(  \psi^*_{i_j R} (s_0)  ) $ as $b_{i_j}^* (\psi^*_{i_j L} (s_0)) = b_{i_j}^*(  \psi^*_{i_j R} (s_0)  )$by \defref{event_S}. 

$(iii)$: By the fact that that $b_i' (b_i(\underline{u}_i)) = 0$ for all $i$, it suffices to show that 
 \begin{equation} \label{uniformInftyLimit}
 \lim_{u \rightarrow 0}\max_{i, \theta, \beta} b_i'(u) =  - \infty \text{ and }  \lim_{u \rightarrow 1}\min_{i, \theta, \beta}  b'_i(u) = \infty.
 \end{equation}
 
Note that  
\[
b_i'( u)  =  \frac{\pi_{li}}{1 - \pi_{ri} - \pi_{li}} h'_{li} (u) +  \frac{\pi_{ri}}{1 - \pi_{ri} - \pi_{li}} h'_{ri}(u),
\]
where $h'_{li}$ has the form
\[
h_{li}'(u) = B(k_{li}, \gamma_l)^{-1} [(k_{li} - 1) u^{k_{li} - 2} (1 - u)^{\gamma_l - 1} - ({\gamma_l - 1}) (1 - u)^{\gamma_l - 2} u^{k_{li} - 1}]
\]
as shown in the proof of \lemref{convex}.  
Define, for $u  \in (0, 1)$, the functions
\[
\bar{h}'_l (u)= B_{\max}^{-1} [ (\bar{k} - 1) u^{\bar{k}  - 2} (1 - u)^ {\gamma_l - 1}-  (\gamma_l - 1) (1 - u)^{\gamma_l - 2} u^{\bar{k} - 1} ],
\]
\[
\underline{h}'_l (u)=  B_{\min}^{-1} [(\underline{k} - 1) u^{\underline{k}  - 2} (1 - u)^ {\gamma_l - 1}-  (\gamma_l - 1) (1 - u)^{\gamma_l - 2} u^{\underline{k} - 1} ]
\]
so that $ \underline{h}'_l (u) \leq  h'_{li }(u)  \leq \bar{h}'_l (u) < 0 $ for all 
$i$. Note that 
\[
 \lim_{u \rightarrow 0} \bar{h}'_l (u)=  - \infty \text{ and }   \lim_{u \rightarrow 1 } \underline{h}'_l (u)=  0.
 \]
  One can similarly define functions $\bar{h}_r'$ and $\underline{h}_r'$ on $(0, 1)$ such that $ 0 < \underline{h}'_r (u) \leq  h'_{ri }(u)  \leq \bar{h}'_r (u)$ for all $i$ and 
\[
 \lim_{u \rightarrow 0} \bar{h}'_r (u)=   0 \text{ and }   \lim_{u \rightarrow 1 } \underline{h}'_r (u)=  \infty.
 \]
The fact that $b_i'(u) \leq \frac{\bar{\pi}}{ 1 -  2 \bar{\pi}} \bar{h}_l'(u)  +  \frac{\bar{\pi}}{ 1 -  2 \bar{\pi}} \bar{h}_r'(u)$, together with two of the limit results above, has shown the first limit in \eqref{uniformInftyLimit}. Similarly, that $b_i'(u)  \geq  \frac{\underline{\pi}}{ 1 -  2 \underline{\pi}} \underline{h}_l'(u)  +  \frac{\underline{\pi}}{ 1 -  2 \underline{\pi}} \underline{h}_r'(u)$, together with the other two limit results above, has shown the second limit in \eqref{uniformInftyLimit}.

$(iv)$: By \lemref{uniformity}$(iii)$, pick $ u_0 \in (0, 0.5)$ such that 
$
\underline{u}_i \in [u_0, 1 - u_0]
$
for all $i$. Note that
\[
b_i'' (u)= \frac{\pi_{li}}{1 - \pi_{li} - \pi_{ri} }   h_{li}''(u)+ \frac{\pi_{ri}}{1 - \pi_{li} - \pi_{ri} }   h_{ri}''(u),
\]
where the dependence on $X_i$ has been suppressed in notations for brevity. By \lemref{convex} ,  $b_i''$ is always positive, hence there exists a universal positive number $c > 0$ such that 
\[
b_i''(u) > c
\]
for all $i$ and all $u \in [u_0/2, 1 - u_0/2 ]$, considering that $\Theta$, $B$ and $[u_0/2, 1 - u_0/2]$ are all compact. Now  for each $i$ consider the quadratic function 
\[
f_i(u) = \frac{c(u - \underline{u}_i )^2}{2} + b_i(\underline{u}_i )
\]
defined on $[u_0/2, 1  - u_0/2]$. Then on the interval  $[u_0/2, 1 - u_0/2]$ , $b_i \geq f_i$  since $g_i = b_i - f_i$ is strictly convex  with $  g_i'  (\underline{u}_i )=   g_i  (\underline{u}_i ) = 0 $. Then 
\[
\lambda \left( b_i^{-1} ([ m_i,m_i + \epsilon  )) \right) \leq \lambda \left( f_i^{-1} ([ m_i,m_i + \epsilon  )) \right)  = \sqrt{\frac{8 \epsilon}{c}}
\]
where the right hand side obviously converges to zero as $\epsilon \rightarrow 0$.

\end{proof}


%

We will now state  two crucial ``event inclusion" lemmas that involve the most delicate proofs in this paper, and may be skipped at first reading. To state them, we conveniently define the long vectors $\Pi = \Pi(\theta) = (\pi_{li}, \pi_{ri})_{i =1}^m$ and 
$K = K(\beta) = (k_{li}, k_{ri})_{i =1}^m$ with $2m$ components. Note that they implicitly depend on the unspecified parameters $\{\theta, \beta\}$. As such, we can also define $\Pi^* = \Pi(\theta^*)$ and $K^* = K(\beta^*)$ to be the versions evaluated at $\theta^*$ and  $\beta^*$. 
\begin{lemma} [First event inclusion lemma] \label{lem:critLem1}
 Suppose \assumpsref{cpct} and  \assumpssref{regularity}$(i)$ are true and let $(\theta, \beta) \in \Theta \times B$.  
For given $\underline{t} > 0$ and $\epsilon > 0$, there exists a $\delta >0$ such that whenever  $\| \Pi - \Pi^*\|_\infty \vee \|K- K^*\|_\infty < \delta$,

\[
\left\{ S_i^* \leq c_i^*(t) -  \epsilon \right\} \subset \left\{  S_i \leq c_i (t)\right\}   \subset \left\{ S_i^* \leq c_i^*(t) + \epsilon\right\} \text{ for all } i \text{ and all } t > \underline{t}. 
\]

\end{lemma}

\begin{proof} [Proof of \lemref{critLem1}]
By \defref{event_T}, we will show, equivalently, that there exists $\delta > 0$ such that whenever $\| \Pi - \Pi^*\|_\infty \vee \|K - K^*\|_\infty < \delta$, 
\begin{equation} \label{inclusion1}
\{U_i \in   \underbrace{(\omega^*_{iL} (t_{i, +\epsilon}) , \omega^*_{iR} (t_{i, +\epsilon}) )}_{\text{length}= 1  - c_i^*(t)  - \epsilon} 
\}\subset \{U_i \in  \underbrace{ (\omega_{iL} (t) , \omega_{iR} (t) )}_{\text{length}  = 1 - c_i(t)}\} 
\subset \{U_i  \in \underbrace{(\omega^*_{iL} (t_{i, - \epsilon}) , \omega^*_{iR} (t_{i, - \epsilon}) )}_{ \text{length} = 1 - c_i^*(t) + \epsilon}\}, 
\end{equation}
where we define $ t_{i, -\epsilon}  :=    {c_i^*}^{-1} (c_i^*(t) - \epsilon)$ and $t_{i, + \epsilon} :=   {c_i^*}^{-1} (c_i^*(t) + \epsilon)$ that are respectively less and greater than $t$. 
In particular, we will first focus on showing the second inclusion in in \eqref{inclusion1} ,  which amounts to showing
\begin{equation} \label{newGoal}
b_i(\omega_{iL}^{ *}(t_{i, -\epsilon})) \wedge  b_i(\omega_{iR}^{ *}(t_{i, -\epsilon})) \geq 1/t
\end{equation}
whenever  $\|\Pi - \Pi^*\|_\infty \vee \|K - K^*\|_\infty < \delta$, in light of the fact that 
\begin{equation} \label{horizontalLevel}
b_i^*(\omega_{iL}^*( t_{i, -\epsilon}) ) =  b_i^*(\omega_{iR}^*(t_{i, -\epsilon}) )=1/ t_{i, -\epsilon}, 
\end{equation}
by the definition of $\omega^*_{iL} (\cdot)$ and  $\omega^*_{iR} (\cdot)$ in \defref{event_T} and properties of $b_i$ from \lemref{recipProp}.  

Let $u_0 > 0$ be a small positive number such that $(\omega^*_{iL} (\underline{t}) , \omega^*_{iR} (\underline{t})) \subset [u_0, 1 - u_0]$ by \lemref{uniformity}(i),  and consider the even larger compact interval $[u_0/2, 1 - u_0/2]$. Consider each $b_i (u) = b (u; \pi_i, k_i )$ as  a function in $(u, \pi_i, k_i)$, and let 
\[
\nabla_{\pi, k} b_i(u)\equiv 
\left( {\frac{\partial}{ \partial \pi_i}b(u; \pi_i, k_i) }^T,  {\frac{\partial}{ \partial k_i}b(u; \pi_i, k_i) }^T \right)^T
\]
 be the gradient of $b_i$ with respect to  $(\pi_i, k_i)$ evaluated at $u$. Using the compactness of  $[u_0/2, 1 - u_0/2]\times \Theta \times \mathcal{B}$ and \assumpref{regularity}$(i)$ again, one can find a universal constant
 $C(u_0 )> 0$ such that the gradient bounds
\begin{equation} \label{gradBound_inclusion1}
\|\nabla_{\pi, k} b_i(u)\|_1  < C(u_0) \text{ for all } i, \text{ all } u \in \left[ \frac{u_0}{2}, 1 - \frac{u_0}{2}\right], \text{ all } (\theta, \beta) \in \Theta \times \boldB.
\end{equation}
On the other hand, without loss of generality, we will let 
\begin{equation} \label{smallerEp}
\epsilon <  u_0/2
\end{equation}
and, with \lemref{uniformity}$(iv)$, 
take $\tilde{\epsilon} > 0 $ be a small enough constant such that 
\begin{equation} \label{tildeEpCond}
  \lambda(  {b^*_i}^{-1} ( [b^*_i(\underline{u}_i^*), b^*_i(\underline{u}_i^*)+ \tilde{\epsilon} )  ) < \epsilon \text{ for all } i . 
\end{equation}
By the  mean-value theorem and the gradient bound \eqref{gradBound_inclusion1}, one can then find $\delta > 0$ such that when $\|\Pi - \Pi^*\|_\infty \vee \|K - K^*\|_\infty < \delta$
\begin{equation} \label{tilde_ep_bound_inclusion1}
|b_i (u) - b_i^* (u )| < \tilde{\epsilon} \text{ for all } u \in \left[\frac{u_0}{2}, 1 - \frac{u_0}{2} \right]. 
\end{equation}

By the construction of $\tilde{\epsilon}$ in \eqref{tildeEpCond} and convexity properties from  
\lemref{recipProp}, one must have for all $i$
\begin{equation*}
 \label{mustbeLowerBdd}
b_i^*(\omega_{iL}^{ *}(t_{i, -\epsilon}) )- b_i^*(\omega_{iL}^{ *}(t)) = b_i^*(\omega_{i
R}^{ *}(t_{i, -\epsilon}) )- b_i^*(\omega_{iR}^{ *}(t))\geq \tilde{\epsilon},
\end{equation*}
which implies 
\begin{equation} \label{anotherLowerBdd}
1/t_{i, -\epsilon} \geq \tilde{\epsilon} + 1/t. 
\end{equation}
by the last property in  \defref{event_T}. 
Since $\{S_i^* >  c^*_i(t) -  \epsilon\}  \supset \{S_i^* >  c^*_i(t) \} $, from the property \eqref{lengthOfSetS} in \defref{event_S} both 
\begin{equation*} \label{both}
\omega_{iL}^{ *}(t_{i, -\epsilon}) \in (\omega_{iL}^{ *}(t)  - \epsilon, \omega_{iL}^{ *}(t)) \text{ and } \omega_{iR}^{ *}(t_{i, -\epsilon}) \in  (\omega_{iR}^{ *}(t),  \omega_{iR}^{ *}(t)+ \epsilon)
\end{equation*}
 are true, 
which implies
\[
[\omega_{iL}^{ *}(t_{i, -\epsilon}), \omega_{iR}^{ *}(t_{i, -\epsilon})]\subset [u_0/2, 1 - u_0/2],
\]
 considering \eqref{smallerEp} and $ \omega_{iL}^{ *}( t ), \omega_{iR}^{ *}( t )\in [u_0, 1- u_0]$ (as $t  \geq \underline{t}$). 
 Therefore by \eqref{tilde_ep_bound_inclusion1}, we must have
\begin{equation} \label{anotherLowerBdd2}
b_i(\omega_{iL}^{ *}(t_{i, -\epsilon})) \wedge  b_i(\omega_{iR}^{ *}(t_{i, -\epsilon}))
\geq
 1/ t_{i, - \epsilon}- \tilde{\epsilon} 
 \end{equation}
 given \eqref{horizontalLevel}.
Combining \eqref{anotherLowerBdd}  and \eqref{anotherLowerBdd2} gives \eqref{newGoal}. 

The proof for the first inclusion in \eqref{inclusion1} follows an analogous argument but is with less resistance, since $[\omega_{iL}^{ *}(t_{i, +\epsilon}), \omega_{iR}^{ *}(t_{i, +\epsilon})]\subset [u_0, 1 - u_0]$ for all $i$. We leave it to the reader. 

\end{proof}

\begin{lemma} [Second event inclusion lemma] \label{lem:critLem2}  Suppose \assumpsref{cpct} and  \assumpssref{regularity}$(i)$ are true and let $(\theta, \beta) \in \Theta \times B$.  
For any fixed $ \bar{t} < \min_i {c_i^*}^{-1}(\bar{s})$ with $\bar{s} < 1$ and any ${\epsilon} > 0$, there exists a $\delta >0$ such that whenever $\| \Pi- \Pi^*\|_\infty \vee \|K- K^*\|_\infty < \delta$,

\[
\{S_i^* >  1 - c^*_{i} ( t) - \epsilon \} \supset
 \{ S_i > 1 - c_i(t)\}  \supset
\{ S_i^* >  1 - c^*_{i} ( t) + \epsilon\}, 
\]
for all $i$ and all $t \leq \bar{t}$. 
\end{lemma}

\begin{proof} [Proof of \lemref{critLem2}]
%

Note that from \eqref{lengthOfSetS} in \defref{event_S} and \lemref{critLem1} we can conclude there exists a $\delta_1 >0$ such that 
whenever $\| \Pi - \Pi^*\|_\infty \vee\|K - K^*\|_\infty < \delta_1$, 
\begin{equation} \label{firstLengthBdd}
 1 - c^*_{i} ( t) -  \epsilon\leq 1 - c_i(t) \leq  1 - c^*_{i} ( t) + \epsilon. 
\end{equation}
Based on \eqref{firstLengthBdd}, it suffices to show that there exists a $\delta_2 > 0$ such that whenever $\| \Pi - \Pi^*\|_\infty \vee \|K - K^*\|_\infty < \delta_2$, 
\begin{equation} \label{lowerBddprime}
\{S_i > 1 - c^*_{i} ( t) -   \epsilon\} \subset \{S_i^* >  1 - c^*_{i} ( t) -  2\epsilon\}, 
\end{equation}
and
\begin{equation} \label{upperBddprime}
\{S_i > 1 - c^*_{i} ( t) +  \epsilon\} \supset \{S_i^* >  1 - c^*_{i} ( t) +  2\epsilon\}
\end{equation}
which conclude the lemma by taking $\delta = \delta_1\vee \delta_2$ and replacing $\epsilon$ with $\epsilon/2$. In fact, since $t \leq \bar{t}$ and $c_i^*(\bar{t})$ are bounded away from 1, we will show the more general statement: For a given $\underline{s} > 0$, there exists $\delta = \delta(\underline{s}) > 0$ such that whenever $\| \Pi - \Pi^*\|_\infty \vee \|K - K^*\|_\infty < \delta$, 
\begin{equation} \label{lowerBdd}
\{S_i >  s\} \subset \{S_i^* >  s - \epsilon\}
\end{equation}
and
\begin{equation} \label{upperBdd}
\{S_i > s\} \supset \{S_i^* >  s + \epsilon\}
\end{equation}
for all $s \geq \underline{s}$ and all $i$.
This will necessitate \eqref{lowerBddprime} and \eqref{upperBddprime}  for $t \leq \bar{t}$.


We will first show \eqref{lowerBdd} which amounts to
\begin{equation} \label{lowerBdd1prime}
( \psi_{iL}( s)  , \psi_{iR}( s )   )
 \subset 
 ( \psi_{iL}^*( s - \epsilon)  , \psi_{iR}^*( s - \epsilon)   )
\end{equation}
in light of \defref{event_S}. In particular, it suffices to only consider the case where $\epsilon< \underline{s}$, since if $s - \epsilon \leq  0$,  $\{S_i^* > s - \epsilon\} = \{S_i^* \geq 0\}$ becomes the whole underlying probability space which makes \eqref{lowerBdd1prime} trivially true. Now for each $i$,  let 
\[
s(i):= \sup\{s' \in (0, 1):  ( \psi_{iL}( s)  , \psi_{iR}( s )   )
 \subset ( \psi_{iL}^*( s')  , \psi_{iR}^*( s')   )\}.
\]
By \defref{event_S} it must be the case that 
\begin{equation} \label{SiEqualsS}
s(i) = 1 -\psi_{iR}^*( s(i)) + \psi_{iL}^*( s(i))\leq  1 - \psi_{iR}( s ) + \psi_{iL}( s) = s,
\end{equation}
and only one of the following possibilities can be true:
\begin{enumerate}
\item $\psi_{iL}^* (s(i)) = \psi_{iL} (s) \text{ and } \psi_{iR}(s) < \psi_{iR}^* (s(i))$, 
\item $\psi_{iL}^* (s(i)) < \psi_{iL} (s) \text{ and } \psi_{iR}(s) = \psi_{iR}^* (s(i))$,
\item $\psi_{iL}^* (s(i)) = \psi_{iL} (s) \text{ and } \psi_{iR}(s) = \psi_{iR}^* (s(i))$.
\end{enumerate}
In light of the monotone properties in \defref{event_S},  it suffices to show that 
\begin{equation}\label{sufficetoshow}
s - s(i) < \epsilon,
\end{equation}  which will then imply \eqref{lowerBdd1prime}. Obviously, if  $(iii)$ is true then \eqref{sufficetoshow} must be true in light of \eqref{SiEqualsS}. We will focus on showing \eqref{sufficetoshow} in the case of $(i)$ since the proof for the case of $(ii)$ follows a parallel argument.

By \lemref{uniformity}$(ii)$, there exists a  $u_0 = u_0(\underline{s})> 0$ such that
 \begin{equation}
( \psi_{iL} (s) ,  \psi_{iR} (s))    \subset [u_0, 1 - u_0] \text{ for all } i \in \mathbb{N} \text{ and } s \geq \underline{s}. 
 \end{equation}
Consider each $b_i (u) = b (u; \pi_i, k_i )$ as  a function in $(u, \pi_i, k_i)$, and let 
\[
\nabla_{\pi, k} b_i(u): = 
\left( {\frac{\partial}{ \partial \pi_i}b(u; \pi_i, k_i) }^T,  {\frac{\partial}{ \partial k_i}b(u; \pi_i, k_i) }^T \right)^T
\]
 be the gradient of $b_i$ with respect to  $(\pi_i, k_i)$ evaluated at $u$. Using the compactness of  $[u_0, 1 - u_0] \times \Theta \times \mathcal{B}$ and \assumpref{regularity}$(i)$ again, one can find a  constant
 $C(u_0 )> 0$ such that the gradient bounds
\begin{equation} \label{gradBound}
\|\nabla_{\pi, k} b_i(u)\|_1  < C(u_0) \text{ for all } i \text{ and for all } u \in \left[ \frac{u_0}{2}, 1 - \frac{u_0}{2}\right].
\end{equation}
On the other hand, without loss of generality, with \lemref{uniformity}$(iv)$, 
let $\tilde{\epsilon} > 0 $ be a small enough constant such that 
\begin{equation}\label{tildeep2}
  \lambda(  {b^*_i}^{-1} ( [m^*_i, m^*_i+ \tilde{\epsilon} )  ) < \frac{\epsilon}{2} \text{ for all } i . 
\end{equation}
By the  mean-value theorem and the gradient bound \eqref{gradBound}, one can then find $\delta > 0$ such that when $\|\Pi - \Pi^*\|_\infty \vee \|K - K^*\|_\infty < \delta$
\begin{equation} \label{tilde_ep_bound}
|b_i (u) - b_i^* (u )| < \tilde{\epsilon} \text{ for all } u \in \left[u_0, 1 - u_0\right]. 
\end{equation}
Since  $\psi_{iL}^*(s(i)) = \psi_{iL}(s) \in [u_0, 1 - u_0]$,  \eqref{tilde_ep_bound} and the last property in \defref{event_S} suggest that
\begin{equation} \label{bistar1}
b_i^*\Bigl( \psi_{iR}^*(s(i))\Bigr) = b_i^*\Bigl( \psi_{iL}^*(s(i)) \Bigr)<  b_i\Bigl( \psi_{iL}(s)\Bigr) + \tilde{\epsilon}    = b_i \Bigl(\psi_{iR}(s) \Bigr)  + \tilde{\epsilon}.
\end{equation}
But since $\psi_{iR}(s)$ is also in the interval $[u_0, 1 - u_0]$, we must have
\begin{equation} \label{bistar2}
b_i^*\Bigl( \psi_{iR}(s) \Bigr) > b_i \Bigl(\psi_{iR}(s) \Bigr)  - \tilde{\epsilon}.
\end{equation}
Combining \eqref{bistar1} and \eqref{bistar2}, we get that $2 \tilde{\epsilon} >  b_i^*\Bigl( \psi_{iR}^*(s(i))\Bigr) - b_i^*\Bigl( \psi_{iR}(s) \Bigr)$ which in light of the construction of $\tilde{\epsilon}$  in \eqref{tildeep2} and convexity properties from  
\lemref{recipProp} gives that
\[
\psi_{iR}^*(s(i)) - \psi_{iR}(s) < \epsilon,
\]
which in turn implies \eqref{sufficetoshow} by the property \eqref{lengthOfSetS} in \defref{event_S}.
The  proof of \eqref{upperBdd} is similar. It amounts to showing
\[
 ( \psi_{iL}^*( s + \epsilon)  , \psi_{iR}^*( s + \epsilon)   )
 \subset
( \psi_{iL}( s)  , \psi_{iR}( s )   ). 
\]
We will alternatively define 
\[
\tilde{s}(i):= \inf\{s' \in (0, 1):   ( \psi_{iL}^*( s')  , \psi_{iR}^*( s')   ) 
\subset( \psi_{iL}( s)  , \psi_{iR}( s )   )
\}.
\]
then show $\tilde{s}(i) - s < \epsilon$. We leave the details to the reader. 
\end{proof}

\subsection{A Glivenko-Cantelli theorem}


\begin{lemma} [Pre-Glivenko-Cantelli theorem] \label{lem:pre-GCthm}
Under \assumpsref{cpct}-\assumpssref{convgAssump} , for any $\epsilon >0$ and positive numbers $0 < \underline{t} < \bar{t} < \min_i {c_i^*}^{-1} (\bar{s}) $, there exists $\delta = \delta(\epsilon) >  0$  such that, for sufficiently large $m$, 
\begin{align}
\sup_{ \substack{\max(\|K - K^*\| ,\|\Pi - \Pi^*\|  )< \delta\\   \underline{t} \leq t \leq \bar{t}}}\left|   \frac{1}{m} \sum_{i = 1}^m I(T_i\leq  t)  - G (t)\right| 
\leq  \epsilon  \label{denominatorConvg}\\
\sup_{ \substack{\max(\|K - K^*\| ,\|\Pi - \Pi^*\|  )< \delta\\   \underline{t} \leq t \leq \bar{t}}} \left|\frac{1}{m} \sum_{i = 1}^m (1- H_i) I( T_i  \leq  t )  -  \pi_0{G}_0(t)\right| \leq \epsilon ,\label{PropernumeratorConvg}\\
\sup_{ \substack{\max(\|K - K^*\| ,\|\Pi - \Pi^*\|  )< \delta\\   \underline{t} \leq t \leq \bar{t}}} \left|\frac{1}{m}\sum_{i = 1}^m  I( S_i  \geq  1 - c_i( t ))  -  \bar{G}_0(t)\right| \leq  \epsilon
 \label{numeratorConvg} 
\end{align}
with probability 1. 
\end{lemma}

\begin{proof}[Proof of \lemref{pre-GCthm}]

In this proof, for any function $F(\cdot)$,  $F(t-)$ denotes the left limit at the point $t$.

Proof of \eqref{denominatorConvg}: Let 
\[
G_{\theta, \beta} (t)= \frac{1}{m} \sum_{i = 1 }^m I( T_i  \leq t ), 
\]
where the subscript emphasizes that the $T_i$'s are defined with an unspecified $(\theta, \beta)$, to distinguish from $G$ in \assumpref{convgAssump}. 
Let $n$ be large enough such that $1/n<  \epsilon/2$ and consider $G^{\leftarrow}  (1/n) \leq \dots \leq G^{\leftarrow}  (n/n)$. If we define
\[
n' : = \min\left\{ i:  \underline{t} <   G^{\leftarrow} (i/n) < \bar{t}, i = 1, \dots, n \right\}, 
\]
\[
d' : = |\left\{ i:  \underline{t} <   G^{\leftarrow} (i/n) <  \bar{t}, i = 1, \dots, n \right\}|. 
\]
Define $t_1:=  G^{\leftarrow} (n'/n), t_2 := G^{\leftarrow} (  (n' +1)/n), \dots, t_{d'} := G^{\leftarrow} (  (n' + d'-1)/n)$, as well as $t_0 = \underline{t}$ and $t_d = \bar{t}$ with $d = d' +1$. Following the proof of the Glivenko-Cantelli theorem in \citet[p.224]{resnick2019probability},  we have
\begin{multline} \label{resnick}
\sup_{\underline{t} \leq t \leq \bar{t}} \left| G_{\theta, \beta}(t)  - G(t) \right|  
\leq  \left(\bigvee_{v = 0}^d  |G_{\theta, \beta} (t_{v}) - G(t_v) |\vee| G_{\theta, \beta} (t_v-) - G (t_v - )|\right)  +  1/n
\end{multline}

We will first bound the terms of the form $|G_{\theta, \beta}(t_{v}) - G(t_v) |$ in \eqref{resnick}. The  strong law of large numbers for $G(\cdot)$ in \assumpref{convgAssump} suggests that 
\begin{equation} \label{ABcanbdd}
|G_{\theta, \beta} (t_{v}) - G(t_v) | \leq  
  \left|m^{-1}\sum_{i = 1 }^m (I( T_i  \leq t_v ) - I(T_i^*  \leq t_v))\right| + R_v, 
\end{equation}
where the remainder term $R_v \longrightarrow 0$ almost surely.  
 Now, realizing $\{T_i^*  \leq t_v\} = \{S_i^*  \leq c_i^*(t_v)\}$, by \lemref{critLem1} and $\underline{t} >0$, pick $\delta >0$ such that 
 \begin{multline} \label{useInclusionLemma}
 \left|\frac{1}{m}\sum_{i = 1 }^m (I( T_i  \leq t_v ) - I(T_i^*  \leq t_v))\right| \leq  \\
 \underbrace{\left|  m^{-1}\sum_{i = 1 }^m \left[ I(S^*_i  \leq c_i^*( t_v) + \epsilon) - I(S_i^*  \leq c_i^*(t_v))\right]\right|}_{(A)} \vee  \underbrace{  \left|  m^{-1}\sum_{i = 1 }^m \left[I(S^*_i  \leq c_i^*( t_v)  - \epsilon) -I(S_i^*  \leq c_i^*(t_v))\right]\right|}_{(B)}.
 \end{multline}
for $\max(\|\Pi - \Pi^*\|, \|K - K^*\|) < \delta$. 
 This is because  $\sum_{i = 1 }^m \left[ I(S^*_i  \leq c_i^*( t_v) + \epsilon) - I(S_i^*  \leq c_i^*(t_v))\right] \geq \sum_{i = 1 }^m (I( T_i  \leq t_v ) - I(T_i^*  \leq t_v)) $ if the latter term is greater than $0$; likewise,
 $\sum_{i = 1 }^m \left[ I(S^*_i  \leq c_i^*( t_v) - \epsilon) - I(S_i^*  \leq c_i^*(t_v))\right] \leq \sum_{i = 1 }^m (I( T_i  \leq t_v ) - I(T_i^*  \leq t_v)) $ if the latter is less than $0$.

 We will first develop a bound for term $(A)$. One have
  \begin{align}
 (A) &
  \leq \mathbb{E} \Bigl|        \frac{1}{m}\sum_{i = 1}^m  [ P(S^*_i  >  c_i^*( t_v) |X_i ) - P(S^*_i  > c_i^*( t_v) + \epsilon | X_i)] \Bigr| + Q_v^+\notag\\
 &=  \mathbb{E} \Bigl|\frac{1}{m}\sum_{i = 1}^m \bigl[ P\bigl( U_i \in (\psi^*_{iL} (c_i^*( t_v)) ,  \psi^*_{iR} (c_i^*( t_v) )|X_i \bigr)  - \notag\\
 & \hspace{3cm} P(  U_i \in (\psi^*_{iL} (c_i^*( t_v) + \epsilon ) ,  \psi^*_{iR} (c_i^*( t_v) + \epsilon) | X_i)\bigr]\Bigr| + Q_v^+ \notag\\
 &\leq C(\underline{t}) \epsilon + Q_v^+ \label{bddfortermA}
 \end{align}
 where  $Q_v^+ = o_{a.s.}(1)$ is a remainder term coming from the strong law of $G(\cdot)$ in \assumpref{convgAssump}. The second equality comes from \defref{event_S}. Note that the intervals 
 \[
 (\psi^*_{iL} (c_i^*( t_v)) ,  \psi^*_{iR} (c_i^*( t_v) ), i =1, \dots, m
 \]
can be equivalently represented as 
\[
 (\omega^*_{iL} (t_v) ,  \omega^*_{iR} (t_v)), i =1, \dots, m
\]
by \defref{event_T}, which all belong to a compact sub-interval in $(0, 1)$ by \lemref{uniformity}$(i)$ and the fact that $\underline{t} > 0$. As such, \assumpref{regularity}$(ii)$ can be applied to give the last inequality \eqref{bddfortermA}. By realizing, from \defsref{event_T} and \defssref{event_S}, that the event
\[\{S^*_i  >  c_i^*( t_v)  - \epsilon \}\] 
is equivalent to $U_i$ belonging to an interval that is $ (\omega^*_{iL} (t_v) ,  \omega^*_{iR} (t_v))$ expanded by a further $\epsilon$ width, it is obvious that one can analogously develop the bound
\begin{equation} \label{bddfortermB}
(B) \leq C(\underline{t}) \epsilon+ Q_v^{-} 
\end{equation}
for a constant $C(\underline{t}) $ and $Q_v^{-} = o_{a.s.} (1)$. Combining \eqref{ABcanbdd},  \eqref{bddfortermA} and \eqref{bddfortermB} gives (by appropriately adjusting $\epsilon$)
\begin{equation} \label{g_theta_beta_first_bound}
|G_{\theta, \beta} (t_{v}) - G(t_v) | \leq  \epsilon/2 \text{ a.s.}, 
\end{equation}
 for sufficiently large $m$. A similar bound
\begin{equation} \label{g_theta_beta_second_bound}
|G_{\theta, \beta} (t_{v}-) - G(t_v-) | \leq  \epsilon/2 \text{ a.s.}, 
\end{equation}
can be derived in much the same way with no difficulty by first writing 
\[G_{\theta, \beta} (t_{v}-) - G(t_v-) =  m^{-1} \sum_{i =1}^mI(T_i < t_v)- G(t_v) 
\] using the continuity of $G(\cdot)$, and the proof is omitted for brevity. Combining \eqref{resnick}, \eqref{g_theta_beta_first_bound} and \eqref{g_theta_beta_second_bound} give \eqref{denominatorConvg}. \eqref{PropernumeratorConvg} can be proved in the same way by first writing
\[
 \sum_{i = 1}^m (1- H_i) I( T_i  \leq  t )   = \sum_{H_i = 0} I( T_i  \leq  t ),
\]
and noting that $U_i$ is uniformly distributed given  $H_i =0$ and $X_i$,  and thus omitted.

%

Proving \eqref{numeratorConvg} is similar. One proceed by developing the bound 
\begin{align}
 & \left|  \frac{1}{m}\sum_i[ I( S_i  >   1 - c_i( t_v )) - P(S_i^* > 1 - c_i^*(t_v))] \right|  \notag \\
 & \leq \left|  \frac{1}{m}\sum_i [ I( S_i^*  >   1 - c_i^*( t_v ) - \epsilon) - P(S_i^* >  1 - c_i^*(t_v))] \right|  \vee  \notag\\
 & \hspace{4cm} \left|  \frac{1}{m}\sum_i[ I( S^*_i  >   1 - c^*_i( t_v ) + \epsilon) - P(S_i^* >  1 - c_i^*(t_v))] \right| \notag 
  \end{align}
with  \lemref{critLem2}
for $t_v$ that are now quantiles of $\bar{G}_0$, for sufficiently small $\delta >0$ and $\max(\|\Pi - \Pi^*\|, \|K - K^*\|) < \delta$.  From \defref{event_S}, the events $\{S_i^*  >   1 - c_i^*( t_v ) \}$ have the form
\[
\{U_i \in (\psi^*_{iL}(1 - c_i^*( t_v ) ) , \psi^*_{iR}(1 - c_i^*( t_v ) ))\},
\]
and to show that the intervals $(\psi^*_{iL}(1 - c_i^*( t_v ) ) , \psi^*_{iR}(1 - c_i^*( t_v ) )$ can be placed in a compact sub-interval of $(0, 1)$ by \lemref{uniformity}$(ii)$ to apply \assumpref{regularity}$(ii)$, one need to show that $c_1^*( t_v ) , \dots, c_m^*( t_v ) $ are bounded away from 1. This is true because  $c_i^*(t_v)\leq c_i^*(\bar{t}) \leq \bar{s} < 1$ for all $i$ by definition. The same proof rundown goes through, again, by realizing that $(\psi^*_{iL}(1 - c_i^*( t_v ) - \epsilon) , \psi^*_{iR}(1 - c_i^*( t_v ) - \epsilon)$ are just $\epsilon$-expansion of $(\psi^*_{iL}(1 - c_i^*( t_v ) ) , \psi^*_{iR}(1 - c_i^*( t_v ) )$ from \defref{event_S}. The rest of the proof goes thru with no resistance.

\ \

\ \

\ \

\end{proof}

\begin{lemma}[Glivenko-Cantelli Theorems] \label{lem:GCthm}
Under \assumpsref{cpct}-\assumpssref{convgAssump}, for any $0 < \underline{t} < \bar{t} \leq \min_i {c_i^*}^{-1} (\bar{s}) $, we have
\begin{align}
&\sup_{ \underline{t}  \leq t \leq \bar{t}} \left|\frac{1}{m} \sum_{i = 1}^mI( \hat{T}_i  \leq t)  - G(t)\right|  \longrightarrow 0, \label{GCdenom}\\
& \sup_{ \underline{t}  \leq t \leq \bar{t}} \left|\frac{1}{m} \sum_{i = 1}^m (1- H_i)I( \hat{T}_i  \leq t)  - \pi_0G_0(t)\right|  \longrightarrow 0, \label{GCnumer}\\
& \sup_{ \underline{t}  \leq t \leq \bar{t}} \left|\frac{1}{m} \sum_{i = 1}^mI( \hat{S}_i  \geq 1 - \hat{c}_i(t))  - \bar{G}_0(t)\right|  \longrightarrow 0,  \label{GCnumerbar}
\end{align}
almost surely.
\end{lemma}
\begin{proof} [Proof of \lemref{GCthm}]
Let $\hat{\theta} = (\hat{\theta}_l^T, \hat{\theta}_r^T)^T, \hat{\beta} = (\hat{\beta}_l^T, \hat{\beta}_r^T)^T$, and define $\hat{\Pi} = \Pi(\hat{\theta})$ and $\hat{K} = K(\hat{\beta})$.
We will first show that 
\begin{equation} \label{PiKconvg}
\|\hat{\Pi} - \Pi^*\|_\infty \vee \|\hat{K} - K^*\|_\infty \longrightarrow 0 \text{ a.s.},
\end{equation}
which is a consequence of 
\begin{equation}\label{thetabetaconvg}
\hat{\theta}  \longrightarrow \theta^* \text{ and } \hat{\beta} \longrightarrow  \beta^* \ \  a.s.
\end{equation}
by the mean value theorem, the compactness of $\Theta \times \boldB$ and \assumpref{regularity}$(i)$. To show \eqref{thetabetaconvg},  it suffices to bound
\begin{equation}\label{absloglike}
\left|\log [ (1 - \pi_{li} - \pi_{ri}) + \pi_{li}h_{li} (U_i) + \pi_{ri} h_{ri} (U_i) ]\right| = |\log h_{X_i} (U_i)|
\end{equation}
by an integrable function in $U_i$ that doesn't depend on $(\theta, \beta)$ \citep[Theorem 2.1]{white1981consequences}. We first let $u_0 >0$ be as in \lemref{uniformity}$(iii)$. By the compactness of $\Theta \times \boldB \times [u_0, 1 - u_0]$ and \assumpref{regularity}$(i)$, there exists  universal constants $C, c >0$ such that $\log(c) < 0$, $\log(C) >0$, and
\[
C \geq h_{X_i} (U_i) \geq c
\]
for all  $(\theta, \beta, X_i)$ whenever $U_i \in [u_0, 1 - u_0]$. Note that \lemref{recipProp} also implies that 
\[
h_{X_i} (U_i) \geq c \text{ for } U_i \in (0, u_0)\cup (1 - u_0, 1).
\]
Moreover, for $U_i \in (0, u_0)$, $\pi_{ri}h_{ri} (U_i) \leq C$ by \lemref{convex}, hence borrowing notations from  \secref{technical}, 
\begin{align*}
\log h_{X_i} (U_i) &\leq \log\left(1 + C +  \bar{\pi}\frac{{(1 - U_i)}^{\gamma_l - 1}U_i^{\underline{k}  - 1}}{B_{\min} }  \right) \\
&\leq \log(2 + C)  + \log  \left( \tilde{C}\frac{ (1- U_i)^{\gamma_l - 1}U_i^{\underline{k}  - 1}}{(1- u_0)^{\gamma_l - 1}u_0^{\underline{k}  - 1}}  \right)  \equiv \underbrace{m_l(U_i)}_{>0},
\end{align*}
for a constant $\tilde{C} > 1 \vee \frac{\bar{\pi}}{B_{\min}}$. 
Similarly, there exists a positive function $m_r (\cdot)$ such that
\[
\log h_{X_i} (U_i) \leq m_r (U_i)
\]
for $U_i \in (1 - u_0, 1)$. Combining these facts we have for all $U_i \in (0, 1)$, 
\[
|\log h_{X_i} (U_i)| \leq  |\log(c)|  \vee \log (C)\vee m_r(U_i) \vee m_l(U_i),
\]
where the right hand side is integrable by \assumpref{regularity} $(iii)$ and \eqref{thetabetaconvg} is proved. 
%
%
%
%

We will only prove \eqref{GCdenom}, and \eqref{GCnumer} and \eqref{GCnumerbar} can be shown the same way. Let
$
D_m = D_m (\omega) = \sup_{ \underline{t}  \leq t \leq \bar{t}} \left|\frac{1}{m} \sum_{i = 1}^mI( \hat{T}_i  \leq t)  - G(t)\right|
$,
where $\omega$  denotes a point in the underlying probability space $\Omega$. It suffices to show that for any $\epsilon > 0$, there exists a subspace $\Lambda(\epsilon) \subset \Omega$ such that $P(\Lambda(\epsilon)) = 1$ and 
$
D_m (\omega) <  \epsilon 
$
for sufficiently large $m$ and every $\omega \in \Lambda(\epsilon)$. By \lemref{pre-GCthm}, there exists $\Lambda_1$ with $P(\Lambda_1) = 1$ such that \eqref{denominatorConvg} holds on $\Lambda_1$ for $\delta(\epsilon) >0$. By \eqref{PiKconvg}, there exists $\Lambda_2$ with $P(\Lambda_2) = 1$ such that 
$
\|\hat{\Pi} - \Pi\|_\infty \vee \|\hat{K} - K\|_\infty  < \delta(\epsilon)
$
on $\Lambda_2$  for sufficiently large $m$. Take $\Lambda(\epsilon) = \Lambda_1\cap \Lambda_2$. 
\end{proof}

\subsection{Proof of \thmref{ZAPasympControl}} \label{sec:final_pf}

The proof is similar to that of \citet[Theorem 4]{storey2004strong} but is a bit more sutble. Recall the ratio in \eqref{FDP_est_zap_asymp}. We shall first show that 
under all the assumptions of \thmref{ZAPasympControl}, for any $ \underline{t} >0$ and any $ \bar{t} \in (t^\infty_{\alpha}, {\min}_i {c_i^*}^{-1}(\bar{s}) ]$, 
\begin{equation} \label{lowerbdd_FDP_ZAP}
 \liminf_{m \rightarrow \infty} \inf_{ \underline{t} \leq t \leq \bar{t} } \left \{ \widehat{\mbox{FDP}}_{asymp} (t) -  
FDP(t)
 \right\}
 \geq 0 \text{ a.s.}
\end{equation}
where $FDP(t) \equiv  \frac{\sum_i (1- H_i) I( \hat{T}_i  \leq   t)}{
\left(  \sum_i  I( \hat{T}_i  \leq   t)\right) \vee 1}$ for any $t > 0$. 
%
From the first two Glivenko-Cantelli statements in  \lemref{GCthm}, we see that
%
%


\begin{multline}
\lim_{m \rightarrow \infty} \sup_{\underline{t} \leq t \leq \bar{t}} \left|FDP(t) - \frac{ \pi_0 m G_0 (t) }{\left(  \sum_{i = 1}^m I( \hat{T}_i  \leq   t)\right) \vee 1}\right|\\
= \lim_{m \rightarrow \infty} \sup_{\underline{t} \leq t \leq \bar{t}} \left|  \frac{\sum_{i} (1- H_i) I( \hat{T}_i  \leq   t)}{
\left(  \sum_{i = 1}^m I( \hat{T}_i  \leq   t)\right) \vee 1}  - \frac{ \pi_0 m G_0 (t) }{\left(  \sum_{i = 1}^m I( \hat{T}_i  \leq   t)\right) \vee 1}\right|\\
\leq \lim_{m \rightarrow \infty}    \left| \frac{m}{ \left(  \sum_{i = 1}^m I( \hat{T}_i  \leq  \underline{t})\right) \vee 1  } \right| \sup_{\underline{t} \leq t \leq \bar{t}}  \left|   \frac{\sum_i (1 - H_i) I( \hat{T}_i  \leq   t)}{m}  - \pi_0 G_0(t) \right| = 0 \text{ a.s.}\label{fdd}
\end{multline}
since $\lim_{m \rightarrow \infty}    \left| \frac{m}{ \left(  \sum_{i = 1}^m I( \hat{T}_i  \leq \underline{t})\right) \vee 1  } \right|  =  1/G(\underline{t}) < \infty$ almost surely, given that $G(\underline{t}) > 0$. On the other hand, it must be that $\bar{G}_0(t) \geq \pi_0 G_0(t)$ considering that $ P( S_i^*  \geq  1 - c^*_i( t )| H_i = 0) = P(T_i^* \leq t | H_i = 0)$, which, together with the last Glivenko-Cantelli statement in  \lemref{GCthm},  
gives
\[
\liminf_m \inf_{\underline{t} \leq t \leq \bar{t}}  \left\{ \frac{ \sum_{i = 1}^m I( \hat{S}_i  \geq  1 - \hat{c}_i( t )) }{m} -  \pi_0 G_0(t) \right\} \geq 0.
\]
The preceding display and  \eqref{fdd} will lead to
\begin{multline*}
\liminf_{m \rightarrow \infty} \inf_{ \underline{t} \leq t \leq \bar{t}} \left \{ \widehat{\mbox{FDP}}_{asymp} (t) -  
FDP(t)
 \right\}\\
  \geq \liminf_{m \rightarrow \infty} \inf_{\underline{t} \leq t \leq\bar{t}}  \left\{
   \frac{ \sum_{i = 1}^m I( \hat{S}_i  \geq  1 - \hat{c}_i( t )) }{ \left(  \sum_{i = 1}^m I( \hat{T}_i  \leq   t)\right) \vee 1}
    - FDP(t) \right\} \geq 0,
\end{multline*}
which is \eqref{lowerbdd_FDP_ZAP}. 
%
 
Towards finishing, we will establish that, almost surely, 
\begin{equation} \label{t_zap_range}
\liminf_m \hat{t}_{asymp} (\alpha) > 0 \text{ and }  \limsup_m \hat{t}_{asymp} (\alpha) \leq\min_i {c_i^*}^{-1} (\bar{s}).
\end{equation} 
Fix $t_1 \in (t^\infty_{\alpha},  \min_i {c_i^*}^{-1} (\bar{s}))$. By the definition of $t^\infty_\alpha$ in \assumpref{convgAssump} it must be the case that
\[
\frac{\bar{G}_0(t_1)}{G(t_1)} > \alpha = \frac{\bar{G}_0(t^\infty_{\alpha})}{G(t^\infty_{\alpha})} 
\]
and we can let $\frac{\bar{G}_0(t_1)}{G(t_1)} - \alpha = \epsilon_1 >0$. For sufficiently large $m$, because of \lemref{GCthm} we can get that 
$
  \left| \frac{ \sum_i I( \hat{S}_i  \geq  1 - \hat{c}_i( t_1)) }{ \left(  \sum_i I( \hat{T}_i  \leq   t_1)\right) \vee 1} - \frac{\bar{G}_0(t_1)}{G(t_1)}  \right| < \epsilon_1/2  \text{ a.s.,}
$
which implies that $\widehat{\mbox{FDP}}_{asymp} (t_1) > \alpha$ almost surely to give  the ``limsup" statement in  \eqref{t_zap_range}. On the other hand, let $\alpha  - \frac{\bar{G}_0(t_0)}{G(t_0)} = \epsilon_0 > 0$ for $t_0$ in \assumpref{convgAssump}. Since $t_0 < t^\infty_{\alpha}$ (by continuity of the ``G" functions), \lemref{GCthm} also suggests that for large enough $m$, $|\frac{\bar{G}_0(t_0)}{G(t_0)}  - \widehat{\mbox{FDP}}_{asymp} (t_0)| < \epsilon_0/2$ almost surely, which implies $ \widehat{\mbox{FDP}}_{asymp} (t_0) < \alpha$ almost surely and hence the ``liminf" statement in \eqref{t_zap_range}.

Now given \eqref{t_zap_range} is true, since $\widehat{\mbox{FDP}}_{asymp} (\hat{t}_{asymp} (\alpha)) \leq \alpha$, by \eqref{lowerbdd_FDP_ZAP} it must be true that
\[
\limsup_m FDP(\hat{t}_{asymp} (\alpha)) \leq \alpha\text{ a.s. }
\]
By the reverse Fatou's lemma, this implies
\[
\limsup_m  \mathbb{E}[FDP(\hat{t}_{ asymp} (\alpha) ) ] \leq \mathbb{E} \left[  \limsup_m  FDP(\hat{t}_{ asymp} (\alpha) )\right] \leq \alpha, 
\]
and  \thmref{ZAPasympControl}  is proved. 

\section{Proof for the finite-sample method} \label{app:finiteZAPpf}

The proof is almost exactly the same as that of \citet[Theorem 1]{lei2016adapt} which relies on the key  lemma in that paper \citep[Lemma 2]{lei2016adapt}, and we will only define the notation required to apply their argument. 
First, for each $t = 0, 1, \dots$, let
$
\mathbb{V}_t = \# \{i: U_i \in \mathcal{R}_t\text{ and } H_i = 0\}
$
and
$
\mathbb{U}_t = \# \{i: U_i \in \mathcal{A}_t\text{ and } H_i = 0\}
$
which are respectively the numbers of true nulls in the rejection set and acceptance set at step $t$. 
Define
\[
m_i =  I(U_i \geq 0.5) ( \widecheck{U}_i \vee U_i) +   I(U_i <  0.5) (\widecheck{U}_i\wedge U_i)
\]
and
\[
b_i  = I(  0.25 \leq  U_i \leq  0.75) 
\] 
so that
\[
U_i = b_i \{ I(m_i \geq 0.5) ( 1.5 - m_i )  + I(m_i < 0.5) (0.5 - m_i)\} + (1 - b_i)m_i. 
\]
Also define
$
\mathcal{C}_t = \{i: i \in \mathcal{A}_t \cup  \mathcal{R}_t \text{ and }H_i = 0\}
$
to give
\[
\mathbb{U}_t = \sum_{i \in \mathcal{C}_t} b_i
\text{ and }
\mathbb{V}_t = \sum_{i \in \mathcal{C}_t} ( 1 - b_i) = |\mathcal{C}_t| - \mathbb{U}_t.
\]
If we set the initial sigma-algebra $\mathcal{G}_{-1} = \sigma\{ (X_i, m_i)_{i\in \{1, \dots, m\}} , (b_i)_{i:H_i \neq 0}\}$, then $P( b_i = 1| \mathcal{G}_{-1}) = 0.5$ almost surely for a null $i$ under the uniform null distribution of $U_i$. With these ingredients, the arguments in the proof of \citet[Theorem 1]{lei2016adapt} will follow line by line, where the $U_i$'s will take the role of the $p$-values in that paper.

\section{Supplementary algorithms} \label{app:EM}

We will inherit the simplified notation in \appref{workModProp}. The \emph{complete data log-likelihood} for Model \eqref{hu-beta}, treated as a function of $\{\theta, \beta\}$,  has the form
\begin{align}
l(\theta, \beta) 
&= \sum_{i = 1}^m  H_{li} \left[ (k_{li}   - 1) \log( U_i )  + (\gamma_l - 1) \log (1 - U_i) - \log B(k_{li}, \gamma_l) \right] + \notag\\ 
&\quad \sum_{i = 1}^mH_{ri} \left[ (k_{ri} - 1) \log( 1 - U_i )  + (\gamma_r - 1) \log (U_i) - \log B( \gamma_r, k_{ri}) \right]  + \notag\\
 &\quad  \left[ \sum_{i = 1}^m (1 - H_{li} - H_{ri}) \log(1 - \pi_{li} - \pi_{ri}) +   H_{li} \log( \pi_{li}) +  H_{ri} \log( \pi_{ri})\right], \label{multinomialReg}
 \end{align}
where  for each $i$, $H_{li}$ and $H_{ri}$ are Bernoulli random variables with respective success probabilities $\pi_{li}$ and $\pi_{ri}$, and $H_{li}$ and $H_{ri}$ cannot be both equal to 1 at the same time. Note that the last line in \eqref{multinomialReg} amounts to a multinomial logistic regression with three classes. 

%
%
%

\subsection{EM algorithm for asymptotic ZAP} \label{app:EM_zap_asymp}

\ \

{\LinesNumberedHidden
\begin{algorithm}[H]
\caption{EM algorithm for asymptotic ZAP}
 \label{alg:EMasymp} 
\KwData{$U_1, \ldots, U_m, X_1, \dots, X_m$}
\KwIn{initial guess $\beta^{(0)}$, $\theta^{(0)}$}

\While{ $(\beta^{(j)}, \theta^{(j)})$ not converged}{
{\bf E step}: Let $\pi_{li}^{(j)}$, $\pi_{ri}^{(j)}$, $h_{li}^{(j)}$, $h_{ri}^{(j)}$ and $h_{X_i}^{(j)}$ be
as defined in \secref{betaMod} evaluated at $(\beta^{(j)}, \theta^{(j)})$. Compute 
\begin{multline*}
Q^{(j)}(\beta, \theta) = \\\sum_{i = 1}^m \left\{w_{li}^{(j)} \log[\pi_{li} h_{li} (U_i) ]+w_{ri}^{(j)} \log[\pi_{ri} h_{ri} (U_i) ]  +
 (1 - w_{li}^{(j)} - w_{ri}^{(j)}) \log (  1- \pi_{li}- \pi_{ri} )\right\},
\end{multline*}
where 
\begin{multline*}
w_{li}^{(j)} = \mathbb{E}_{ \beta^{(j)}, \theta^{(j)}}[H_{li} \mid X_i, U_i] =  \frac{ \pi_{li}^{(j)} \cdot  h_{li}^{(j)}(U_i)}{ h_{X_i}^{(j)} (U_i) }, \\
 w_{ri}^{(j)} = \mathbb{E}_{\beta^{(j)}, \theta^{(j)}}[H_{ri} \mid X_i, U_i ]=  \frac{ \pi_{ri}^{(j)} \cdot  h_{ri}^{(j)}(U_i)}{ h_{X_i}^{(j)} (U_i) }.
\end{multline*}

{\bf M step}:  Compute
$
(\beta^{(j+1)}, \theta^{(j+1)} ) = \arg \max_{\beta, \theta} Q^{(j)}(\beta, \theta))
$. 
}
\KwOut{Estimated coefficients $\beta$ and $\theta$}
\end{algorithm}
}

\subsection{Updating the thresholding functions in finite-sample ZAP}  \label{app:EM_finite_update}

We recommend using \algref{updateThreshold} below to update the thresholding functions, which performs estimations of our beta-mixture model, although finite-sample FDR control is guaranteed as long as the conditions in the \thmref{ZAPfiniteControl} are met. As seen in \algref{updateThreshold}, assessor functions for the hypotheses are first constructed based on expression \eqref{assessor_form}, using an EM algorithm that acts on the masked data $\{\tilde{U}_{t, i}, X_i\}_{i = 1}^m$ (\appref{EM_finite_ZAP}) to estimate the parameters. Next, for each masked $i \in \mathcal{A}_t \cup \mathcal{R}_t$, evaluated assessor value $T_i'$ at whichever $U_i$ or $\widecheck{U}_i$ is closer to the extreme ends of the interval $(0,1)$ is computed, and among them the hypothesis $j$ with the largest such value is selected. This step aims to locate the hypothesis in the current masked set that is the most likely to be a true null if all masked hypotheses are presumed to be from the candidate rejection set $\mathcal{R}_t$. Finally, one of the two thresholding functions $s_{l,t}$ and $s_{r,t}$ will be updated, in a manner that satisfies condition $(ii)$ in \thmref{ZAPfiniteControl}, to give a different $s_{l,t+1}$ or $s_{r,t+1}$: If $U_j > 0.5$, $s_{r,t+1}$ will be updated from $s_{r,t}$ at the point $X_j$ as $s_{r,t+1} (X_j) = U_{j}\vee \widecheck{U}_{j}$, and remains the same at all other covariate values; otherwise, $s_{l,t+1}$ will update from $s_{l,t}$ in a similar fashion using the value $U_{j}\wedge \widecheck{U}_{j}$.
As such,   at the next step $t+1$, one of $\mathcal{A}_{t+1}$ or $\mathcal{R}_{t+1}$ will be shrunk by exactly one element which is $j$. This is intuitive since if $\widehat{\text{FDP}}_{finite}(t) > \alpha$ at step $t$, one would hope to reduce the size of $\mathcal{R}_t$.

 \begin{algorithm}[h]
\caption{Update thresholding functions at  step $t$  with Model \eqref{hu-beta}}
 \label{alg:updateThreshold} 
\KwIn{The masked data $\{\tilde{U}_{t, i}, X_i\}_{i = 1}^m$ }

 Compute $\{\hat{\theta}_l,\hat{\theta}_r,  \hat{\beta}_l, \hat{\beta}_r\}$    using the EM algorithm in \appref{EM_finite_ZAP}.  
 
 Construct $\{\hat{a}_{X_i}(\cdot)\}_{i \in  \mathcal{A}_t \cup \mathcal{R}_t}$ with \eqref{assessor_form} by setting the underlying parameters as the estimates in the prior step. \\
 
 Find $j \equiv \arg\max_{i \in \mathcal{A}_t \cup \mathcal{R}_t} T_i'$ for $T_i' = \hat{a}_{X_i}(U_i')$, where
 \[U_i' \equiv I(U_i < 0.5)U_i\wedge \widecheck{U}_i +  I(U_i \geq 0.5)U_i\vee \widecheck{U}_i \]
\\
 \eIf{ $U_j >  0.5$}{
$s_{r, t+1} (X_i) = I(i = j) (U_i\vee \widecheck{U}_i) + I( i \neq  j) s_{r, t} (X_i)$, 
$s_{l,t+1} \equiv s_{l,t}$;
   }{
$s_{l, t+1} (X_i) = I(i = j) (U_i\wedge \widecheck{U}_i) + I( i \neq  j) s_{l, t} (X_i)$, 
$s_{r,t+1} \equiv s_{r,t}$;
  }
\KwOut{$s_{l t+1}$, $s_{r,t+1}$}
%
%
%
%
\end{algorithm}

\subsection{EM algorithm for finite-sample ZAP}  \label{app:EM_finite_ZAP}

We will lay out aspects of the EM algorithm required for \algref{updateThreshold}.

\emph{E-step computations}. Let $\mathcal{D}_{ti}  = (X_i, \widetilde{U}_{t, i})$ be the available data for $i$ at step $t$ of the finite-sample ZAP algorithm. To update from the parameters $(\theta^{(j)}, \beta^{(j)})$ at the $j$-th EM iteration, we need to compute the following quantities:
\begin{multline} \label{different_E_quantities}
\mathbb{E}_{(\theta^{(j)}, \beta^{(j)})}[H_{li}| \mathcal{D}_{ti}], 
\;\ \mathbb{E}_{(\theta^{(j)}, \beta^{(j)})}[H_{r i}| \mathcal{D}_{ti}],\\
 \;\ \mathbb{E}_{(\theta^{(j)}, \beta^{(j)})}[H_{r i} \log (U_i)| \mathcal{D}_{ti}],
 \;\ \mathbb{E}_{(\theta^{(j)}, \beta^{(j)})}[H_{r i}\log (1 - U_i)| \mathcal{D}_{ti}],\\ 
  \;\ \mathbb{E}_{(\theta^{(j)}, \beta^{(j)})}[H_{l i} \log (U_i)| \mathcal{D}_{ti}],
 \;\ \mathbb{E}_{(\theta^{(j)}, \beta^{(j)})}[H_{l i}\log (1 - U_i)| \mathcal{D}_{ti}].
\end{multline}
These quantities are straightforward to compute when $\tilde{U}_{t, i}$ is a singleton, so we will only focus on computing them when  $\tilde{U}_{t, i}$ is a two-element set, i.e. corresponding to a masked $U_i$. Like \algref{EMasymp}, we shall let $\pi_{li}^{(j)}$, $\pi_{ri}^{(j)}$, $h_{li}^{(j)}$, $h_{ri}^{(j)}$ and $h_{X_i}^{(j)}$ be
as defined in \secref{betaMod} evaluated at $(\beta^{(j)}, \theta^{(j)})$. We will have
\begin{align*}
\hat{H}_{li}^{(j)} &\equiv \mathbb{E}_{(\theta^{(j)}, \beta^{(j)})}[H_{li}| \mathcal{D}_{ti}] =  P_{(\theta^{(j)}, \beta^{(j)})}[H_{li} = 1| \mathcal{D}_{ti}]= \frac{ \pi^{(j)}_{li} [h^{(j)}_{li} (U_i)+ h^{(j)}_{li} (\widecheck{U}_i)]}{h^{(j)}_{X_i} (\widecheck{U}_i)+ h^{(j)}_{X_i} (U_i)}\\
\hat{H}_{ri}^{(j)} &\equiv \mathbb{E}_{(\theta^{(j)}, \beta^{(j)})}[H_{ri}| \mathcal{D}_{ti}] =  P_{(\theta^{(j)}, \beta^{(j)})}[H_{ri} = 1| \mathcal{D}_{ti}]= \frac{ \pi^{(j)}_{ri} [h^{(j)}_{ri} (U_i)+ h^{(j)}_{ri} (\widecheck{U}_i)]}{h^{(j)}_{X_i} (\widecheck{U}_i)+ h^{(j)}_{X_i} (U_i)}.
\end{align*}
Moreover, to express the last four quantities in \eqref{different_E_quantities}, we define
\begin{align*}
y^{(j)}_{li, A}&\equiv \mathbb{E}_{\theta^{(j)}, \beta^{(j)}}[H_{l i} \log (U_i)| \mathcal{D}_{ti},H_{l i}  = 1 ] =\frac{h_{li}^{(j)} (U_i) \log (U_i) + h_{li}^{(j)} (\widecheck{U}_i)  \log (\widecheck{U}_i)}
{h_{li}^{(j)} (U_i) + h_{li}^{(j)} (\widecheck{U}_i)},\\
y^{(j)}_{ri, A}&\equiv \mathbb{E}_{\theta^{(j)}, \beta^{(j)}}[H_{ri} \log (U_i)| \mathcal{D}_{ti},H_{ri}  = 1 ] =\frac{h_{ri}^{(j)} (U_i) \log (U_i) + h_{ri}^{(j)} (\widecheck{U}_i)  \log (\widecheck{U}_i)}
{h_{li}^{(j)} (U_i) + h_{li}^{(j)} (\widecheck{U}_i)},\\
y^{(j)}_{li, B}&\equiv \mathbb{E}_{\theta^{(j)}, \beta^{(j)}}[H_{l i} \log (1 - U_i)| \mathcal{D}_{ti},H_{l i}  = 1 ] =\frac{h_{li}^{(j)} (U_i) \log (1 - U_i) + h_{li}^{(j)} (\widecheck{U}_i)  \log ( 1 - \widecheck{U}_i)}
{h_{li}^{(j)} (U_i) + h_{li}^{(j)} (\widecheck{U}_i)},\\
y^{(j)}_{ri, B}&\equiv\mathbb{E}_{\theta^{(j)}, \beta^{(j)}}[H_{r i} \log (1 - U_i)| \mathcal{D}_{ti},H_{r i}  = 1 ]= 
 \frac{h_{ri}^{(j)} (U_i) \log (1 - U_i) + h_{ri}^{(j)} (\widecheck{U}_i)  \log (1- \widecheck{U}_i)}
{h_{ri}^{(j)} (U_i) + h_{ri}^{(j)} (\widecheck{U}_i)},
\end{align*}
then one can express
\begin{multline*}
\mathbb{E}_{(\theta^{(j)}, \beta^{(j)})}[H_{l i} \log (U_i)| \mathcal{D}_{ti}] = y^{(j)}_{li, A} \hat{H}_{li}^{(j)},\quad 
 \mathbb{E}_{(\theta^{(j)}, \beta^{(j)})}[H_{r i} \log (U_i)| \mathcal{D}_{ti}] = y^{(j)}_{ri, A} \hat{H}_{ri}^{(j)}, \\
\mathbb{E}_{(\theta^{(j)}, \beta^{(j)})}[H_{l i} \log (1 - U_i)| \mathcal{D}_{ti}] =y^{(j)}_{li, B} \hat{H}_{li}^{(j)}, \quad \mathbb{E}_{(\theta^{(j)}, \beta^{(j)})}[H_{r i} \log (1 - U_i)| \mathcal{D}_{ti}] =y^{(j)}_{r i, B} \hat{H}_{ri}^{(j)}.
\end{multline*}

\emph{Initialization}. We now discuss how to specify values for $\pi^{(0)}_{li}, \pi^{(0)}_{ri}$, $\beta_l^{(0)}$ and $\beta_r^{(0)}$ to initialize the algorithm. Specifying $\beta_l^{(0)}$ and $\beta_r^{(0)}$ is easy: For  $\beta_l$, we only consider the left group
$
\mathfrak{L} = \{i: \widecheck{U}_i \vee U_i \leq 0.5 \}
$, and fit the left-leaning beta density $h_{li} (\cdot)$ to the data points $\{  U_i \wedge \widecheck{U}_i  : i \in \mathfrak{L}\}$ to obtain an estimate for $\beta_l$ as the initial value  $\beta_l^{(0)}$, with a given value for $\gamma_l$ (such as $4$). Note that we fit the model to the smaller point $U_i \wedge \widecheck{U}_i$ instead of $U_i \vee \widecheck{U}_i$ for each $i \in \mathfrak{L}$ with the goal of having a more ``aggressive" left alternative distribution. 
The initial value $\beta_r^{(0)}$ can be obtained similarly by considering 
the right group $\mathfrak{R} \equiv  \{i: \widecheck{U}_i \wedge U_i > 0.5 \}$, the data $\{  U_i \vee \widecheck{U}_i  : i \in \mathfrak{R}\}$ and $h_{ri} (\cdot)$ .

To specifiy $\pi^{(0)}_{li}, \pi^{(0)}_{ri}$, for each $i$, we first define
\begin{multline}
\pi_{ri}^+ \equiv P(H_{ri}= 1|X_i, U_i > 0.5),  \quad  \pi_{li}^+ \equiv P(H_{li}= 1|X_i, U_i > 0.5), \quad \pi^+_i \equiv P(U_i > 0.5|X_i),\\
\pi_{li}^- \equiv P(H_{li} = 1|X_i, U_i \leq  0.5), \quad  \pi_{ri}^- \equiv P(H_{ri} = 1|X_i, U_i \leq  0.5), \quad \pi^-_i \equiv P(U_i \leq  0.5|X_i)
\end{multline}
and note that, by definition, $\pi_{ri}\geq \pi_{ri}^+ \pi^+_i$ and $\pi_{li} \geq\pi_{li}^- \pi^-_i$. We will form estimates for $\hat{\pi}_{ri}^+, \hat{\pi}^+_i, \hat{\pi}_{li}^-,  \hat{\pi}^-_i$ and let $\pi^{(0)}_{ri} = \hat{\pi}_{ri}^+ \hat{\pi}^+_i$ and $\pi^{(0)}_{li} = \hat{\pi}_{li}^+ \hat{\pi}^+_i$ be conservative estimates for $\pi_{ri}$ and $\pi_{li}$. The estimates $\hat{\pi}^+_i$ and $\hat{\pi}^-_i$ can be obtained as predicted probabilities by fitting a logistic regression on the indicator responses $D_i \equiv I(U_i > 0.5)$ with covariates $X_i$. For the rest of this section we will focus on the estimates $\hat{\pi}_{ri}^+$, since the estimates $\hat{\pi}_{li}^-$ can be obtain analogously. 

Let $J_i \equiv I( \tilde{U}_{t, i} \text{ has one element})$. Then
\[
\mathbb{E} [J_i  | X_i, D_i = 1] = P( J_i = 1|X_i,  D_i = 1)\geq (1 - \pi_{ri}^+ - \pi_{li}^+  )\left(  \frac{0.5 - 2(1 - s_{r, t} (X_i))}{0.5}\right)
\]
which is equivalent to
\[
\pi_{ri}^+ + \pi_{li}^+ \geq \mathbb{E}\left[ \tilde{J}_i \Big| X_i, D_i = 1\right],
\]
where $\tilde{J}_i \equiv 1 - \frac{0.5 J_i}{0.5 - 2(1 - s_{r, t}(X_i))} $ and 
\begin{multline*}
\mathbb{E}\left[ \tilde{J}_i \Big| X_i, D_i = 1\right] = \underbrace{P( \tilde{U}_{t, i} \text{ has two elements } | X_i , D_i = 1)}_{  \equiv 1 -  \pi^+_{J_i = 1}} +  \\
\underbrace{P( \tilde{U}_{t, i} \text{ has one element } | X_i , D_i = 1) }_{ \pi^+_{J_i = 1}}\left(1 - \frac{0.5}{0.5 - 2(1 - s_{r, t}(X_i))} \right)
\end{multline*}
As the probability of $H_{li} = 1$ should be small under $U_i > 0.5$, $\pi_{li}^+$ is likely to be negligible, hence the right hand side of the previous display should still be a conservative estimate for $\pi_{ri}^+$, i.e. 
\[
\pi_{ri}^+ \approx (1 -  \pi^+_{J_i = 1}) +  \pi^+_{J_i = 1} \left(1 - \frac{0.5}{0.5 - 2(1 - s_{r, t}(X_i))} \right).
\]
Now estimates for $\pi^+_{J_i = 1}$ for any $i \in \{1, \dots, m\}$ can be obtained as the fitted probability of the logistic regression on $J_i$ with covariates $X_i$ restricted to samples with $U_i > 0.5$. 
%

%
%
%


\  \ 

 \ \



\section{Component probability estimated with the beta mixture} \label{app:betaModSim}

Following up on \secref{beta_expr}, we briefly test how well our beta-mixture model in \secref{betaMod} can be leveraged to robustly estimate the non-null probabilities under more setups. We generate 8000 i.i.d. $z$-values from a normal mixture model with the density
\begin{equation} \label{normMixMod}
(1  - w) \times f_0(z)  + \underbrace{(w \cdot (1 - \rho) )}_{\equiv w_l}\times \phi(z; \mu_l, 1) +  \underbrace{(w \cdot\rho)}_{\equiv w_r} \times \phi(z; \mu_r, 1),
\end{equation}
where the simulation parameters $\mu_l$, $\mu_r$, $w$, $\rho$ range as
\begin{multline*}
\mu_l \in \{-2.5, -2, -1.5, -1, -0.5\}, \quad \mu_r \in \{0.5, 1, 1.5, 2, 2.5\}, \\
w \in \{0.1, 0.15, 0.2\}, \quad \rho \in \{0.5, 0.7, 0.9\}.
\end{multline*}
Apparently,  $w$ is the non-null probability, $\mu_l$ and $\mu_r$ are respectively the mean parameters for the alternative normals on the left and right, and $\rho$ parametrizes the degree of asymmetry reflected in the mixing  probabilities $w_l$ and $w_r$. For each set of $8000$ z-values, the beta mixture model \eqref{hu-beta} for $(\gamma_l, \gamma_r) = (4, 4)$  is fitted with regression intercepts only by an EM algorithm, and the resulting left and right model-based non-null probabilities $\hat{\pi}_l$ and $\hat{\pi}_r$ estimates serve as estimates for $w_l$ and $w_r$. 

\begin{table}[ht]
\centering
\caption[Table caption text]{Estimated probabilities $\hat{\pi}_l$ and $\hat{\pi}_r$ based on the beta mixture \eqref{hu-beta}  with $(\gamma_l, \gamma_r) = (4, 4)$ and regression intercepts only, for 8000  $z$-values generated by  the normal mixture    \eqref{normMixMod} with $\rho = 0.5$. $\hat{\pi}_l$ and $\hat{\pi}_r$ are respectively the left and right entries in each cell. 
}
\begin{tabular}{c*{5}{|cc}}
  \hline
$\mu_l \backslash \mu_r$ & \multicolumn{2}{c|}{ $ 0.5$}  & \multicolumn{2}{c|}{ $ 1$}  & \multicolumn{2}{c|}{ $ 1.5$} & \multicolumn{2}{c|}{ $ 2$}  & \multicolumn{2}{c}{ $ 2.5$}  \\ 
  \hline
 & \multicolumn{10}{c}{ $w = 0.1$, $(w_l, w_r) = (0.05, 0.05)$} \\ 
  \hline
-2.5 & 0.073 & 0.056 & 0.066 & 0.055 & 0.064 & 0.064 & 0.071 & 0.065 & 0.080 & 0.069 \\ 
  -2 & 0.061 & 0.037 & 0.062 & 0.038 & 0.057 & 0.058 & 0.059 & 0.061 & 0.071 & 0.079 \\ 
  -1.5 & 0.040 & 0.022 & 0.048 & 0.037 & 0.042 & 0.047 & 0.049 & 0.065 & 0.058 & 0.070 \\ 
  -1 & 0.044 & 0.029 & 0.037 & 0.039 & 0.034 & 0.046 & 0.047 & 0.064 & 0.066 & 0.072 \\ 
  -0.5 & 0.023 & 0.024 & 0.024 & 0.047 & 0.041 & 0.073 & 0.043 & 0.072 & 0.041 & 0.076 \\  \hline
   & \multicolumn{10}{c}{  $w = 0.15$, $(w_l, w_r) = (0.075, 0.075)$} \\ \hline
-2.5& 0.110 & 0.050 & 0.103 & 0.072 & 0.100 & 0.087 & 0.106 & 0.095 & 0.102 & 0.095 \\ 
 -2& 0.085 & 0.045 & 0.087 & 0.058 & 0.089 & 0.070 & 0.086 & 0.091 & 0.100 & 0.096 \\ 
    -1.5& 0.080 & 0.052 & 0.076 & 0.065 & 0.082 & 0.064 & 0.074 & 0.082 & 0.082 & 0.094 \\ 
-1& 0.065 & 0.040 & 0.044 & 0.041 & 0.041 & 0.064 & 0.068 & 0.087 & 0.068 & 0.102 \\ 
  -0.5 & 0.022 & 0.024 & 0.049 & 0.065 & 0.050 & 0.075 & 0.041 & 0.071 & 0.050 & 0.091 \\  \hline
    & \multicolumn{10}{c}{ $w= 0.2$, $(w_l, w_r) = (0.1, 0.1)$} \\ \hline
-2.5 & 0.116 & 0.073 & 0.118 & 0.101 & 0.123 & 0.119 & 0.127 & 0.132 & 0.122 & 0.136 \\ 
  -2 & 0.112 & 0.053 & 0.113 & 0.086 & 0.115 & 0.127 & 0.113 & 0.132 & 0.120 & 0.133 \\ 
  -1.5 & 0.089 & 0.047 & 0.090 & 0.074 & 0.086 & 0.104 & 0.096 & 0.124 & 0.086 & 0.135 \\ 
  -1 & 0.067 & 0.042 & 0.055 & 0.068 & 0.060 & 0.095 & 0.071 & 0.117 & 0.090 & 0.134 \\ 
  -0.5 & 0.025 & 0.046 & 0.042 & 0.083 & 0.032 & 0.106 & 0.068 & 0.124 & 0.069 & 0.142 \\
   \hline
\end{tabular}
      \label{tab:betaModSimTab1}
\end{table}
      
  \begin{table}[ht]
\centering
\caption[Table caption text]{Estimated probabilities $\hat{\pi}_l$ and $\hat{\pi}_r$ based on the beta mixture \eqref{hu-beta}  with $(\gamma_l, \gamma_r) = (4, 4)$ and regression intercepts only, for 8000  $z$-values generated by  the normal mixture   \eqref{normMixMod} with $\rho = 0.7$. $\hat{\pi}_l$ and $\hat{\pi}_r$ are respectively the left and right entries in each cell. 
 }
\begin{tabular}{c*{5}{|cc}}
  \hline
$\mu_l \backslash \mu_r$ & \multicolumn{2}{c|}{ $ 0.5$}  & \multicolumn{2}{c|}{ $ 1$}  & \multicolumn{2}{c|}{ $ 1.5$} & \multicolumn{2}{c|}{ $ 2$}  & \multicolumn{2}{c}{ $ 2.5$}  \\ 
  \hline
 & \multicolumn{10}{c}{ $w = 0.1$, $(w_l, w_r) = (0.03, 0.07)$} \\ 
  \hline
-2.5 & 0.049 & 0.043 & 0.037 & 0.046 & 0.047 & 0.060 & 0.050 & 0.079 & 0.058 & 0.091 \\ 
  -2 & 0.046 & 0.042 & 0.036 & 0.059 & 0.037 & 0.072 & 0.046 & 0.077 & 0.042 & 0.089 \\ 
  -1.5 & 0.065 & 0.042 & 0.023 & 0.044 & 0.028 & 0.072 & 0.051 & 0.076 & 0.037 & 0.089 \\ 
  -1 & 0.036 & 0.041 & 0.057 & 0.061 & 0.034 & 0.072 & 0.049 & 0.086 & 0.047 & 0.087 \\ 
  - 0.5 & 0.031 & 0.040 & 0.037 & 0.050 & 0.042 & 0.077 & 0.036 & 0.074 & 0.030 & 0.091 \\  \hline
   & \multicolumn{10}{c}{  $w = 0.15$, $(w_l, w_r) = (0.045, 0.105)$} \\ \hline
-2.5 & 0.066 & 0.054 & 0.057 & 0.058 & 0.062 & 0.106 & 0.058 & 0.124 & 0.070 & 0.137 \\ 
  -2 & 0.047 & 0.051 & 0.057 & 0.079 & 0.058 & 0.112 & 0.067 & 0.124 & 0.057 & 0.134 \\ 
  -1.5 & 0.043 & 0.055 & 0.044 & 0.075 & 0.034 & 0.091 & 0.048 & 0.118 & 0.050 & 0.128 \\ 
  -1 & 0.053 & 0.058 & 0.024 & 0.069 & 0.049 & 0.100 & 0.053 & 0.129 & 0.069 & 0.136 \\ 
  -0.5 & 0.038 & 0.057 & 0.034 & 0.076 & 0.034 & 0.111 & 0.035 & 0.125 & 0.045 & 0.135 \\  \hline
    & \multicolumn{10}{c}{ $w = 0.2$, $(w_l, w_r) = (0.06, 0.14)$} \\ \hline
   -2.5 & 0.078 & 0.070 & 0.079 & 0.116 & 0.088 & 0.152 & 0.079 & 0.159 & 0.088 & 0.185 \\ 
  -2 & 0.062 & 0.050 & 0.071 & 0.104 & 0.072 & 0.143 & 0.075 & 0.164 & 0.073 & 0.181 \\ 
  -1.5 & 0.049 & 0.052 & 0.058 & 0.114 & 0.059 & 0.146 & 0.059 & 0.164 & 0.074 & 0.184 \\ 
  -1 & 0.029 & 0.041 & 0.028 & 0.088 & 0.046 & 0.128 & 0.075 & 0.174 & 0.087 & 0.192 \\ 
  -0.5 & 0.031 & 0.059 & 0.021 & 0.106 & 0.026 & 0.137 & 0.050 & 0.164 & 0.059 & 0.186 \\ 
   \hline
\end{tabular}
      \label{tab:betaModSimTab2}
\end{table}

  \begin{table}[ht]
\centering
\caption[Table caption text]{Estimated probabilities $\hat{\pi}_l$ and $\hat{\pi}_r$ based on the beta mixture \eqref{hu-beta}  with $(\gamma_l, \gamma_r) = (4, 4)$ and regression intercepts only, for 8000  $z$-values generated by  the normal mixture    \eqref{normMixMod} with $\rho = 0.9$. $\hat{\pi}_l$ and $\hat{\pi}_r$ are respectively the left and right entries in each cell. 
}
\begin{tabular}{c*{5}{|cc}}
  \hline
$\mu_l \backslash \mu_r$ & \multicolumn{2}{c|}{ $ 0.5$}  & \multicolumn{2}{c|}{ $ 1$}  & \multicolumn{2}{c|}{ $ 1.5$} & \multicolumn{2}{c|}{ $ 2$}  & \multicolumn{2}{c}{ $ 2.5$}  \\ 
  \hline
 & \multicolumn{10}{c}{ $w = 0.1$, $(w_l, w_r) = (0.01, 0.09)$} \\ 
  \hline
-2.5 & 0.019 & 0.044 & 0.017 & 0.072 & 0.015 & 0.093 & 0.030 & 0.100 & 0.026 & 0.114 \\ 
  -2 & 0.037 & 0.065 & 0.008 & 0.064 & 0.017 & 0.102 & 0.009 & 0.108 & 0.035 & 0.123 \\ 
  -1.5 & 0.022 & 0.050 & 0.011 & 0.071 & 0.018 & 0.086 & 0.034 & 0.117 & 0.053 & 0.120 \\ 
  -1 & 0.019 & 0.044 & 0.019 & 0.066 & 0.021 & 0.079 & 0.026 & 0.112 & 0.041 & 0.125 \\ 
  -0.5 & 0.010 & 0.050 & 0.011 & 0.062 & 0.019 & 0.092 & 0.035 & 0.120 & 0.049 & 0.116 \\  \hline
   & \multicolumn{10}{c}{  $w = 0.15$, $(w_l, w_r) = ( 0.015, 0.135)$} \\ \hline
 -2.5 & 0.024 & 0.060 & 0.022 & 0.096 & 0.024 & 0.115 & 0.023 & 0.149 & 0.024 & 0.164 \\ 
  -2 & 0.020 & 0.071 & 0.021 & 0.098 & 0.020 & 0.128 & 0.017 & 0.150 & 0.035 & 0.167 \\ 
  -1.5 & 0.029 & 0.080 & 0.015 & 0.092 & 0.025 & 0.141 & 0.035 & 0.148 & 0.035 & 0.156 \\ 
  -1 & 0.021 & 0.060 & 0.026 & 0.112 & 0.017 & 0.124 & 0.044 & 0.150 & 0.034 & 0.163 \\ 
  -0.5 & 0.011 & 0.065 & 0.017 & 0.105 & 0.012 & 0.127 & 0.018 & 0.139 & 0.032 & 0.158 \\  \hline
    & \multicolumn{10}{c}{ $w = 0.2$, $(w_l, w_r) = (0.02, 0.18)$} \\ \hline
  -2.5 & 0.018 & 0.059 & 0.024 & 0.134 & 0.021 & 0.162 & 0.026 & 0.194 & 0.039 & 0.223 \\ 
  -2 & 0.015 & 0.046 & 0.017 & 0.125 & 0.026 & 0.162 & 0.020 & 0.196 & 0.043 & 0.226 \\ 
  -1.5 & 0.012 & 0.072 & 0.010 & 0.134 & 0.027 & 0.154 & 0.039 & 0.214 & 0.032 & 0.216 \\ 
  -1 & 0.014 & 0.077 & 0.010 & 0.120 & 0.021 & 0.175 & 0.037 & 0.213 & 0.035 & 0.225 \\ 
  -0.5 & 0.016 & 0.104 & 0.013 & 0.115 & 0.023 & 0.167 & 0.033 & 0.209 & 0.049 & 0.230 \\ 
   \hline
\end{tabular}
      \label{tab:betaModSimTab3}
\end{table}    
      
From the results in  \tabsref{betaModSimTab1} - \tabssref{betaModSimTab3}, one can see that our beta mixture produces fairly accurate non-null probability estimates $(\hat{\pi}_l, \hat{\pi}_r)$ for $(w_l, w_r)$. Generally speaking, the estimates are the most inaccurate when one of $\mu_l$ or $\mu_r$ has a small magnitude, which is reasonable since one of the two non-null components has a weak signal and many z-values which are non-nulls could be regarded as null by the EM fitting algorithm. This should not be too concerning, as it simply means some of the hypotheses pose hard testing problems to begin with. 

\clearpage
\section{Additional numerical results on simulated data} \label{app:more_simulation_results}


In addition to the seven methods already introduced in \secref{sim_data}, we have also experimented with another seven methods:

 \begin{enumerate}[(a)]
       \item BH:  the vanilla BH procedure \citep{benjamini1995controlling}
     
    \item BL:  Boca and Leek procedure  \citep{boca2018direct}. 
    \item SK:  The  signed-knockoff  procedure by \citet{tian2021powerful}, which is a $z$-value based data masking procedure similar to ZAP (finite) and AdaPT-$\text{GMM}_g$, except that it doesn't leverage covariate information.
     \item CAMT with sign:  this is almost exactly the same as CAMT in \secref{sim_data}, but  the sign of the $z$-value $\text{sgn}(Z_i)$ is inserted as an extra covariate. 
      \item AdaPT with sign:  this is almost exactly the same as AdaPT in \secref{sim_data}, but the sign of the $z$-value $\text{sgn}(Z_i)$ is inserted  as an extra covariate. 
        \item oracle: The oracle procedure $\pmb\delta^\mathcal{Z}$.
 \end{enumerate}

 Of these additional methods, BH and  BL are $p$-value based whereas the rest are $z$-value based, although SK is a $z$-value based method that doesn't leverage covariate information. 
Note that CAMT with sign and AdaPT with sign are considered $z$-value based because with the incorporation of $\{\text{sgn}(Z_i)\}_{i=1}^m$ as covariates, they essentially operate on the full data $\{Z_i, X_i\}_{i=1}^m$ since  the  information lost from collapsing into two-sided $p$-values is recovered. We now extend the simulation studies under Setups 1-3 in \secref{sim_data} by exploring more choices for the simulation parameters of the data generating mechanism in \eqref{normMix}, under which, in addition to the classical BH procedure, these methods offer  finite-sample FDR control:
\begin{equation}\label{finitesamplemethods}
\text{AdaPT}, \quad \text{AdaPT-$\text{GMM}_g$}, \quad   \text{ AdaPT with sign}, \quad \text{ZAP (finite)}, \quad \text{SK}.
\end{equation}
%
Before proceeding,  we point out that no matter the choices of the simulation parameters, no current methods  have demonstrated near-optimal power under Setup 3 where the covariates influence the alternative means; as discussed in \secref{discuss}, we see closing this gap as a future research problem. With respect to the simulation studies below, our general conclusion is as follows: While the data masking methods in \eqref{finitesamplemethods} have the attractive finite-sample FDR controlling property, it always comes at a cost, such as the instability in power resulting from data masking (see \secref{discuss}). On the contrary, the asymptotic ZAP, which is justified by regular asymptotics (\thmref{ZAPasympControl}) but not a strong assumption like \eqref{FDRregAssumption}, demonstrates robust performances both in terms of FDR control and power throughout.

 \subsection{All-method comparison for the baseline simulations in \secref{sim_data}} \label{app:baseline_all_methods}

 \figref{all_setup_append} extends the results in  \figref{all_setup} by also displaying the performances of the above added methods under the simulation setups in \secref{sim_data}.  

In Setup 1, note that FDRreg in fact has slightly more power than the oracle procedure when the covariates are the most informative, which, admittedly, comes at the expense of violating the FDR bound. Also, SK, as a $z$-value based method,  has increasing power as the effect size increases due to its capitalizing on the increasing asymmetry of the $z$-values, a known phenomenon in the literature  \citep{sun2007oracle, storey2007optimal}; however, that it doesn't leveraging covariate information has put a cap on its power in comparison with the $z$-value based covariate-adaptive methods (ZAP, AdaPT-$\text{GMM}_g$, CAMT with sign and AdaPT with sign).

 In Setup 2, CAMT with sign and AdaPT with sign cannot match the power of the latest $z$-value covariate-adaptive methods (ZAP and AdaPT-$\text{GMM}_g$). This is not unexpected, because as described in \secref{sim_data}, the interaction between the $z$-values and the covariates (i.e. $f_{1, x}$ changes from concentrating on the -ve $z$-values to the +ve $z$-values as $x_\bullet$ goes from -ve to +ve) in this setup can be more accurately described by the $z$-value based working model employed by ZAP and AdaPT-$\text{GMM}_g$; a simple plug-in of the sign as a covariate into AdaPT and CAMT is not sufficient to describe such an interaction, because  the two-component beta-mixture model employed by these method tries to fit the interaction between $P_i$ and $\{X_i, \text{sgn}(Z_i)\}$ to no avail.   
 
 In Setup 3, no existing  methods has power nearing  the
oracle procedure, despite ZAP attaining top performances here; again, refer to our discussion in \secref{discuss}.
 
  \begin{figure}[t]
\centering
\includegraphics[width=\textwidth]{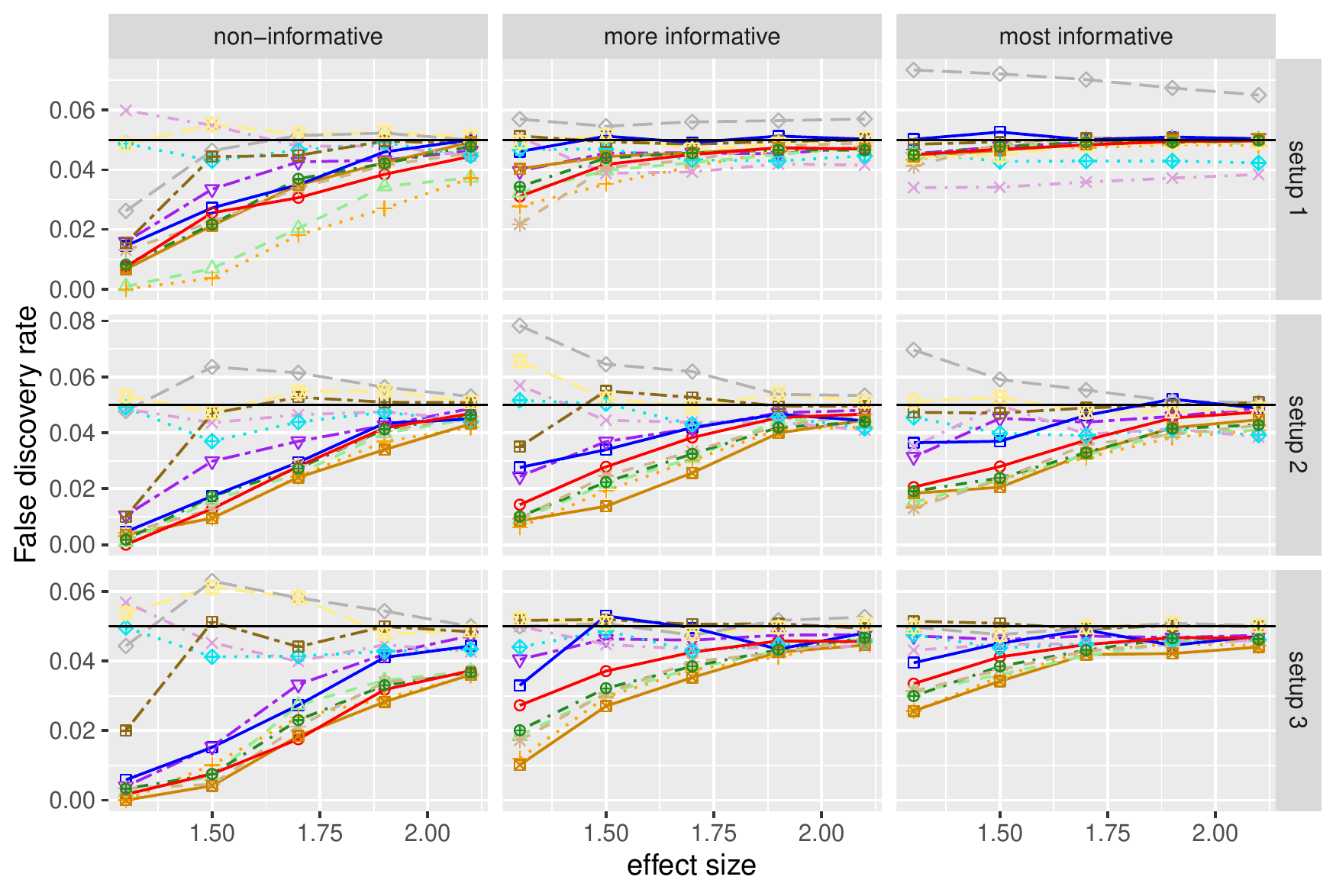}
\includegraphics[width=\textwidth]{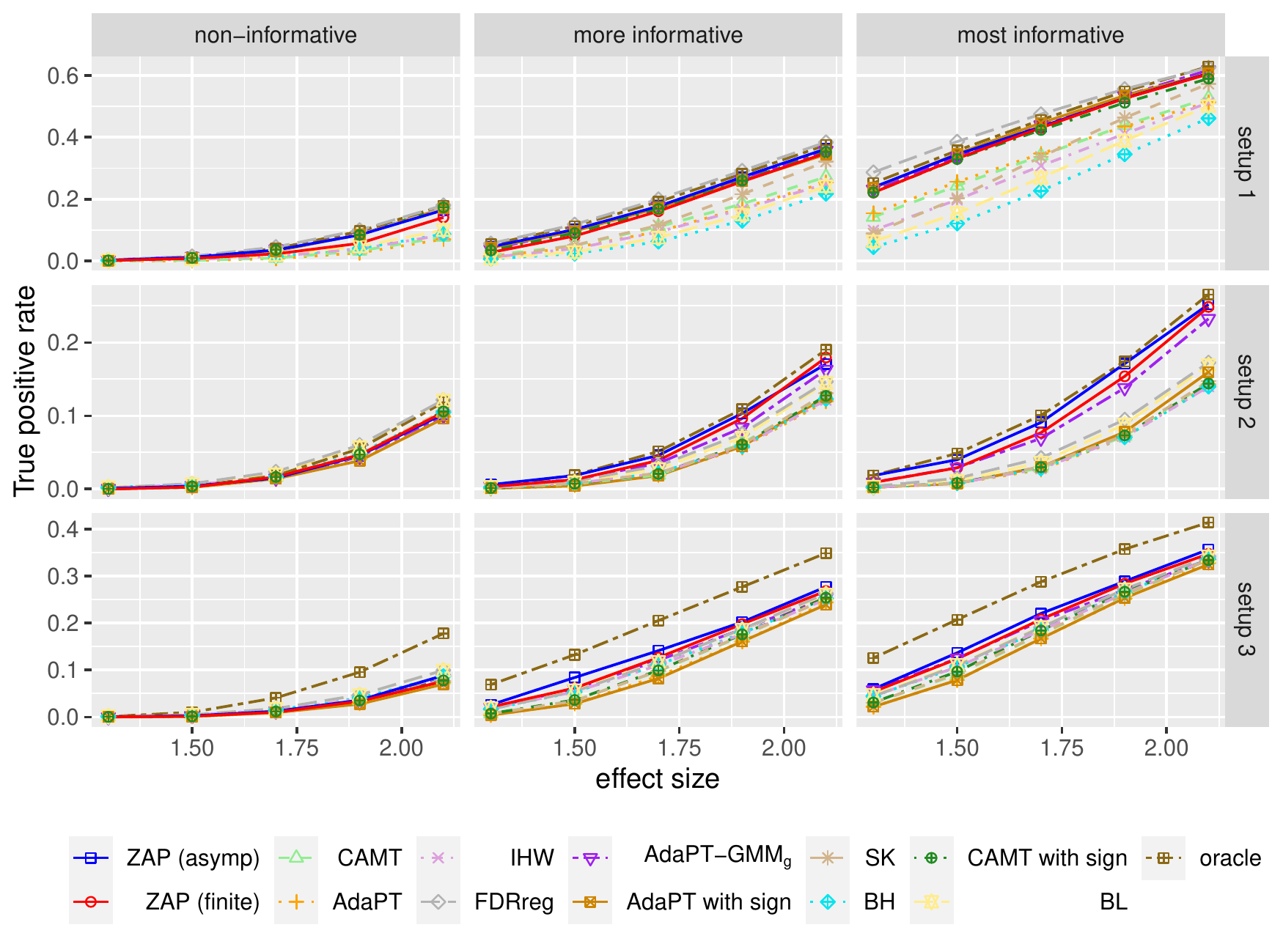}
\caption{
\emph{Baseline setting}: FDR and TPR performances of the full list  of methods for the simulations in \secref{sim_data}. All methods are applied at a targeted FDR level of $0.05$ (the horizontal black lines). The x-axes and panels are labelled as in \figref{all_setup}.
 }
\label{fig:all_setup_append}
\end{figure}

 \subsection{Different alternative variances} \label{app:adjusting_alt_var}
 
 \figsref{all_methods_alt_var_half} and \figssref{all_methods_alt_var_two} show the performances of different methods by setting the alternative component variance $\sigma^2$ in \eqref{normMix} as  $0.5$ and $2$ respectively; they help gauge how well our working beta-mixture model can fit different shapes of the alternative distributions. The two ZAP methods maintain top power performances in both scenarios, with the empirical FDR of asymptotic ZAP only overshoots the desired $0.05$ bound very slightly for Setup 1 when $\sigma^2 = 2$. Some other methods violate the desired FDR bound to varying degree; in particular, when $\sigma^2= 0.5$,  FDRreg's empirical FDR violates the bound of $0.05$ drastically in seven panels. This makes sense and is in agreement with the analogous simulation results in \citet[Figure A2]{zhang2020covariate}, as a reduced alternative variance can amplify the influence of the covariates on the overall non-null density in  \eqref{setup_cond_alt_density} and the strong assumption \eqref{FDRregAssumption} employed by FDRreg is more likely to invalidate the FDR control. 
 
\emph{Remark}: In theory, the BH procedure  has exact finite-sample FDR control in these setups, but is  empirically seen to moderately violate the FDR bound in some of the panels in  \figref{all_methods_alt_var_half}. This can happen with a limited number of simulated repetitions; our choice of $150$ is limited by our available computing facilities and the running times of non-BH methods (\appref{speed}). We have independently verified that, once the number of repetitions is brought up to $1000$, BH's empirical FDR drops below $0.05$.

  \begin{figure}[t]
\centering
\includegraphics[width=\textwidth]{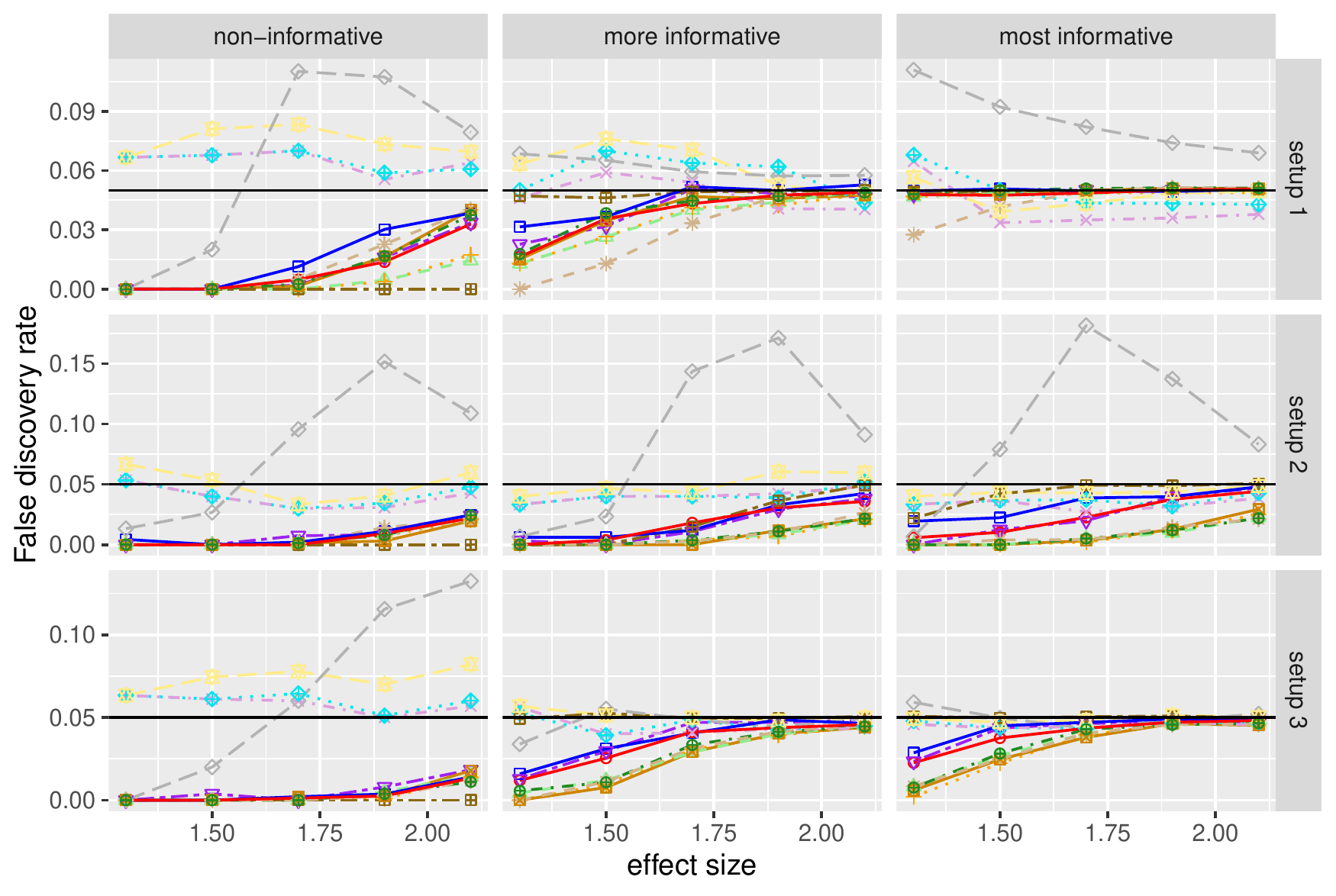}
\includegraphics[width=\textwidth]{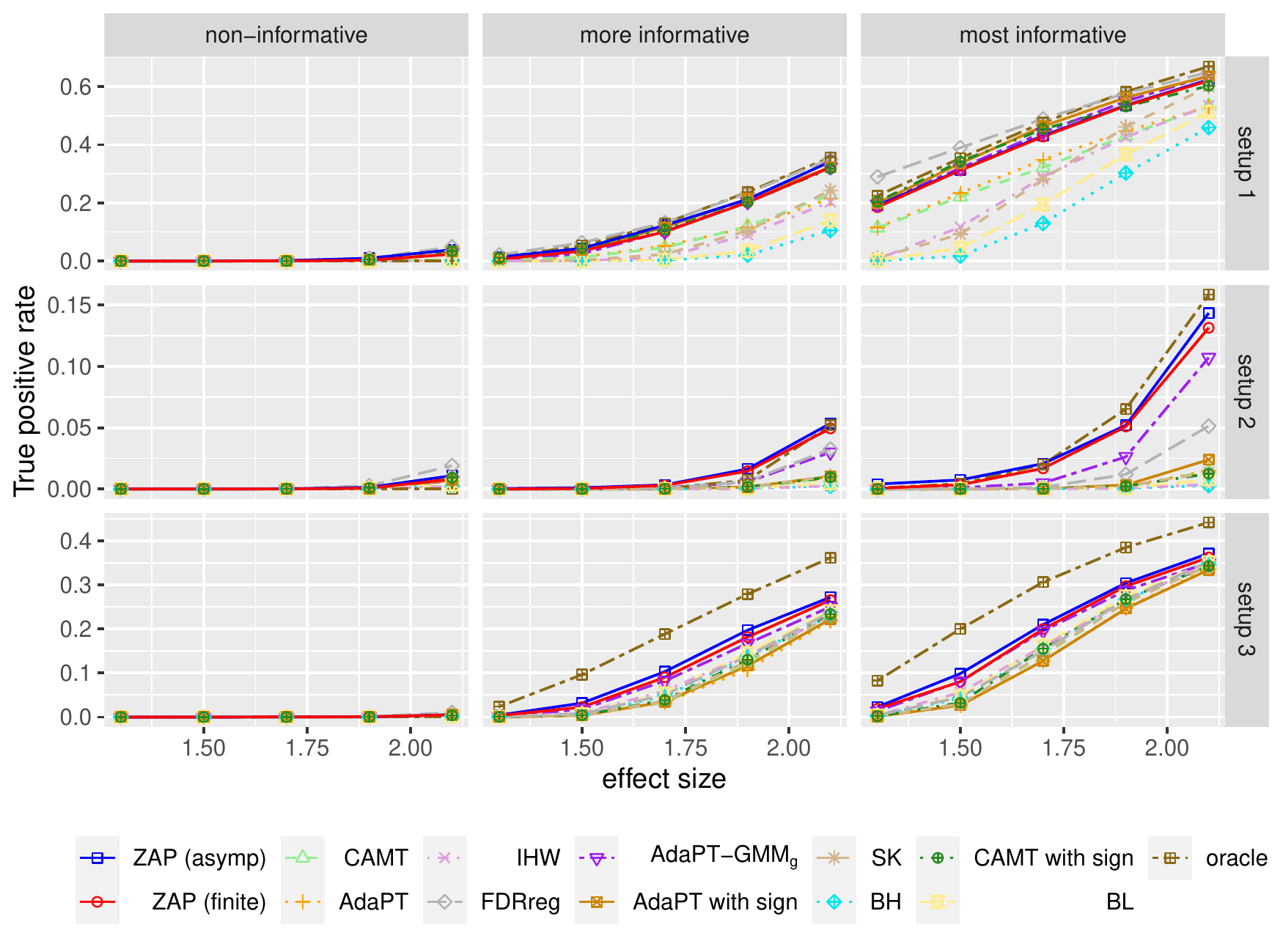}
\caption{
\emph{$\sigma^2= 0.5$}: FDR and TPR performances of the full list  of methods for the simulations in \secref{sim_data}, except that the alternative variance is set as $\sigma^2 = 0.5$. All methods are applied at a targeted FDR level of $0.05$ (the horizontal black lines). The x-axes and panels are labelled as in \figref{all_setup}.
 }
\label{fig:all_methods_alt_var_half}
\end{figure}

  \begin{figure}[t]
\centering
\includegraphics[width=\textwidth]{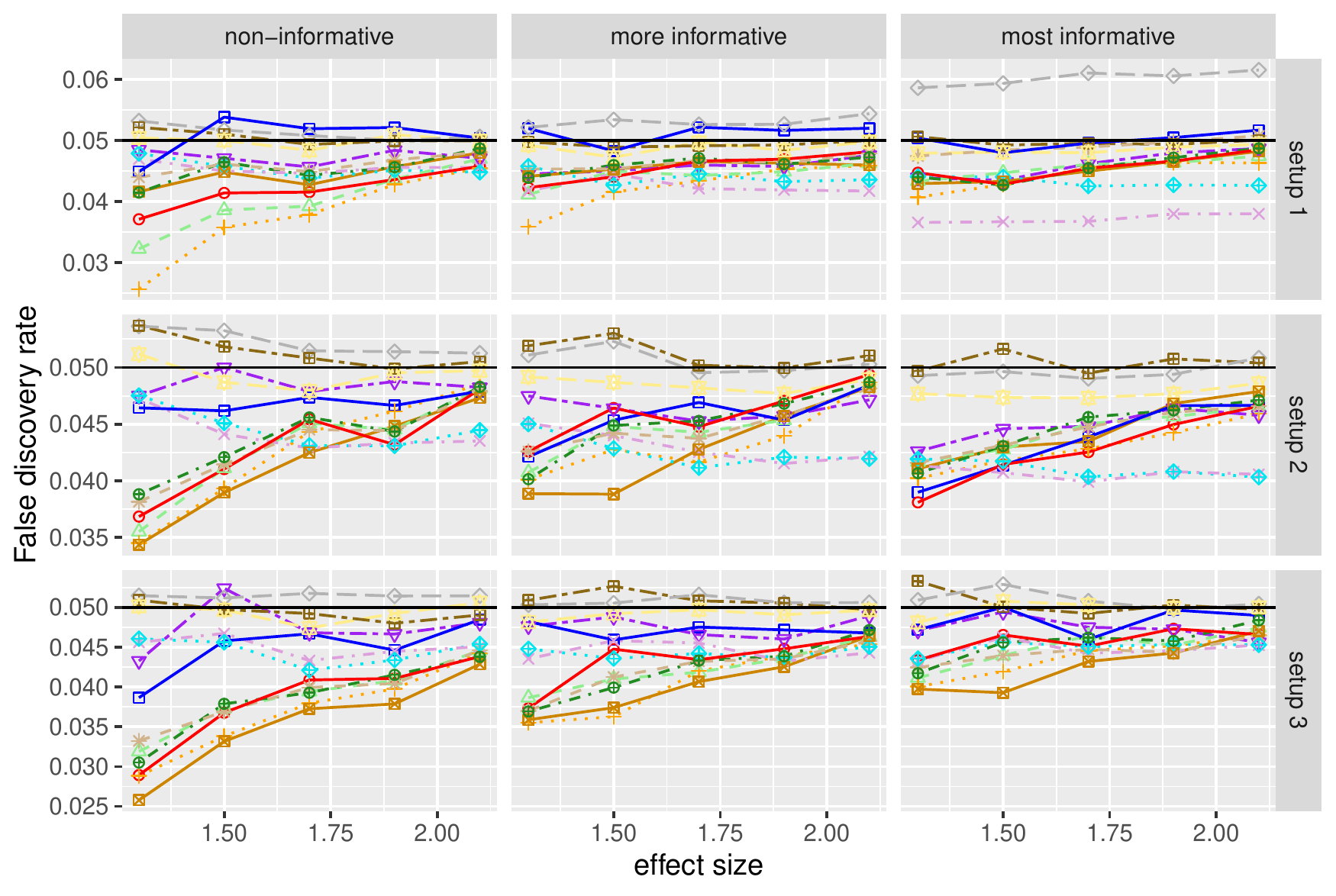}
\includegraphics[width=\textwidth]{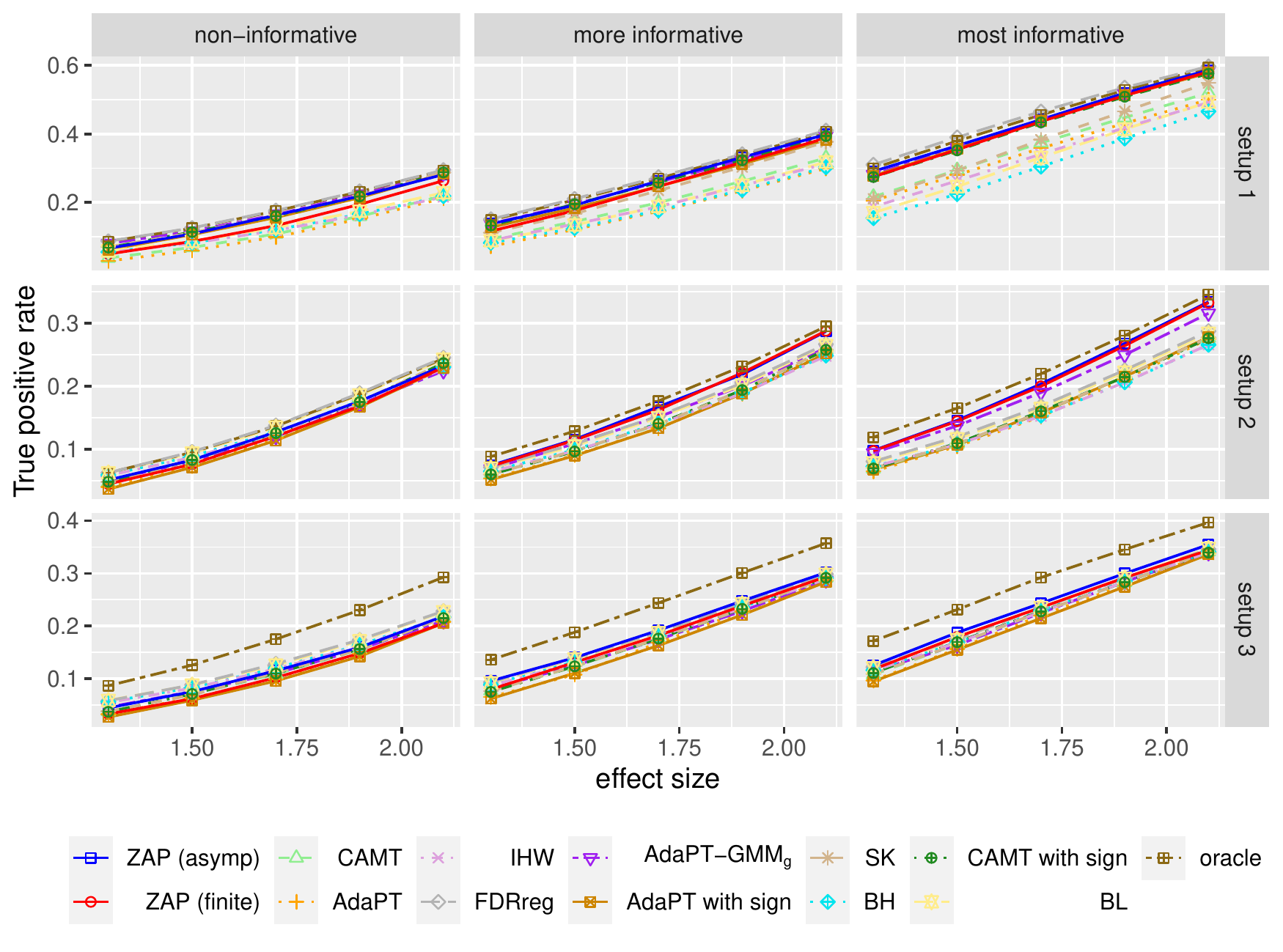}
\caption{
\emph{$\sigma^2= 2$}: FDR and TPR performances of the full list  of methods for the simulations in \secref{sim_data} except that the alternative variance is set as $\sigma^2 = 2$. All methods are applied at a targeted FDR level of $0.05$ (the horizontal black lines). The x-axes and panels are labelled as in \figref{all_setup}.
 }
\label{fig:all_methods_alt_var_two}
\end{figure}

\subsection{Sparse signals} \label{app:sparse_signal}

\figsref{moderate_sparse} and \figssref{sparsest}  show the performances of the methods when the setups of \secref{sim_data} contain sparser signals by setting the sparsity parameter $\eta$ to be smaller.  In the moderately sparse settings of  \figref{moderate_sparse}, when the covariates are completely uninformative ($\zeta = 0$),  Setups 1 and 3 have a signal density of $7.5 \%$ and Setup 2 has a signal density of $9.0 \%$. By and large, the methods of ZAP maintain top performances there. 

In the sparsest settings of \figref{sparsest}, Setups 1 and 3 have a signal density of $2.9 \%$ and Setup 2 has a signal density of $3.5 \%$ when $\zeta = 0$; the methods of ZAP, while still being quite competitive, lose their places as the top power performers in a number of  scenes, notably to FDRreg, BL, IHW and BH (note that BH  always controls the FDR under $0.05$ even visually may not look so, as explained at the end of \appref{adjusting_alt_var}).  FDRreg, although being powerful at times, can easily violate the FDR target due to the strong assumption \eqref{FDRregAssumption} as we have seen.  In addition, the BC-type methods (ZAP and all methods belonging to the CAMT or AdaPT families) are known to be generally more conservative than the BH-type methods (IHW, BH, BL) when the potential number of rejections is small (due to sparsity in this case); this is because  the BC-type methods generally form a more conservative FDR estimators for calibrating the rejection thresholds than the BH-type methods, a fact that has been exposited at length in other papers \citep{zhang2020covariate, ignatiadis2021covariate, chao2021adapt} and not repeated here for brevity. However, one can argue that the ZAP methods (and other BC-type methods) having a \emph{realized}   FDR much lower than the target is an advantage, as fewer-than-expected false discoveries are made. Visually, IHW and BL, both of which do not have finite-sample FDR guarantee, moderately overshoot the FDR bound of $0.05$   when the the covariates are non-informative.

 \begin{figure}[t]
\centering
\includegraphics[width=\textwidth]{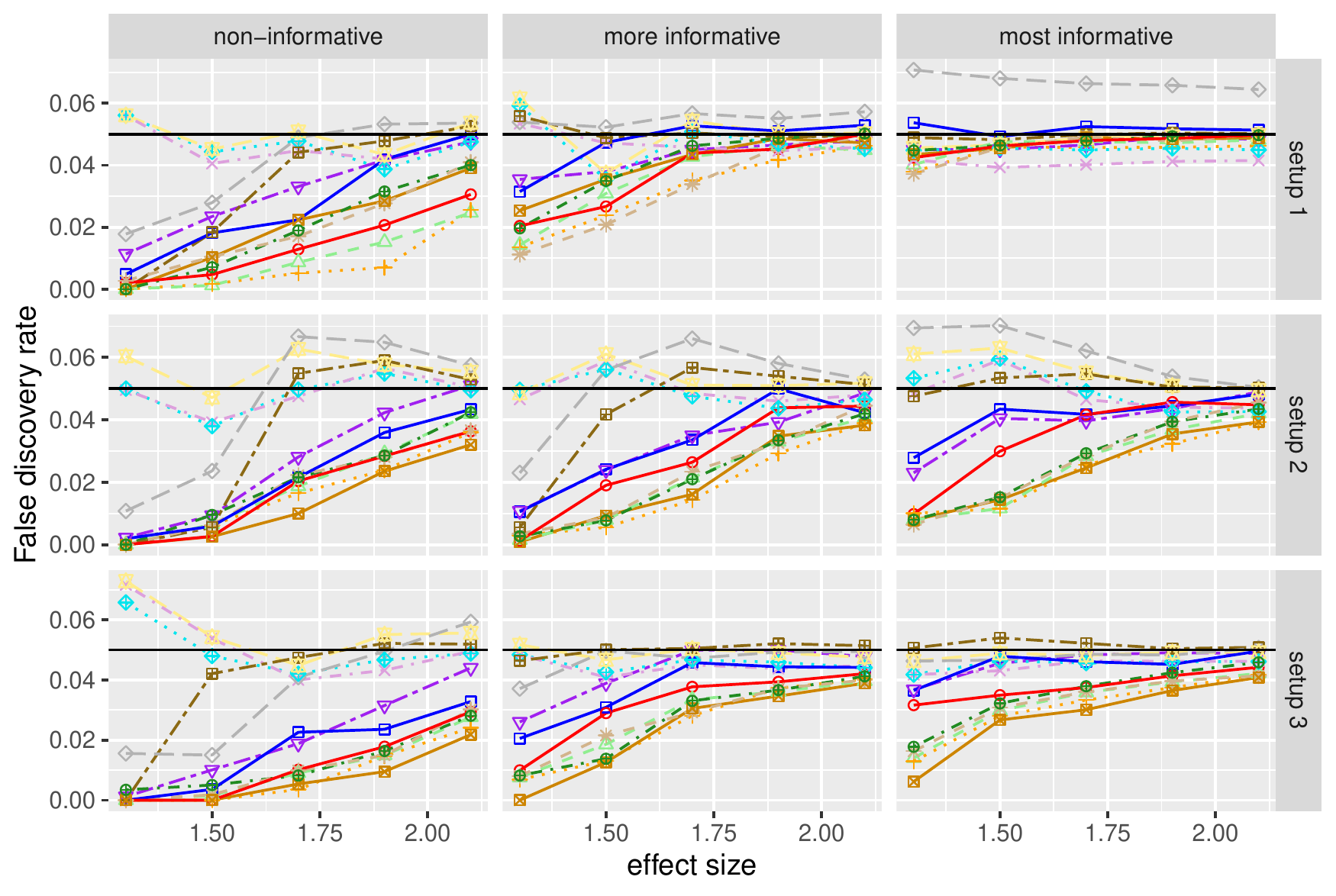}
\includegraphics[width=\textwidth]{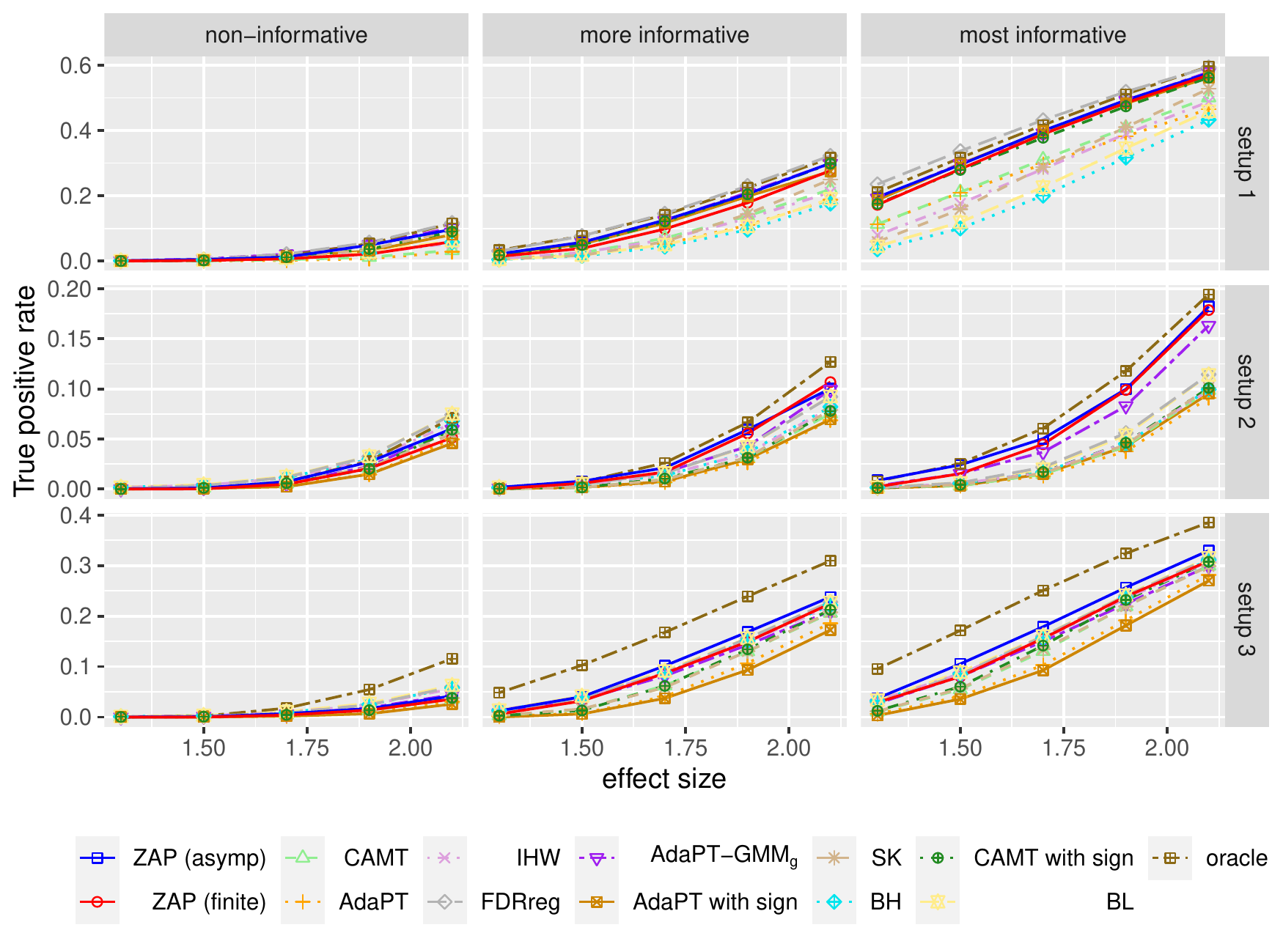}
\caption{
\emph{Moderately sparse setting}: FDR and TPR performances of the full list  of methods for the simulations in \secref{sim_data}, but with the density parameter $\eta$ set to $-2.5$ for Setups 1 and 3 and to $-3$ for Setup 2. All methods are applied at a targeted FDR level of $0.05$ (horizontal black lines). The x-axes and panels are labelled as in \figref{all_setup}.
 }
\label{fig:moderate_sparse}
\end{figure}
 \begin{figure}[t]
\centering
\includegraphics[width=\textwidth]{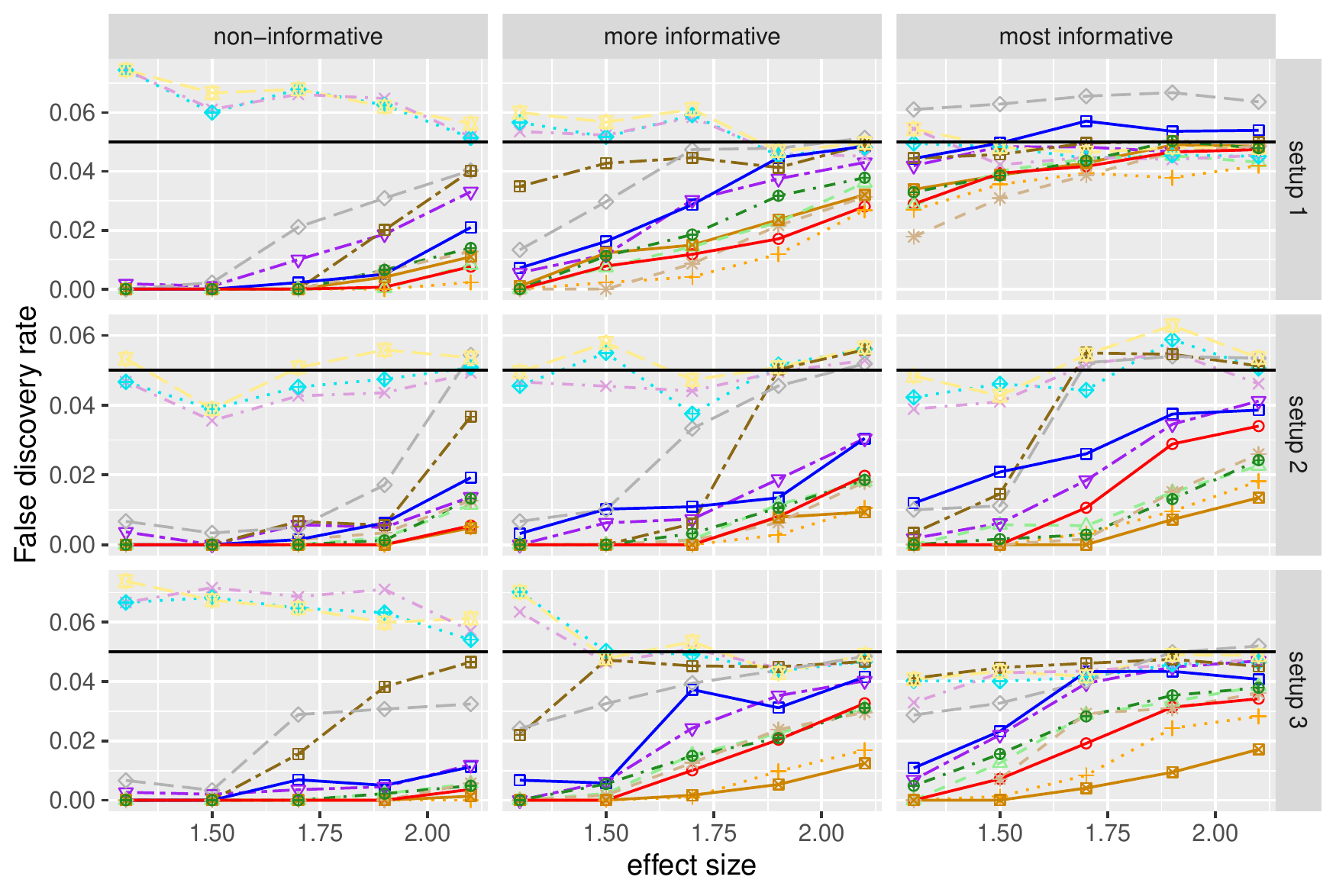}
\includegraphics[width=\textwidth]{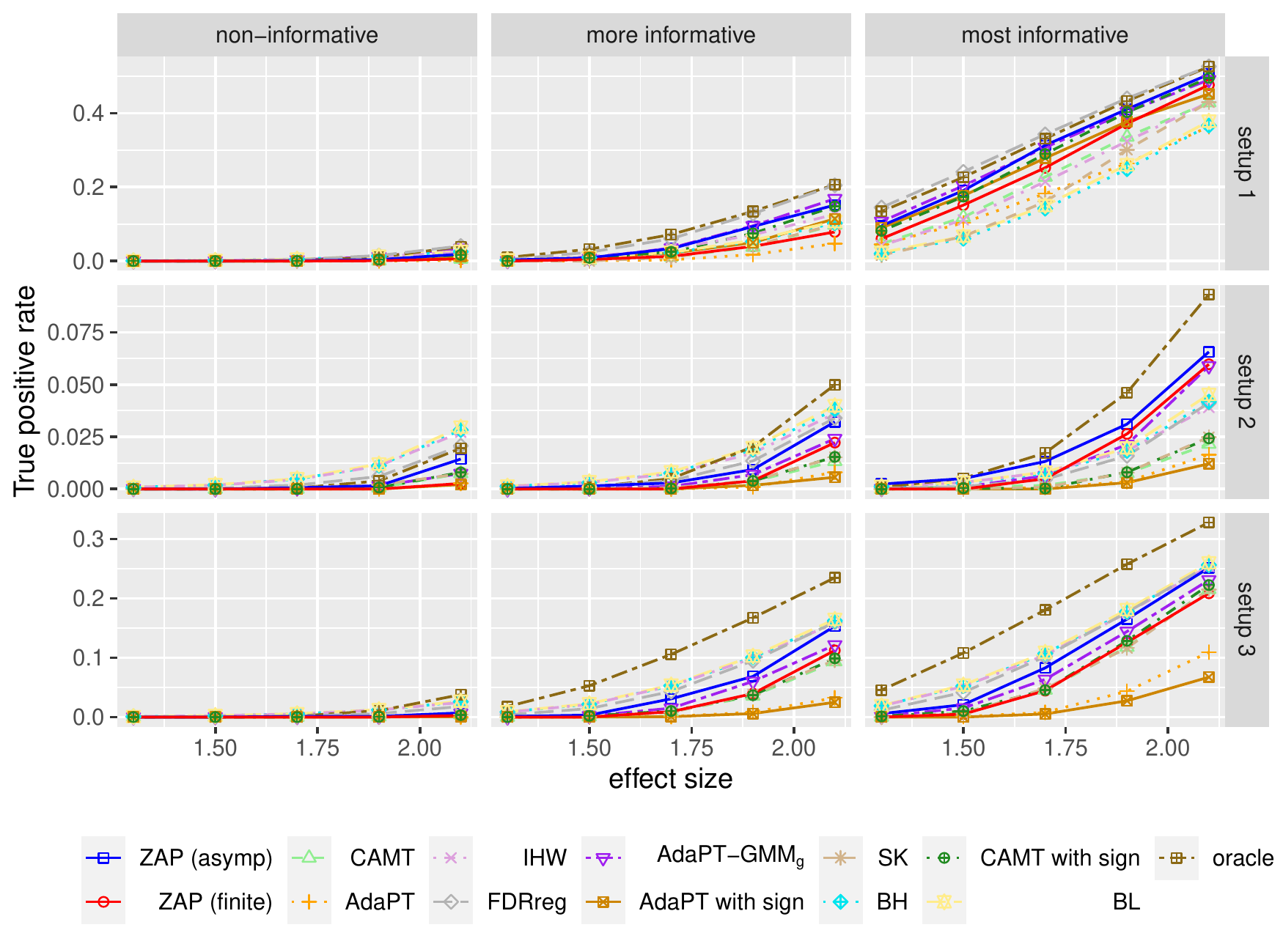}
\caption{
\emph{Sparsest setting}: FDR and TPR performances of the full list  of methods for the simulations in \secref{sim_data}, but with the density parameter $\eta$ set to $-3.5$ for Setups 1 and 3 and to $-4$ for Setup 2. All methods are applied at a targeted FDR level of $0.05$ (horizontal black lines). The x-axes and panels are labelled as in \figref{all_setup}.
 }
\label{fig:sparsest}
\end{figure}

 \subsection{Smaller-scale problem} \label{app:small_m_problem}

  \figref{all_methods_m_1000} shows the results of the same baseline setting in \appref{baseline_all_methods}, but with $m$ brought down to $1000$. The power of ZAP (asymp) remains  competitive, except in Setup 3 where the BH-type methods (BH, IHW, BL) again dominate in power.  The latter is again likely  attributable to the small potential number of hypotheses that can be rejected when $m = 1000$ and the conservative nature of the FDR estimates employed by the BC-type methods, as we have explained in \appref{sparse_signal}

 \begin{figure}[t]
\centering
\includegraphics[width=\textwidth]{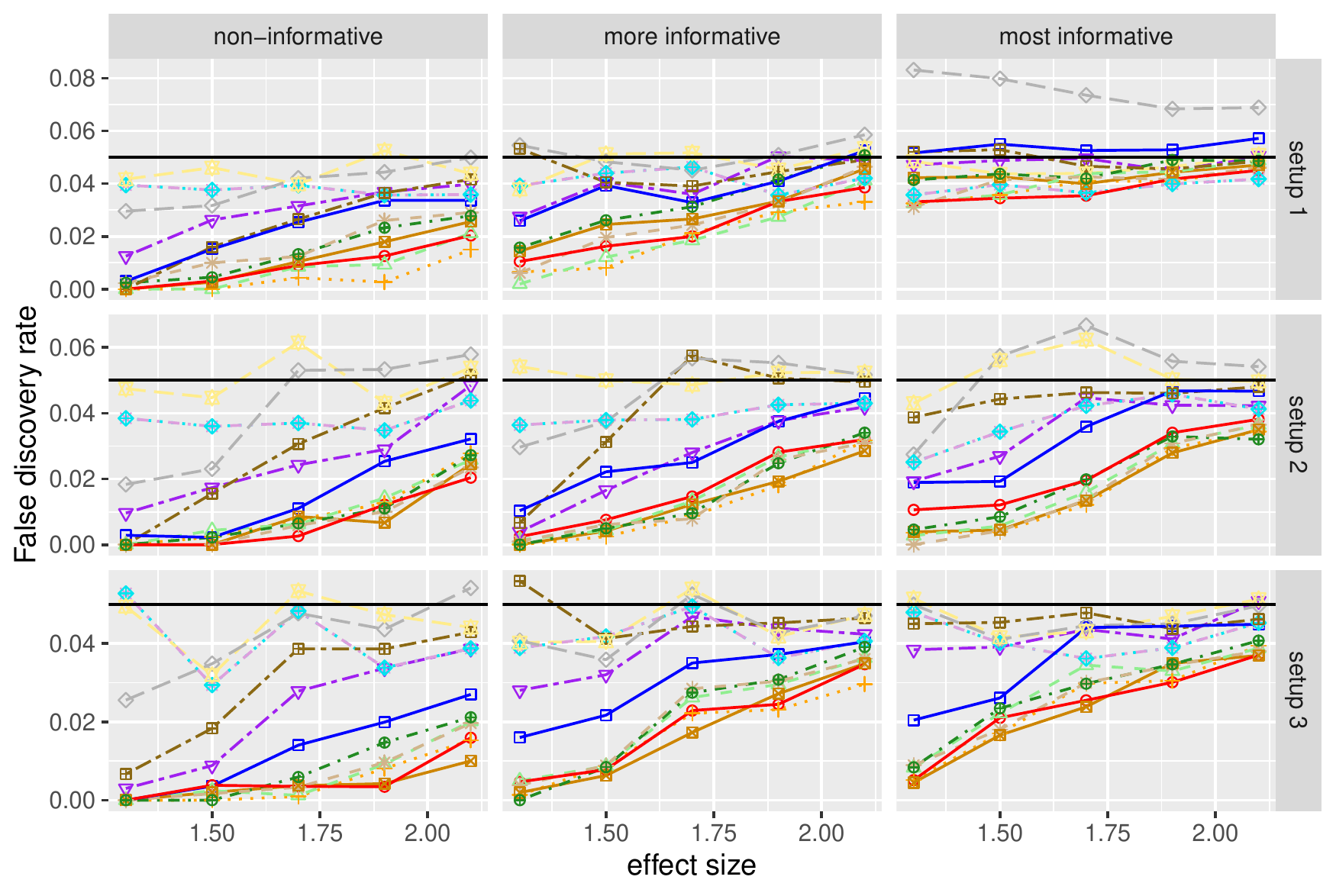}
\includegraphics[width=\textwidth]{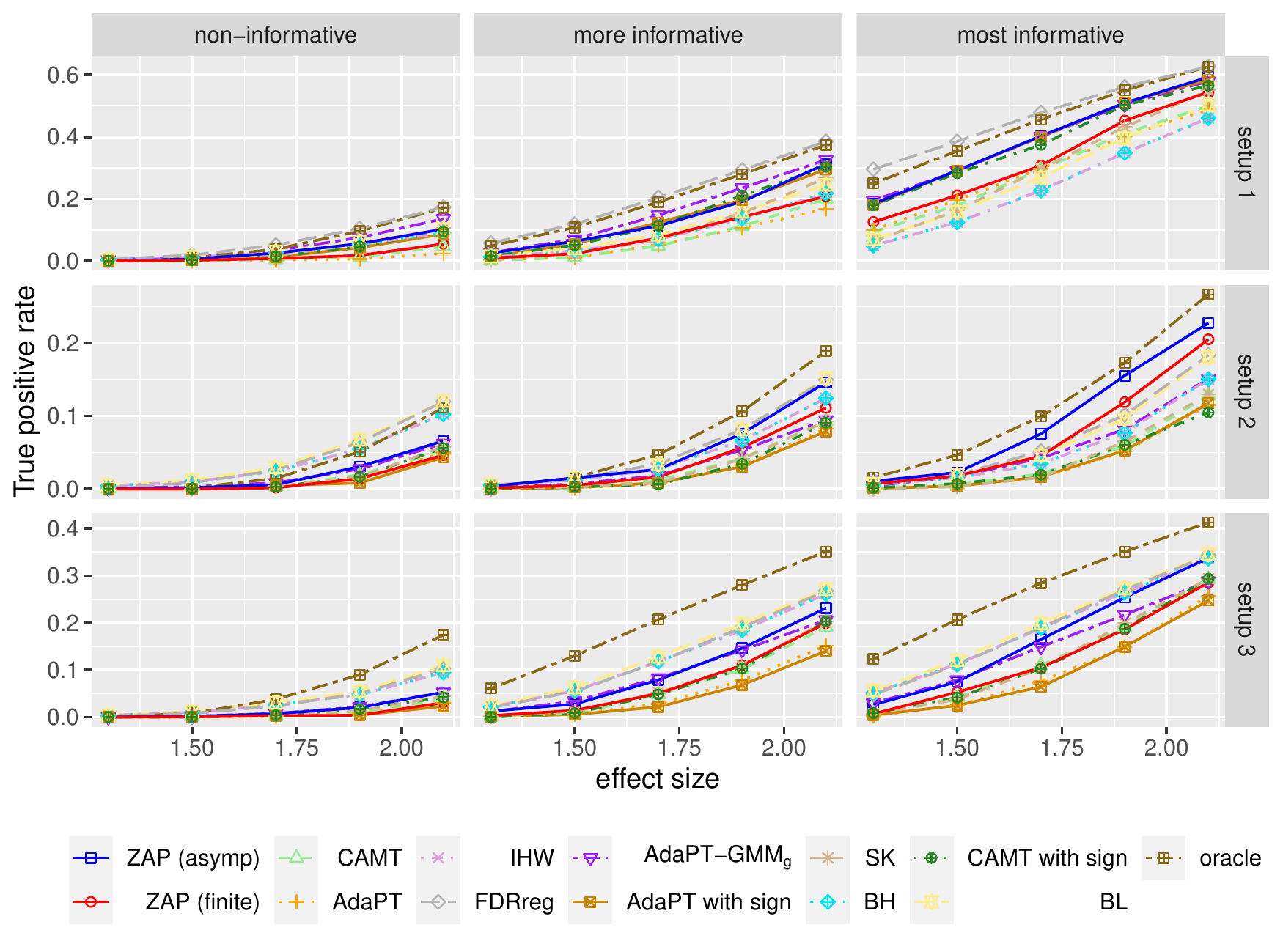}
\caption{
\emph{$m= 1000$}: FDR and TPR performances of the full list  of methods for the simulations in \secref{sim_data}, except that $m = 1000$. All methods are applied at a targeted FDR level of $0.05$ (the horizontal black lines). The x-axes and panels are labelled as in \figref{all_setup}.
 }
\label{fig:all_methods_m_1000}
\end{figure}


 \subsection{Global null} \label{app:global_null} 
 
  \figref{global_null} shows the empirical FDR of different methods at different target FDR levels, based on 300 simulated repetitions and $m = 5000$ null $z$-values, so any rejection constitutes a false discovery. Visually, except for IHW, BL and BH, no other methods are seen to violate the target FDR bound. 
As explained in the remark of \appref{adjusting_alt_var}, one can always get the empirical FDR of BH to be less than the target by running more repetitions. Our choice of $300$ repetitions is limited by the running time of other methods and our available  computing facilities.

  \begin{figure}[h]
\centering
\includegraphics[width=01\textwidth]{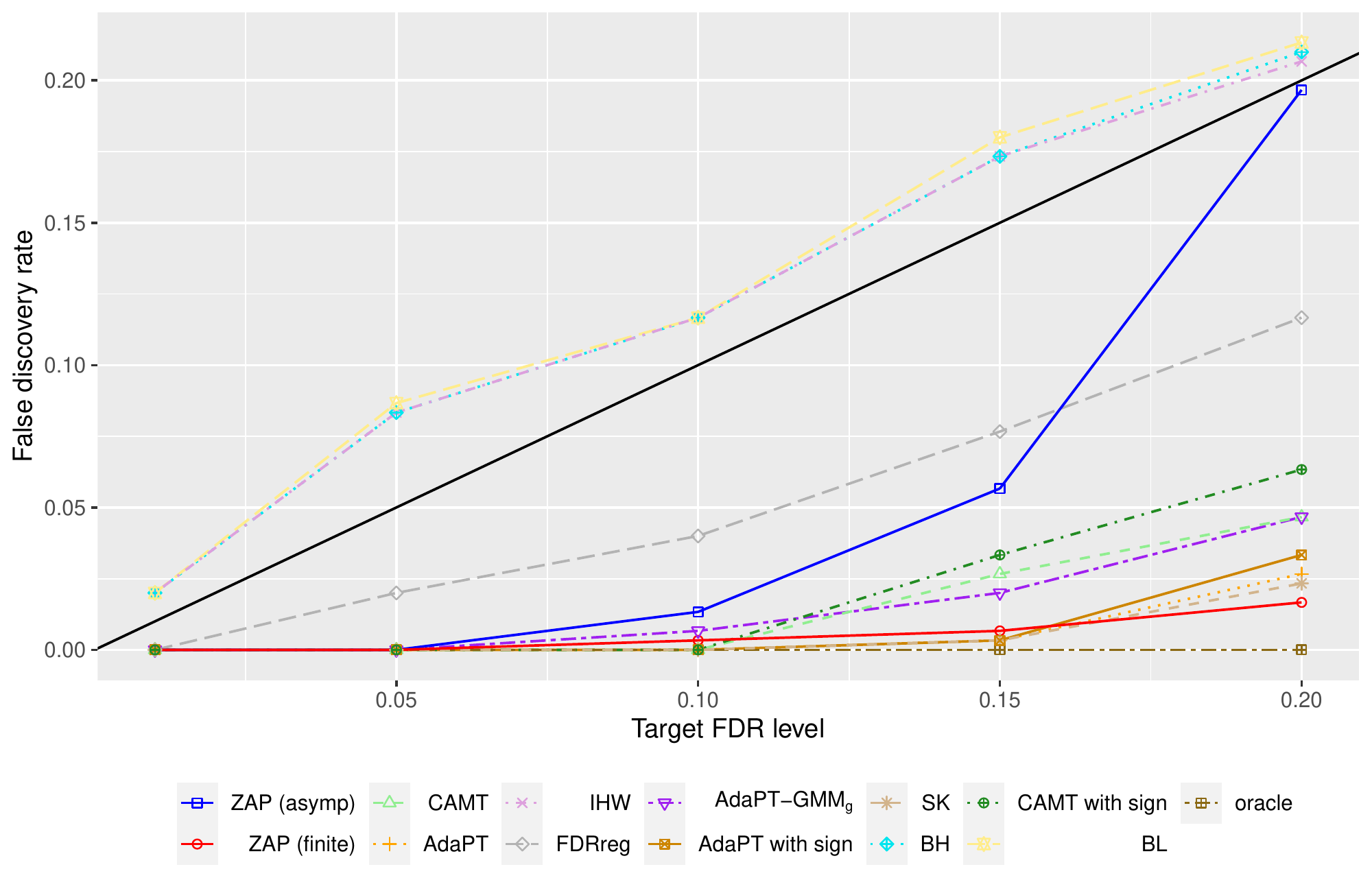}
\caption{
\emph{Global null}: The empirical FDRs of the full list of methods when applied to $m = 5000$ null $z$-values, based on $300$ repetitions. The $45^\circ$ line (in black)  is included as a reference.
 }
\label{fig:global_null}
\end{figure}

 \subsection{Comparison of computation costs} \label{app:speed}
 
\figref{comp_cost} compares the runtimes of different methods under Setup 1 for different values of $m$, with the simulation parameters described in \secref{sim_data} arbitrarily set at $(\epsilon, \eta, \zeta, \sigma) = (2.1, -2, 1, 1)$. While the speeds of different algorithms  depend on many factors including the complexities of the working models being fitted and the internal architectures of the packages implementing them, a fact that immediately stands out is the higher computational cost of the five data  masking algorithms (ZAP(finite), AdaPT-$\text{GMM}_g$, AdaPT, AdaPT with sign, SK). This is fully expected because all these algorithms have to iteratively fit a working model for  many times, and each fit is based on running an EM algorithm which is  computationally costly. It is conceivable that SK has the fastest speed among them, because its  model does not involve any covariates (and is hence easier to fit). As an aside, in a recent manuscript \citep{leung2022z}, a much more efficient data-masking algorithm for \emph{directional} FDR control is devised based on a working model whose fitting can be resorted to much faster interior point methods instead of the EM. 

The asymptotic ZAP  does use an EM algorithm  to fit the beta-mixture working model, but only once, and is seen to have shorter runtimes than all the data masking algorithms. It has longer runtimes than the rest of the methods, which is mainly due to its need for evaluations at many uniform realizations to obtain each mirror statistic $\hat{T}_i^{\mathfrak{m}}$   in \algref{zapAsymp}. To be more effective at such computations, we have employed parallel computing (based on the \texttt{parallel} package in R) so that these mirror statistics can be  estimated in parallel across $i$; this can be performed on most multi-core laptops available these days. 

Lastly, for both ZAP methods, we have leveraged the celebrated R package \texttt{SQUAREM}  to speed up each EM algorithm; see  \citet{varadhan2008simple} for details.

  \begin{figure}[h]
\centering
\includegraphics[width=01\textwidth]{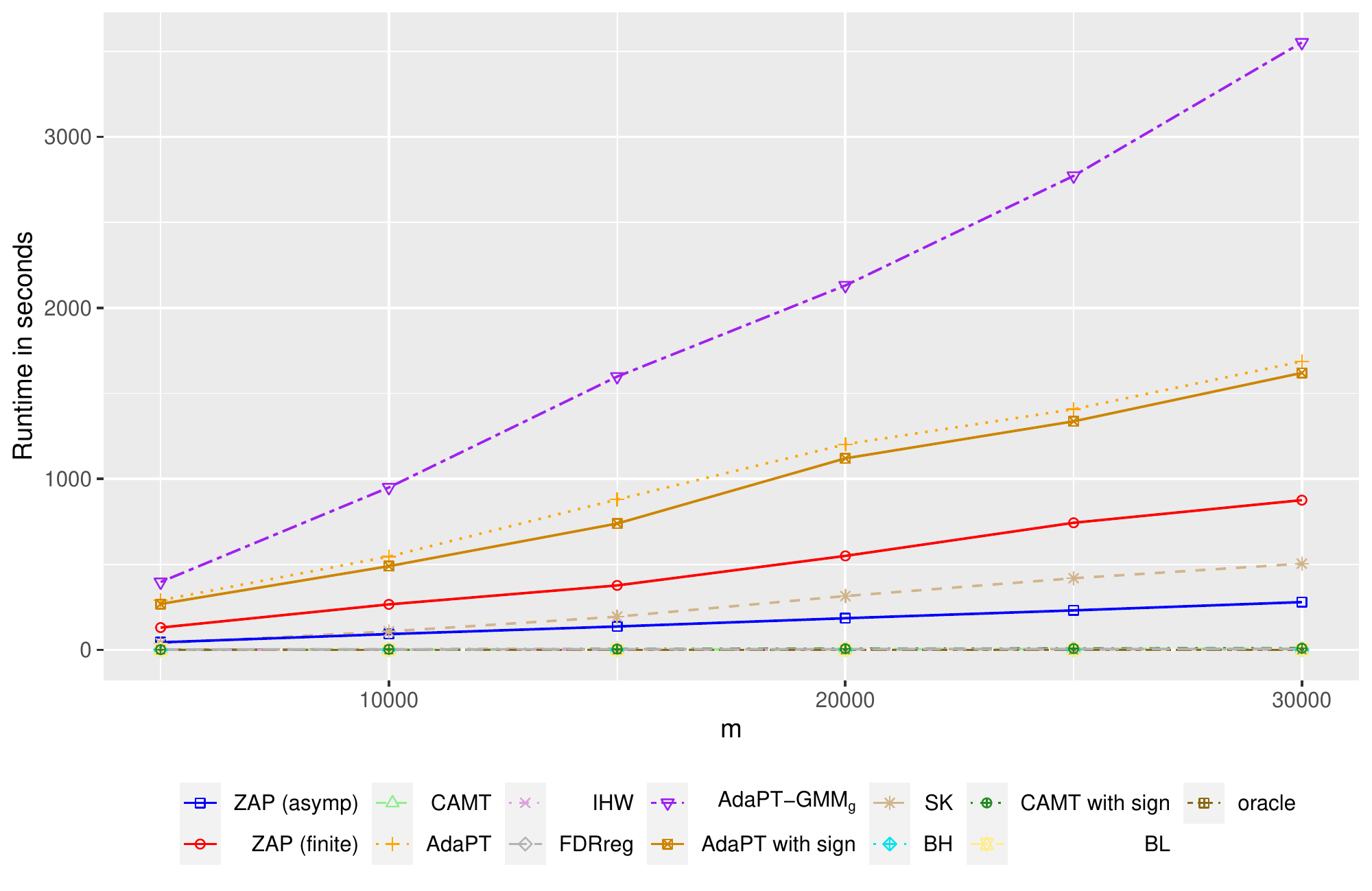}
\caption{
\emph{Runtime comparison}: The runtimes of all methods under Setup 1 for different values of $m$, with the simulation parameters set at  $(\epsilon, \eta, \zeta, \sigma) = (2.1, -2, 1, 1)$. The runtimes are recorded based on a single repetition of the simulations. 
 }
\label{fig:comp_cost}
\end{figure}


\section{Additional numerical results on real data} \label{app:real_data_append}

Inclusive of the additional methods introduced in \appref{more_simulation_results} (but exclusive of the oracle), the numbers of rejections for the full list of  methods  applied to the real data in \secref{real_data} are shown in  \figref{real_data_plot_append};  the conclusions we can arrive at are essentially the same as in the main text.   One can refer to \citet{lei2016adapt} for access to the \texttt{bottomly} and \texttt{airway} datasets. The other two real datasets are available at:

\begin{enumerate}
\item \texttt{hippo}: \url{https://www.raynamharris.com/DissociationTest/}
\item \texttt{scott}: \url{https://github.com/jgscott/FDRreg}  
\end{enumerate}

\subsection{Differential expression analysis of RNA-Seq data } \label{app:rnaseq_processing}

We  briefly discuss the importance of pre-filtering genes with excessively low read counts before applying FDR methodologies for differential expression (DE) analysis of RNA-Seq data. For the unfamiliar reader,  a good open resource on the relevant analysis pipeline can be found on \url{https://github.com/hbctraining/DGE_workshop}. It typically begins with a raw ``count matrix" with the expression read counts as entries, where each row corresponds to a mapped gene and each column corresponds to  a sample/library that is either in the treatment or the control group. This count matrix is taken as an input to a suite of statistical analysis tools  available from one of the \texttt{R} packages for DE analysis that differ by their underlying modelling assumptions, to produce test statistics that are re-scaled measures of differential expression between the two groups for all the genes involved. The two most popular such R pipelines which can produce the z-values considered by the current paper  are \texttt{limma} (with the ``voom" function therein) \citep{ritchie2015limma}  and \texttt{DESeq2} \citet{love2014moderated}.  \texttt{limma}  operates with a linear model to produce  $t$-statistics, and \texttt{DESeq2} operates with a negative binomial model to produce Wald statistics. These primary statistics can then undergo the further transformation in \secref{assessor} to give the $u$-values, on which our ZAP methods can be applied. 

However, without suitable pre-processing, a raw count matrix will typically produce  unusual distributions for the $u$-values (or $p$-values). \figref{real_data_histograms_unfiltered} plots histograms of the $u$-values produced by the original raw count matrices of the three RNA-Seq datasets in the main text processed with \texttt{DESeq2}. Normally, if the test statistics are well-calibrated, the null $u$-values should be approximately uniformly distributed, and one should only expect spikes  near the two ends of the interval $(0,1)$ (or only close to 0 if the histogram is for two-sided $p$-values) which represent genes that are non-null. This is clearly not the case in \figref{real_data_histograms_unfiltered}, and the spurious spikes in the middle of the unit interval for all three histograms are typically results of genes that have excessive low read counts for which reliable DE analysis is impossible and can at best be considered as nulls. In particular, the presence of such spikes will make a procedure like the BH overly conservative. A standard practice is to filter out these genes according to some rules of thumb which have been discussed by 
\citet{chen2016reads} in some length. In the analysis of the main text, we have adopted a simple convention of filtering out genes with a total raw counts less than $15$ using the function \texttt{filterByExpr} in the R package \texttt{edgeR}, which implements the method in \citet{chen2016reads}. Apparently, spurious structures in the $u$-value histograms have been more or less removed as a result, as is evident by comparing \figref{real_data_histograms_unfiltered} $(c)$ and \figref{real_data_plot}$(e)$, the latter of which has its $u$-values produced by the filtered version of \texttt{hippo} dataset.

 \begin{figure}[t]
\centering
\includegraphics[width=\textwidth]{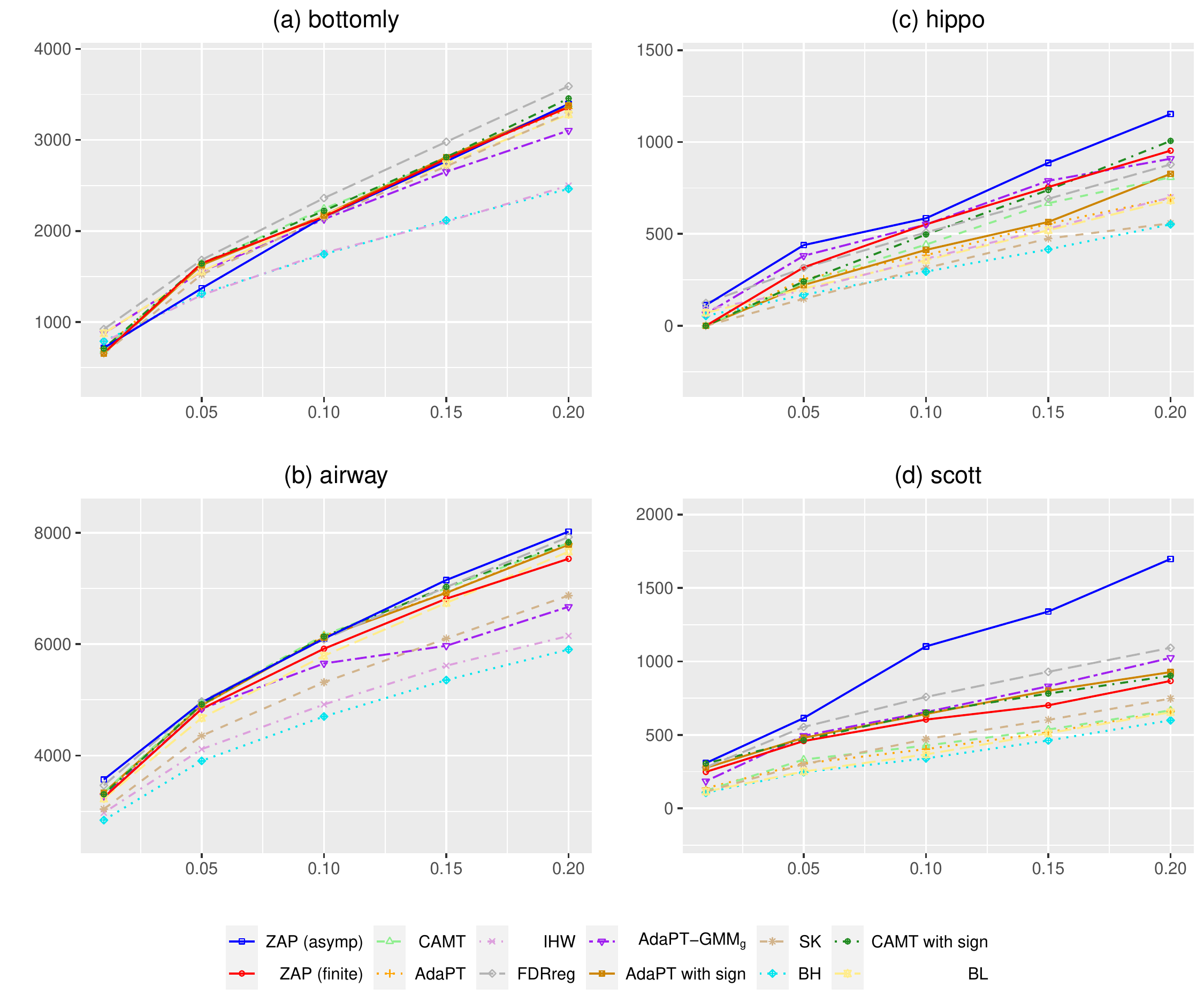}
\caption{(a)-(d) plot the numbers of rejections for different methods across datasets, against targeted FDR level at $0.01, 0.05, 0.1, 0.15, 0.2$.  }
\label{fig:real_data_plot_append}
\end{figure}

\begin{figure}[h]
\centering
\includegraphics[width=\textwidth]{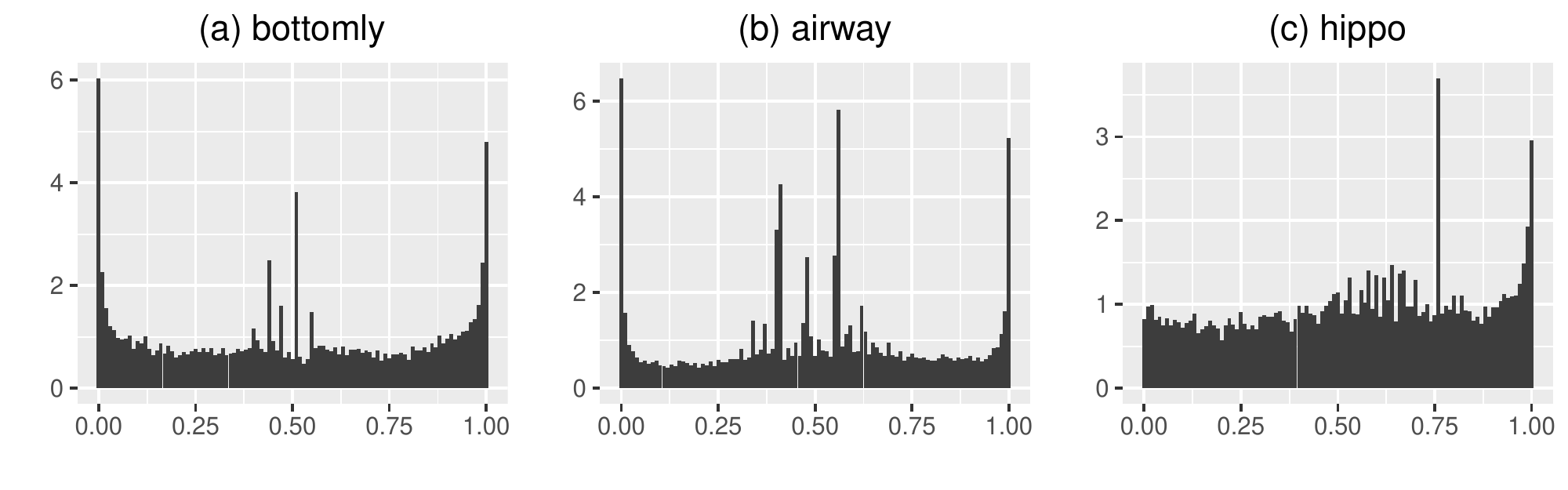}
\caption{Histograms of $u$-values for the original unfiltered versions of the three RNA-Seq datasets in the main text. In principle, one should only see at most two ``spikes" on the two ends of the interval $(0, 1)$. Spikes not located close to 0 or 1 in any histogram result from genes with extremely low read counts. 
 }
\label{fig:real_data_histograms_unfiltered}
\end{figure}

\end{document}